%% file: 21_07_09_RevisedVersionForDMJ.tex
\documentclass[reqno]{article}

\usepackage{amsthm}
\usepackage{amsmath}
\usepackage{amssymb}
\usepackage{graphicx}
\usepackage{color}
\usepackage{bbm}
\usepackage[top=2.5cm, bottom=2.5cm, left=2.5cm, right=2.5cm]{geometry}
\usepackage{enumerate}
\usepackage{faktor}
\usepackage{nicefrac}
\usepackage[nottoc]{tocbibind}
\usepackage{slashed}
\usepackage{authblk}
\usepackage[normalem]{ulem}

\usepackage[scr=euler]{mathalfa}

\def\mathclap#1{\text{\hbox to 0pt{\hss$\mathsurround=0pt#1$\hss}}}

\newcommand{\N}{\mathbb{N}}
\newcommand{\R}{\mathbb{R}}

\newcommand{\CH}{\mathcal{CH}^+}
\newcommand{\Hp}{\mathcal{H}^+}

\newcommand{\vol}{\mathrm{vol}}

\newcommand{\ux}{\underline{x}}
\newcommand{\loc}{\mathrm{loc}}
\newcommand{\Reps}{R_{\varepsilon_0, \varepsilon_1}}
\newcommand{\mg}{\mathring{\gamma}}
\newcommand{\uh}{u_{\mathcal{H}^+}}
\newcommand{\Mext}{M_{\mathrm{ext}}}

\newcommand{\rd}{\partial}

\begin{document}

\numberwithin{equation}{section}
\newtheorem{theorem}[equation]{Theorem}
\newtheorem{remark}[equation]{Remark}
\newtheorem{assumption}[equation]{Assumption}
\newtheorem{claim}[equation]{Claim}
\newtheorem{lemma}[equation]{Lemma}
\newtheorem{definition}[equation]{Definition}
\newtheorem{corollary}[equation]{Corollary}
\newtheorem{proposition}[equation]{Proposition}
\newtheorem*{theorem*}{Theorem}
\newtheorem{conjecture}[equation]{Conjecture}
\newtheorem{example}[equation]{Example}

\setcounter{tocdepth}{3}

\title{On holonomy singularities in general relativity \& the $C^{0,1}_{\mathrm{loc}}$-inextendibility of spacetimes}
\author{Jan Sbierski\thanks{Mathematical Institute, 
University of Oxford,
Radcliffe Observatory Quarter,
Woodstock Road, 
Oxford, 
OX2 6GG,
United Kingdom}}
\date{\today}

\maketitle

\begin{abstract}
This paper investigates the structure of gravitational singularities at the level of the connection. We show in particular that for FLRW spacetimes with particle horizons a local holonomy, which is related to a gravitational energy, becomes unbounded near the big-bang singularity. This implies the $C^{0,1}_{\loc}$-inextendibility of such FLRW spacetimes.
Again using an unbounded local holonomy we also give a general theorem establishing the $C^{0,1}_{\loc}$-inextendibility of spherically symmetric weak null singularities which arise at the Cauchy horizon in the interior of black holes. Our theorem does not presuppose the mass-inflation scenario and in particular applies to the Reissner-Nordstr\"om-Vaidya spacetimes as well as to spacetimes which arise from small and generic spherically symmetric perturbations of two-ended subextremal Reissner-Nordstr\"om initial data for the Einstein-Maxwell-scalar field system. In \cite{LukOh19I}, \cite{LukOh19II} Luk and Oh proved the $C^2$-formulation of strong cosmic censorship for this latter class of spacetimes -- and based on their work we improve this to a $C^{0,1}_{\loc}$-formulation of strong cosmic censorship.
\end{abstract}

\tableofcontents

\section{Introduction}

Let $(M,g)$ be a smooth and time-oriented Lorentzian manifold. A $C^{0,1}_{\loc}$-extension of $(M,g)$ is an isometric embedding $\iota : M \hookrightarrow \tilde{M}$ of $M$ into a Lorentzian manifold $(\tilde{M}, \tilde{g})$ of the same dimension, where $\tilde{g}$ is a locally Lipschitz regular metric, such that $\partial \iota(M) \subseteq \tilde{M}$ is non-empty. If no such $C^{0,1}_{\loc}$-extension exists then we say that $(M,g)$ is $C^{0,1}_{\loc}$-inextendible. Extensions of other regularities are defined analogously. 

Our first main result is that cosmological warped product spacetimes with particle horizons are $C^{0,1}_{\loc}$-inextendible. This class in particular contains the FLRW spacetimes with particle horizons.

\begin{theorem} \label{ThmInt1}
Let $(\overline{M}, \overline{g})$ be a $3$-dimensional complete Riemannian manifold and let $a : (0, \infty) \to (0, \infty)$ be a smooth function satisfying \begin{equation*} \begin{aligned}&\lim_{ t \to 0} a(t) = 0 \\
&\int_0^1 \frac{1}{a(t)} \, dt < \infty \\
&\int_1^\infty \frac{a(t)}{ \sqrt{a(t)^2 + 1}}\, dt = \infty \;. \end{aligned} \end{equation*} Let $M = (0, \infty) \times \overline{M}$ and consider the Lorentzian metric $g = -dt^2 + a(t)^2 \, \overline{g}$ on $M$. Then $(M,g)$ is $C^{0,1}_{\loc}$-inextendible.
\end{theorem}

Our second main result is that spherically symmetric weak null singularities are $C^{0,1}_{\loc}$-inextendible. A version of our results can be stated as follows:

\begin{theorem} \label{ThmInt2}
Let $M= (-\infty, 0) \times (- \infty, 0) \times \mathbb{S}^2$ with standard $(u,v)$-coordinates on the first two factors and let $g = -\frac{\Omega^2(u,v)}{2} (du \otimes dv + dv \otimes du) + r^2(u,v) \, \mathring{\gamma}$, where $\mathring{\gamma}$ is the standard metric on $\mathbb{S}^2$ and $\Omega, r : (-\infty, 0) \times (-\infty, 0) \to (0, \infty)$ are smooth positive functions. Fix a time-orientation on $(M,g)$ by stipulating that $\rd_u + \rd_v$ is future directed timelike. Assume that 
\begin{equation*}
\begin{aligned}
&\bullet \Omega \textnormal{ \emph{and }} r \textnormal{\emph{ extend continuously as positive functions to }} (-\infty, 0] \times (-\infty, 0] \\
&\bullet \lim_{v \to 0} \rd_v r (u,v) = - \infty \textnormal{\emph{ for all }} u \in (-\infty, 0) \\
&\bullet \lim_{u \to 0} \rd_u r(u,v) = -\infty \textnormal{ \emph{for all} } v \in (-\infty, 0) \\
&\bullet \textnormal{\emph{ for all $u \in (-\infty, 0)$ there exists $v_0(u) \in (-\infty, 0)$ such that $\rd_u r(u,v) <0$ for all $v \geq v_0(u)$} }\\
&\bullet \textnormal{\emph{ for all $v \in (-\infty, 0)$ there exists $u_0(v) \in (-\infty, 0)$ such that $\rd_v r(u,v) <0$ for all $u \geq u_0(v)$}} \;.
\end{aligned}
\end{equation*}
Then $(M,g)$ is \underline{future $C^{0,1}_{\loc}$-inextendible}, i.e., there does not exist a $C^{0,1}_{\loc}$-extension $\iota : M \hookrightarrow \tilde{M}$ with the property that there is a future directed future inextendible timelike curve in $M$ which has a future limit point in $\tilde{M}$.
\end{theorem}
Note that the spherically symmetric weak null singularities do admit future $C^0$-extensions!

\subsection{Motivation}

There are three main motivations for studying the low-regularity inextendibility of Lorentzian manifolds. 

\textbf{The first motivation} comes from the expectation/possibility that if a solution to Einstein's equations $R_{\mu \nu} - \frac{1}{2} g_{\mu \nu} R = 2 T_{\mu \nu}$ can be continued as a weak solution, then classical time-evolution in general relativity continues, i.e., the classical theory does not break down.  Here, it is of course crucial to discuss what we mean by a \emph{weak solution}.

Considering for simplicity of discussion the vacuum Einstein equations $R_{\mu \nu} - \frac{1}{2} g_{\mu \nu} R = 0$, we see that they take the schematic form $g \rd \rd g + N(g)(\rd g, \rd g) = 0$, where $N(g)(\rd g, \rd g)$ is a nonlinearity that is quadratic in $\rd g$ with coefficients depending on $g$. If the metric is in $C^2$ then clearly the strong, pointwise notion of a solution to Einstein's equations is available. If only $g \in C^0$ and $\rd g\in L^2_{\loc}$, then this regularity is still \emph{sufficient} to define the classical notion of a weak solution, i.e., we require that for all smooth and compactly supported vector fields $X, Y$ on $M$ we have 
\begin{equation}\label{EqWS}
0\overset{!}{=} \int_M \big(R(X,Y) - \frac{1}{2}g(X,Y) \cdot R\big) \vol_g = \int_M \big( g \rd \rd g + N(g)(\rd g, \rd g)\big) X Y \underbrace{\sqrt{-\det g} \, dx}_{= \vol_g} \;.
\end{equation} 
After one integration by parts to move one of the derivatives from the term $g \rd \rd g$ over to the test fields $X,Y$, we see that the regularity assumption $g \in C^0$ and $\partial g \in L^2_{\loc}$ is \emph{sufficient} to ensure that the coefficients of the test fields in \eqref{EqWS} are in $L^1_{\loc}$ -- and thus the Einstein tensor $R_{\mu \nu} - \frac{1}{2} g_{\mu \nu} R$ is a well-defined distribution. 

The regularity class $g \in C^0$ and $\partial g \in L^2_{\loc}$ has been widely (e.g.\ \cite{HawkEllis}, \cite{GerTra87}, \cite{Chrus91}, \cite{Chris09}) considered to be the largest regularity class that admits weak solutions in the above sense \emph{and} still exhibits sufficient stability properties, cf.\ the approximation theorems in \cite{GerTra87}.\footnote{In some references $g \in C^0$ is weakend to $g, g^{-1} \in L^\infty_{\loc}$, cf.\ \cite{GerTra87}.  Although this seems reasonable from a PDE point of view, it clearly poses challenges for the physical interpretation and Lorentzian causality theory alike. Whether they can be overcome seems to be largely an open question. 

Let us also mention that the  regularity class $g \in C^0$ and $\partial g \in L^2_{\loc}$ is not \emph{necessary} for the notion \eqref{EqWS} of a weak solution: one can for example slightly widen it by only requiring that the products of Christoffel symbols which appear in the Einstein equations are locally integrable, \cite{Gar99}, \cite{GerTra87}, but this has the disadvantage that one loses the good stability properties.}  Showing that a particular solution of Einstein's equations is inextendible as a Lorentzian manifold in the above regularity class of course implies that it is also inextendible as a weak solution in the above regularity class. Assuming now that $g \in C^0$ and $\partial g \in L^2_{\loc}$ is indeed the roughest regularity class for which a physical notion of weak solution exists\footnote{It is conceivable that even weaker physical notions of weak solutions are emerging or establishing themselves.} it follows that classical time-evolution terminates and that one has reached indeed the maximal extent of the classical solution. This is the first motivation for the study of inextendibility of Lorentzian manifolds with $g \in C^0$ and $\rd g \in L^2_{\loc}$. 
The  $C^{0,1}_{\loc}$-inextendibility results established in this paper are only one step in this direction. They should be seen as implying that the solution under consideration cannot be continued as a weak solution that is locally Lipschitz regular.

\textbf{The second motivation}, which is closely related to the first one, comes from the \emph{strong cosmic censorship conjecture}, which roughly states that general relativity is \emph{generically} a deterministic theory. A modern mathematical formulation is the following: 
\begin{equation} \label{EqSCC}
\parbox{0.8\textwidth}{The maximal globally hyperbolic development (MGHD) of \emph{generic} compact or asymptotically flat initial data is inextendible as a weak solution to the Einstein equations.}
\end{equation} 
The MGHD is the maximal development of the initial data which a priori can be guaranteed to be uniquely determined by the initial data. Any extension  as a weak solution thereof is in general non-unique. It is in this way that conjecture \eqref{EqSCC} implies that general relativity is generically deterministic. It follows from the earlier discussion that under the assumption that $g \in C^0$ and $\rd g \in L^2_{\loc}$ is the roughest regularity class for which a physical notion of weak solution exists, \eqref{EqSCC} is implied by 
\begin{equation*}
\parbox{0.8\textwidth}{The maximal globally hyperbolic development of \emph{generic} compact or asymptotically flat initial data is inextendible as a  Lorentzian manifold with $g \in C^0$ and $\rd g \in L^2_{\loc}$.}
\end{equation*} 
For further discussion of the strong cosmic censorship conjecture, references, and historical  background we refer the reader to the introduction of \cite{DafLuk17}.

\textbf{The third motivation} stems from the investigation of the physical and geometric structure and strength of spacetime singularities. According to relativistic point mechanics one can crudely identify the spacetime curvature with the tidal forces, the Levi-Civita connection with the gravito-inertial structure, and the spacetime metric with spatial and temporal distances. The statement that a singular spacetime is $C^{1,1}_{\loc}$-inextendible then corresponds to the statement that tidal forces blow up, a singularity which is $C^{0,1}_{\loc}$-inextendible indicates the breakdown of the gravito-inertial structure, and a $C^0$-inextendibility result shows that the notion of space and time cannot be continued. Traditionally the investigation of the different structures of spacetime singularities was mainly at the level of curvature and its effect on Jacobi fields along incomplete timelike geodesics, cf.\  for example \cite{Tip77}, \cite{Ori00}. Investigation at low-regularity level gives new information: in \cite{Sbie15} it was shown that space itself is torn apart at the singularity inside a Schwarzschild black hole while a result of this paper is for example that there is no unique standard of finite inertial energy near a cosmological FLRW singularity with particle horizons. Indeed, in a sense made precise below the local gravitational energy diverges. The gravito-inertial structure of the Schwarzschild singularity is of similar nature.

\subsection{Earlier works on low-regularity inextendibility}

The first study of low-regularity inextendibility of Lorentzian manifolds was carried out in \cite{Sbie15}, where the $C^0$-inextendibility of the Minkowski spacetime and of the maximal analytic Schwarzschild spacetime was shown, see also \cite{Sbie18}. A conditional $C^0$-inextendibility criterion for \emph{expanding singularities} was given by Chru\'sciel and Klinger in \cite{ChrusKli12}.  In \cite{GalLin16} Galloway and Ling showed the $C^0$-inextendibility of the AdS spacetime and presented $C^0$-extendibility results for a class of hyperbolic FLRW spacetimes which they dubbed \emph{Milne-like}. The latter results have been extended by Ling in \cite{Ling20}. In \cite{GalLinSbi17} it was shown that globally hyperbolic and timelike geodesically complete spacetimes are $C^0$-inextendible, a result which was improved in various directions first by Graf and Ling in \cite{GrafLing18} and finally by Minguzzi and Suhr in \cite{MinSuhr19}. For an inextendibility result due to timelike geodesic completeness in Lorentzian length spaces by Grant, Kunzinger, and S\"amann see \cite{GraKuSa19}.

\subsection{Outline and discussion of main results}

All of the previous studies of low-regularity inextendibility results captured geometric obstructions at the $C^0$-level. However, there are physical singularities which are $C^{1,1}_{\loc}$-inextendible due to blow-up of curvature but at the same time do admit continuous extensions. Perhaps most notably let us mention here the weak null singularities inside charged or rotating black holes\footnote{For an overview of results and further references see the introduction of \cite{LukOh19I}.}. The singular behaviour happens there not already at the level of the metric but only at the level of the connection.  This paper widens the geometric investigation of gravitational singularities to the level of the connection. We show how local holonomy transformations can be used to infer its blow-up. 

Curvature and holonomy are of course intimately related. However, intuitively speaking, it depends on the `rate of the blow-up' of curvature near a singularity whether the connection will also blow up or not\footnote{A simple illustration is provided by the well-known exact impulsive gravitational wave spacetimes \cite{Pen72} where curvature has a delta singularity across a hypersurface but the Christoffel symbols are uniformly bounded.}. One can think of holonomy as a geometric way to integrate curvature.
\newline

We proceed by giving a rough sketch of the proof of Theorem \ref{ThmInt1} in Section \ref{SecCos} which outlines the main idea of the argument. All the inextendibility proofs in this paper are by contradiction and have the following first step in common: assuming that there is  a $C^{0,1}_{\loc}$-extension $\iota : M \hookrightarrow \tilde{M}$ of $(M,g)$ one can find a  timelike geodesic $\gamma$ in $M$ which has a limit point $\tilde{q}$ in $\tilde{M} \setminus \iota(M)$, i.e., it leaves $M$. Moreover one can find a small chart $\tilde{U} \subseteq \tilde{M}$ around $\tilde{q}$  in which we have $C^{0,1}_{\loc}$-control over the metric. In particular, in those coordinates the Christoffel symbols of $g$ in $\tilde{U} \cap \iota(M)$ are uniformly bounded. Apart from the $C^{0,1}_{\loc}$-control of the metric these results are valid also for $C^0$-extensions and are found in \cite{Sbie15}, \cite{GalLinSbi17},  and \cite{Sbie18}. They are here summarised in Section \ref{SecFundResults}. 

In Section \ref{SecLemBound} we prove a small lemma showing that if we have uniform bounds on the metric and its derivatives in some coordinate system as above, and if a curve is \emph{local} in the sense that its coordinate velocity is uniformly bounded and also its domain of definition, then the parallel transport map along this curve is also uniformly bounded. This will be used to put upper bounds on holonomy transformations along curves in $\tilde{U} \cap \iota(M)$ which are local in the above sense.

The third condition on the scale factor $a(t)$ in Theorem \ref{ThmInt1} guarantees that all future inextendible timelike geodesics in $M$ are future complete, and thus one can show that $\tilde{q}$ cannot be a future endpoint (\cite{GalLinSbi17}). So it remains to show that   $\tilde{q}$ can neither be a past endpoint of $\gamma$. Here, the structure of the big-bang singularity enters. One first shows that the projection of geodesics in $M = (0, \infty) \times \overline{M}$  to $\overline{M}$ are still geodesics, though not necessarily affinely parametrised. The second condition on the scale factor $a(t)$ in Theorem \ref{ThmInt1} expresses that particle horizons are present and thus the projection of $\gamma$ to $\overline{M}$ has a limit point in $\overline{M}$. Choosing \emph{polar} normal coordinates $(\chi, \theta, \varphi)$ around this limit point, in which the projected geodesic is a radial one, the timelike geodesic $\gamma$ completely lies in the totally geodesic submanifold spanned by the $t$ and the radial $\chi$ coordinate. Because the submanifold is totally geodesic parallel transport tangential to it \underline{in $M$} can be computed intrinsically in the submanifold. Thus we can reduce the problem in the following to a two-dimensional one with metric $g|_{\{t, \chi\}} = -dt^2 + a(t)^2 d\chi^2$. Choosing now a base-point $\gamma(-\mu)$ on $\gamma$  we construct a family of loops as in Figure \ref{Fig1}. Here, $\hat{t} (t) = \int_0^t \frac{1}{a(t')} \, dt'$ is such that the metric in $(\hat{t}, \chi)$-coordinates becomes conformally flat, so that the diagram in Figure \ref{Fig1} is a Penrose diagram.
\begin{figure}[h]
\centering
\begin{minipage}{.33\textwidth}
  \centering
 \def\svgwidth{5cm}
   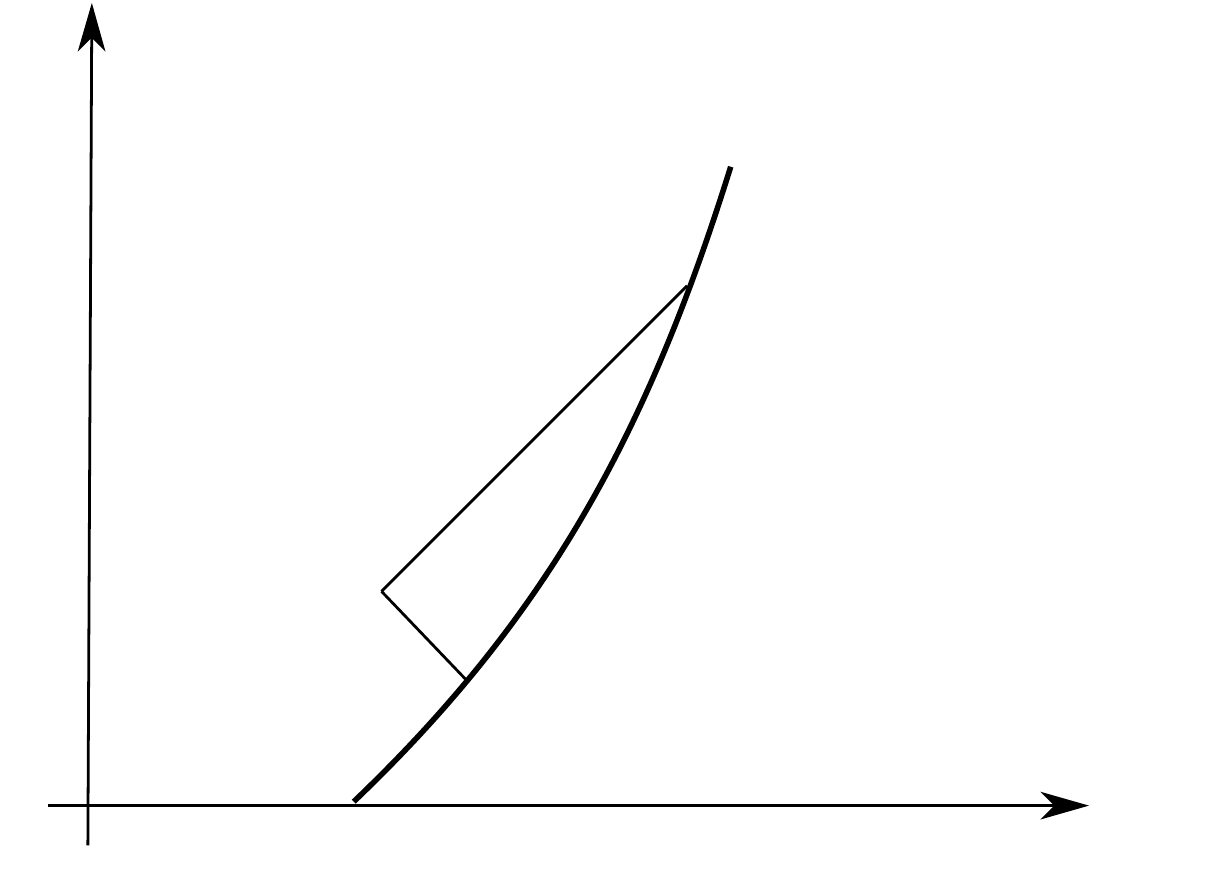
      \caption{Family of loops} \label{Fig1}
\end{minipage}%
\begin{minipage}{.33\textwidth}
  \centering
 \def\svgwidth{5cm}
   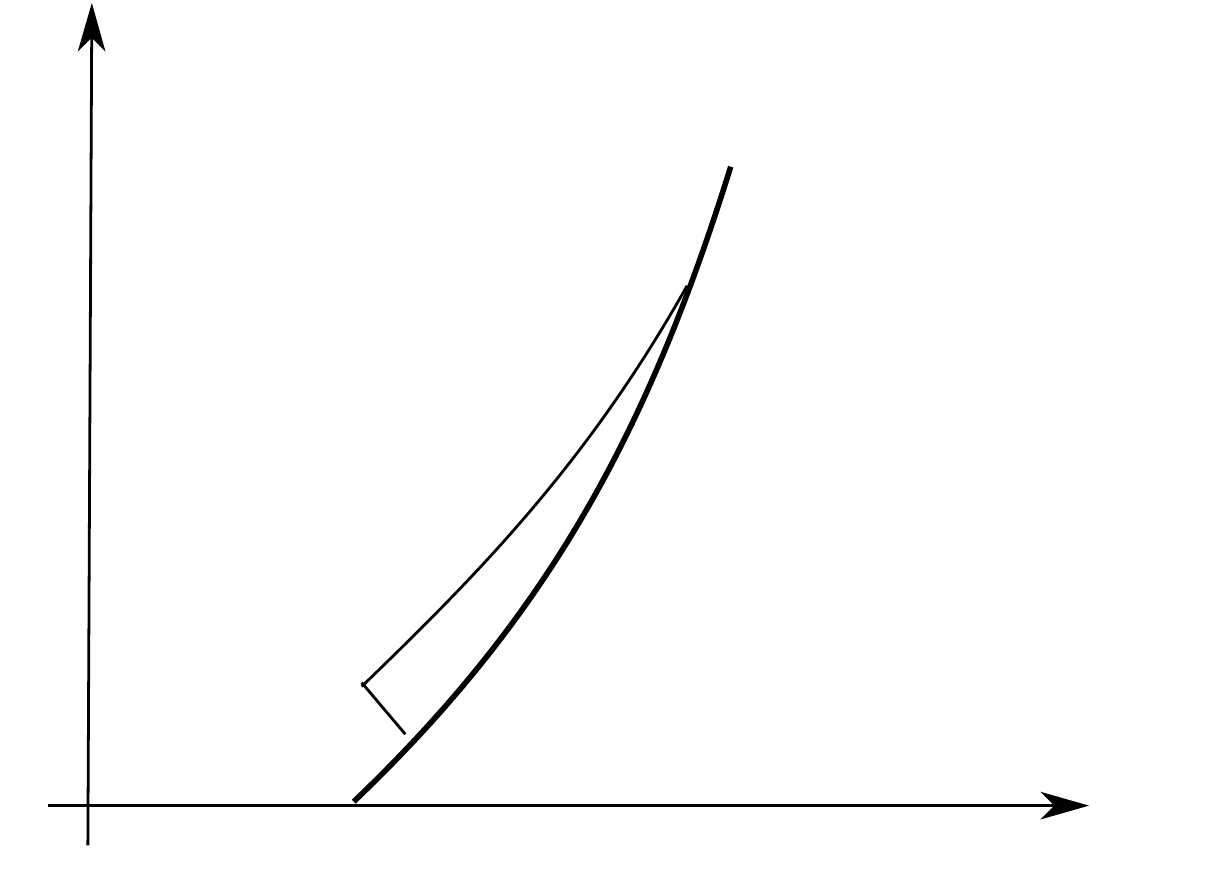
      \caption{Energy of radiation} \label{Fig2}
\end{minipage}%
\begin{minipage}{.33\textwidth}
  \centering
  \def\svgwidth{5cm}
  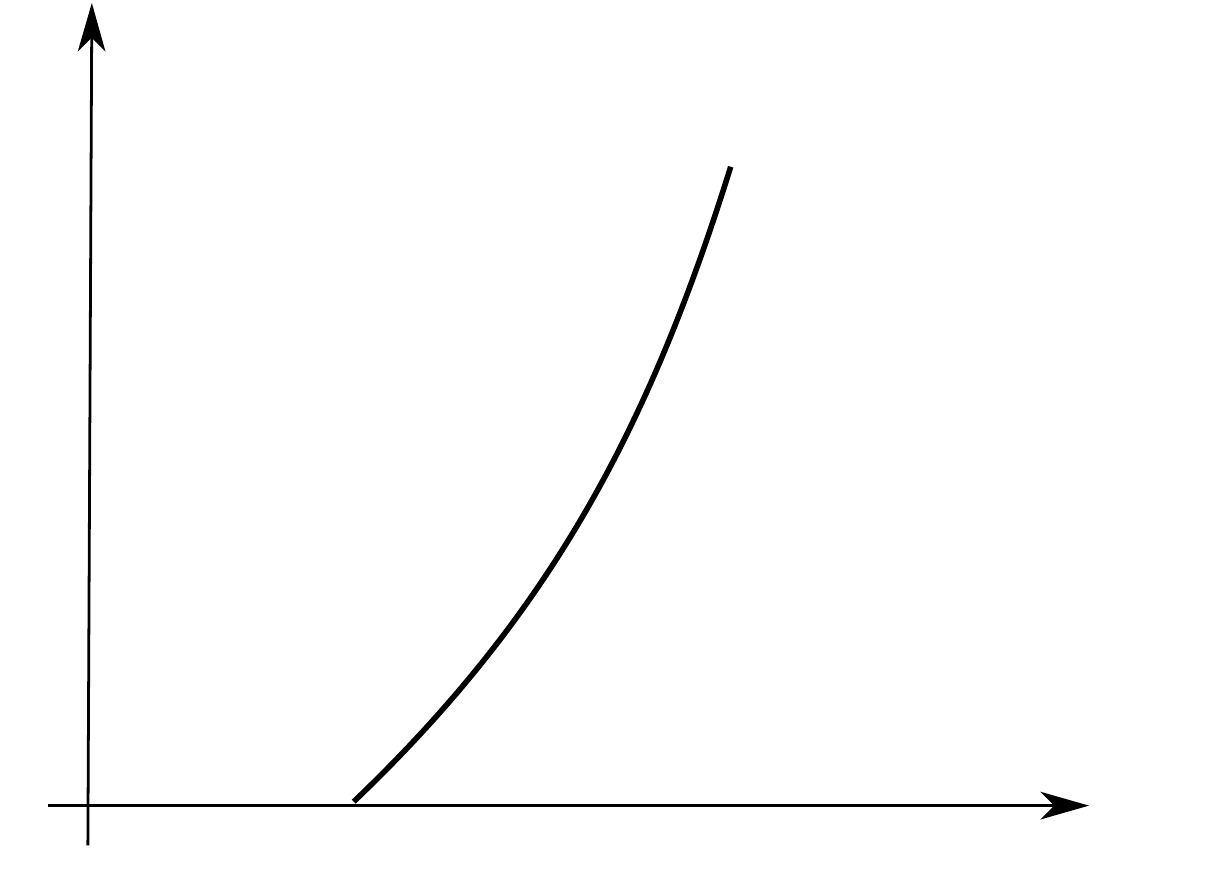
     \caption{Kinetic energy of  $\sigma$} \label{Fig3}
\end{minipage}
\end{figure}
A loop in this family consists first of an outgoing null geodesic segment $\ell$, then concatenated with an ingoing one $\underline{\ell}$ which brings us back to $\gamma$ at a point $\gamma(-\mu_0)$, and then we go back along $\gamma$ to $\gamma(-\mu)$. The family of loops is indexed by the point $\gamma(-\mu_0)$ of return to $\gamma$. One can now show that the holonomy transformations at the base point $\gamma(-\mu)$ associated with this family of loops form an unbounded family of boosts.

Using a homotopy argument and the causal nature of the loops one can show that all the loops in this family are \emph{local} curves in $\tilde{U} \cap \iota(M)$ in the sense discussed above. This gives the contradiction that the corresponding holonomy transformations should form a uniformly bounded subset of the Lorentz group.
\newline

We now briefly discuss the physical implication and interpretation of the unbounded holonomy transformations as above near the big-bang singularity. For this it is helpful to reverse the time-orientation for the time being and consider the inertial observer $\gamma$ \emph{approaching} the singularity. The inertial observer $\gamma$ carries with himself his local orthonormal reference frame $\{f_0 = \dot{\gamma}, f_1, f_2, f_3\}$, where we can assume that $f_1$ lies in the $\{t, \chi\}$-plane. By definition the local reference frame is parallel along $\gamma$. An equivalent statement to the unboundedness of the holonomy transformations with base point $\gamma(-\mu)$ is thus that the parallel transport map along the null segments $\ell$ and $\underline{\ell}$ of the loop from $\gamma(-\mu)$ to $\gamma(-\mu_0)$ is unbounded in the basis of the parallel frame $(f_0, f_1, f_2, f_3)$ along $\gamma$ as $\gamma(-\mu_0)$ approaches the singularity at $\{\hat{t} = 0\}$. More precisely one can show that the parallel transport of the null vector $-\rd_{\hat{t}} + \rd_\chi$ (which is tangent to $\underline{\ell}$) along the null segments $\ell$ and $\underline{\ell}$ becomes unbounded in the reference frame $(f_0, f_1, f_2, f_3)$ along $\gamma$. This of course is a manifestation of the well-known red-shift effect in cosmology which turns into a blue-shift if we reverse the time-orientation. The parallel transport in regions bounded away from $\{\hat{t} = 0\}$ is of course uniformly bounded and thus the unbounded growth happens along the segment $\underline{\ell}$ when it is approaching $\{\hat{t}=0\}$.

A physical scenario of unbounded local energy extraction from the gravitational field near such a singularity might now take the following form: using a finite amount of fuel the observer $\gamma$ ejects a small probe $\sigma$ which then sends radiation back to $\gamma$, see Figure \ref{Fig2}. Using the geometric optics approximation (see for example \cite{MTW} or \cite{Sbie13b}) it then follows that the energy of the radiation received by $\gamma$, which is proportional to minus the inner product of $f_0$ and the parallely propagated null vector along $\underline{\ell}$, exceeds the energy he has put into the probe by an arbitrarily large amount. 

In a second possible scenario the probe $\sigma$ is first accelerated off $\gamma$ and, after a while, it is accelerated back towards $\gamma$ and then approaches $\gamma$ on a \emph{timelike geodesic} trajectory, see Figure \ref{Fig3}. The closer the collision point of $\sigma$ and $\gamma$ is to the singularity, the larger the velocity of $\sigma$ appears to $\gamma$, approaching the speed of light in the limit, see also Remark \ref{RemFLRW}. Since the two phases of acceleration of $\sigma$ are bounded away from $\hat{t} = 0$, these accelerations require a uniformly bounded amount of fuel. On the other hand, the kinetic energy of $\sigma$ gained from the perspective of $\gamma$ can exceed this amount by an arbitrarily large quantity.

Let us remark that in this scenario we have of course reversed the time direction in order to be able to approach the big-bang singularity. However, we point out  in Remark \ref{RemSchwarz} that the Schwarzschild singularity exhibits a similar, although not identical, holonomy structure. It can be approached by a future-heading observer falling into the black hole and the above scenario can be transposed to this case.

This fits into a  body of results which show/indicate that local forms of energy, and not just energy density, can diverge near a singularity, cf.\ for example \cite{PoiIs89}, \cite{HerHis92}, \cite{LukOh19I}, \cite{FouSbi20}.

Motivated by the exhibited holonomy structure of the cosmological and the Schwarzschild singularities, we can define the geometric notion of the \emph{local causal holonomy associated with a timelike geodesic}, which is a subset of the Lorentz group:
Let $(M,g)$ be a time-oriented Lorentzian manifold and $\gamma : [0,1) \to M$ a future directed and future inextendible timelike geodesic. Choose a point $s_0 \in [0,1)$ and set
\begin{equation*}
\begin{split}
l.c.Hol(\gamma, s_0) := \{ P_{\sigma} \in O(1,3)(T_{\gamma(s_0)}M) \; | \; &\textnormal{ for some } s_1 \in (s_0, 1) \textnormal{ we have } \sigma = \overleftarrow{\gamma|_{[s_0, s_1]}} * \tau \;, \\
 &\textnormal{ with } \tau : [s_0, s_1] \to M \textnormal{ being causally homotopic } \\
 &\textnormal{ with fixed endpoints to } \gamma|_{[s_0, s_1]} \}  \;.
\end{split}
\end{equation*}
Here, $P_{\sigma} : T_{\gamma(s_0)}M \to T_{\gamma(s_0)}M$ denotes the parallel transport map along the loop $\sigma$, $O(1,3) (T_{\gamma(s_0)}M)$ denotes the group of Lorentz transformations on $T_{\gamma(s_0)}M$, $\overleftarrow{\gamma|_{[s_0, s_1]}}$ denotes the reversal of $\gamma|_{[s_0, s_1]}$, and a causal homotopy with fixed endpoints is a homotopy with fixed endpoints via causal curves, see also Lemma \ref{LemCausalHomotopy}. 
If $\gamma$ is chosen to approach the Schwarzschild singularity, or, after reversal of time-orientation, $\gamma$ is chosen to approach a cosmological big-bang singularity as above, then in both cases and for all $s_0 \in [0,1)$ we have that $l.c.Hol(\gamma, s_0) $ is a subset of $O(1,3)(T_{\gamma(s_0)}M)$ with non-compact closure. 
\newline

We now discuss the weak null singularities and the proof of Theorem \ref{ThmInt2} in Section \ref{SecGT} by contrasting it with the much cleaner/simpler case of the cosmological singularities. The strategy of the proof is similar, one again assumes that there exists a future $C^{0,1}_{\loc}$-extension and then finds a future directed timelike geodesic $\gamma$ which leaves the original spacetime $(M,g)$ for the extension $(\tilde{M}, \tilde{g})$ through a boundary point around which one has $C^{0,1}_{\loc}$-control over the metric $\tilde{g}$ in a small coordinate neighbourhood. In the cosmological case one could now directly go over to a two-dimensional problem which is not possible here. Another difference to the cosmological case is that the analogous construction of the loops from broken null geodesics gives a uniformly bounded holonomy. In fact, broken timelike geodesics do not work either. This is a manifestation of the fact that unlike in the cosmological case (cf.\ Remark \ref{RemFLRW}) there is only one standard of finite inertial energy at the weak null singularity. We proceed as follows: using homotopy arguments we show that one can also find a radially outgoing null geodesic $\tau$ in $M$ which leaves for the extension. Assume without loss of generality that $\tau$ approaches the singularity $\{v = 0\}$, i.e., $\tau$ is tangent to $\rd_v$. We then construct a family of mixed-null-and-spacelike loops based at a point of $\tau$ near the singularity $\{v=0\}$ by first moving briefly in the angular $\varphi$ direction on the spheres $\mathbb{S}^2$, then moving only in $v$ towards the singularity, moving back in $\varphi$ and returning by moving back in $v$, cf.\ also Figure \ref{FigLoopsNull} on page \pageref{FigLoopsNull}. One can show that the holonomy along those loops becomes unbounded when they approach the singularity at $\{v = 0\}$. It remains to show that these loops are \emph{local} in the sense discussed earlier, which is not as straightforward as in the cosmological case. For this we also make use of the affine structure. 

One could also define a geometric quantity in the case of the weak null singularities which becomes unbounded, although it is more complicated than the local causal holonomy defined earlier. The interested reader however can easily construct such a quantity from the proof in Section \ref{SecGT}.
\newline

As a first self-contained example of a (non-bifurcate) weak null singularity to which our $C^{0,1}_{\loc}$-inextendibility result applies we give the Reissner-Nordstr\"om-Vaidya (RNV) spacetime in Section \ref{SecRNV}, which models the influx of null-dust into a subextremal Reissner-Nordstr\"om black hole. This model has been the earliest exact-solution model used to understand the singularity forming at the Cauchy horizon of dynamical charged or rotating black holes, see \cite{His81}. In \cite{HerHis92} an argument has been made for an infinite energy transfer from the gravitational field near such a singularity to test bodies. We also point out here that the Hawking mass remains uniformly bounded at the Cauchy horizon of the RNV spacetime and that our $C^{0,1}_{\loc}$-inextendibility result does not require mass-inflation (\cite{PoiIs89}, \cite{PoiIs90}, \cite{Ori91}) as an assumption.  

Our main application of Theorem \ref{ThmInt2}  is in Section \ref{SecDLO} to spacetimes arising from generic and sufficiently small spherically symmetric perturbations of asymptotically flat two-ended subextremal Reissner-Nordstr\"om initial data for the Einstein-Maxwell-scalar field system as studied by Luk and Oh in \cite{LukOh19I}, \cite{LukOh19II}. Their work builds up on results by Dafermos \cite{Daf03}, \cite{Daf05a}, \cite{Daf14} and Dafermos-Rodnianski \cite{DafRod05}. A Penrose diagram of those spacetimes is given in Figure \ref{FigPenroseDLO} on page \pageref{FigPenroseDLO} and we refer the reader to \cite{LukOh19I} for a detailed discussion.  In \cite{LukOh19I}, \cite{LukOh19II} Luk and Oh prove the $C^2$-formulation of strong cosmic censorship for this class of spacetimes. Based on their work we improve this to a $C^{0,1}_{\loc}$-formulation in this paper: it can be directly inferred from the estimates in \cite{LukOh19I} that the interior of the black hole satisfies the assumptions in Theorem \ref{ThmInt2} and is thus future $C^{0,1}_{\loc}$-inextendible. In the appendix of this paper we show that the exterior of the black hole is timelike geodesically complete in the sense that any future inextendible timelike geodesic starting in the exterior is either future complete or enters the black hole interior. This then suffices to infer the global $C^{0,1}_{\loc}$-inextendibility of the spacetime and thus the $C^{0,1}_{\loc}$-formulation of strong cosmic censorship.

Other (closely related) examples of spherically symmetric weak null singularities to which our results apply are those constructed in \cite{VdM18} for the Einstein-Maxwell-charged scalar field sysytem and in \cite{CoGiNaDru17} for the Einstein-Maxwell-scalar field system with a positive cosmological constant.
\newline

Let us conclude the introduction by remarking that we do not attempt in this paper to define precisely what we mean by a \emph{holonomy singularity} in general, since we expect that they can come in various forms still to be explored. As we have seen, the holonomy structures of the cosmological singularities considered and the weak null singularities are already very different.

\subsection*{Acknowledgements}  I would like to thank Eric Ling for a stimulating question and I am also grateful to Jonathan Luk for help with the references \cite{LukOh19I} and \cite{LukOh19II}.

\section{Preliminaries} \label{SecPre}

\subsection{Preliminary definitions}

Let $M$ be a smooth $(d+1)$-dimensional manifold\footnote{We will always assume in this paper that manifolds are smooth. Recall that a $C^1$-structure on a manifold can always be refined to a smooth structure, see also the related Remark \ref{RemEmbeddingSmooth}.}. We briefly recall that a Lorentzian metric $g$ on $M$ is \emph{locally Lipschitz} regular (or $C^{0,1}_{\loc}$-regular)  iff for all smooth charts $\psi : M \supseteq U \to V \subseteq \R^{d+1}$ the coordinate expressions of the metric $g_{\mu \nu} \circ \psi^{-1} : \R^{d+1} \supseteq V \to \R$ are locally Lipschitz, i.e., for all compact sets $K \subseteq V$ there exists a constant $\Lambda(K) >0$ such that for all $x,y \in K$ we have
\begin{equation*}
|(g_{\mu \nu} \circ \psi^{-1}) (x) - (g_{\mu \nu} \circ \psi^{-1})(y) | \leq \Lambda(K) ||x - y||_{\R^{d+1}} \;.
\end{equation*}
Here, $||x - y||_{\R^{d+1}} $ denotes the coordinate distance of $x$ and $y$ in $\R^{d+1}$, which, if $V$ is not convex, is not necessarily the distance of $x$ and $y$ in $V \subseteq \R^{d+1}$. However, it is not difficult to show that a function $f : \R^{d+1} \supseteq V \to \R$ is locally Lipschitz with respect to the ambient distance function of $\R^{d+1}$ if, and only if, it is locally Lipschitz with respect to the intrinsic distance function of $V$ (although the optimal Lipschitz constant will be different in general). Thus, both possible definitions give rise to the same class of functions.

Similarly, one defines a locally Lipschitz curve in $M$ or, in general, locally Lipschitz maps between manifolds. Note that all these definitions are independent of a distance function on $M$ but only depend on the smooth structure.

\begin{definition}
Let $(M,g)$ be a smooth Lorentzian manifold and let $\Gamma$ be a regularity class, for example $\Gamma = C^k$ with $k \in \N \cup \{\infty\}$ or $\Gamma = C^{0,1}_{\loc}$.  A \emph{$\Gamma$-extension of $(M,g)$} consists of a smooth isometric embedding $\iota : M \hookrightarrow \tilde{M}$ of $M$ into a Lorentzian manifold $(\tilde{M}, \tilde{g})$ of the same dimension as $M$ where $\tilde{g}$ is $\Gamma$-regular  and such that $\partial \iota(M) \subseteq \tilde{M}$ is non-empty.

If $(M,g)$ admits a $\Gamma$-extension, then we say that $(M,g)$ is \emph{$\Gamma$-extendible}, otherwise we say $(M,g)$ is \emph{$\Gamma$-inextendbile}.
\end{definition}

\begin{remark}\label{RemEmbeddingSmooth}
Recall that the question of extendibility of Lorentzian manifolds is motivated by the physical question of whether spacetime can be continued. In the process of the mathematical modelling one might wonder whether one gains continuability by also lowering the regularity of the differentiable structure of the manifold itself in the above definition. But indeed this is not the case as long as one does not go below a $C^1$-differentiable structure, which is needed for the existence of a continuous tangent space and the notion of a continuous Lorentzian metric. To see this, let $M$ be a smooth manifold, $g$ a smooth Lorentzian metric on $M$, $\tilde{M}$ a $C^1$ manifold with a $C^0$ Lorentzian metric $\tilde{g}$ and let $\iota : M \hookrightarrow \tilde{M}$ be a $C^1$ isometric embedding. Then $\iota$ induces on $\iota(M) \subseteq \tilde{M}$ a smooth structure which is compatible with the given $C^1$ structure on $\iota(M)$ and with respect to which $\tilde{g}|_{\iota(M)}$ is smooth. The proof of Theorem 2.9 in Chapter 2 of \cite{Hirsch12} shows that this smooth structure can be extended to a smooth one on all of $\tilde{M}$ which is compatible with the given $C^1$ structure -- thus turning $\tilde{M}$ into a smooth manifold and $\iota$ into a smooth isometric embedding, and thus recovering the stronger assumptions in the above definition. 
\end{remark}

Let now $(M,g)$ be a Lorentzian manifold with a continuous metric. In this paper we use the convention that a \emph{timelike curve} is a piecewise smooth curve which has a timelike tangent everywhere -- and at the points of discontinuity of the tangent the right and left tangent vectors lie in the same connectedness component of the timelike double cone of tangent vectors. Similarly we define a \emph{causal curve} as a piecewise $C^1$ curve which has a causal, non-vanishing tangent everywhere -- and at points of discontinuity of the tangent the right and left tangent vectors lie in the same connectedness component of the causal double cone of tangent vectors with the origin removed. Let $(M,g)$ be in addition time-oriented. For $p \in M$ we denote the \emph{timelike future} of $p$ in $M$ by  $I^+(p,M)$, which is  the set of all points $q \in M$ such that there is a future directed timelike curve  from $p$ to $q$. The \emph{causal future} of $p$ in $M$, denoted by $J^+(p,M)$, is the set which contains $p$ and all points $q \in M$ such that there is a future directed causal curve from $p$ to $q$. The sets $I^-(p,M)$ and $J^-(p,M)$ are defined analogously.

Note that for Lorentzian manifolds with a merely continuous metric there are  good reasons for defining timelike and causal curves as locally Lipschitz curves with a timelike or causal tangent almost everywhere, which also leads to different causal sets, cf.\ \cite{ChrusGra12}, \cite{GraKuSaSt19}. However this is not needed for our purposes.

\subsection{Fundamentals of $C^0$-extensions} \label{SecFundResults}

We now recall some fundamental definitions and results for $C^0$-extensions. 

\begin{definition}
Let $(M,g)$ be a smooth time-oriented Lorentzian manifold and $\iota : M \hookrightarrow \tilde{M}$ a $C^0$-extension of $M$. The \emph{future boundary of $M$} is the set $\partial^+\iota(M) $ consisting of all points $\tilde{p} \in \tilde{M}$ such that there exists a smooth timelike curve $\tilde{\gamma} : [-1,0] \to \tilde{M}$ such that $\mathrm{Im}(\tilde{\gamma}|_{[-1,0)}) \subseteq \iota(M)$, $\tilde{\gamma}(0) = \tilde{p} \in \partial \iota(M)$, and $\iota^{-1} \circ \tilde{\gamma}|_{[-1,0)}$ is future directed in $M$.
\end{definition}
Clearly we have $\partial^+\iota(M) \subseteq \partial \iota(M)$. The past boundary $\partial^- \iota(M)$ is defined analogously. 

\begin{definition}
Let $(M,g)$  be a smooth time-oriented Lorentzian manifold and let $\Gamma$ be a regularity class that is equal to or stronger than $C^0$. A \emph{future $\Gamma$-extension of $(M,g)$} is a $\Gamma$-extension $\iota : M \hookrightarrow \tilde{M}$  of $M$ with $\partial^+\iota(M) \neq \emptyset$. If no such extension exists, then $(M,g)$ is said to be \emph{future $\Gamma$-inextendible}.
\end{definition}
Past $\Gamma$-extensions are defined analogously.
The next lemma is a reformulation of Lemma 2.17 in \cite{Sbie15}.
\begin{lemma} \label{LemFuturePastExt}
Let $(M,g)$ be a smooth time-oriented Lorentzian manifold and $\iota : M \hookrightarrow \tilde{M}$ a $C^0$-extension of $M$. Then $\partial^+\iota(M) \cup \partial^- \iota(M) \neq \emptyset$.
\end{lemma}
In particular the Lemma shows that if $(M,g)$ is future and past $\Gamma$-inextendible, then it is also $\Gamma$-inextendible.

The past and future boundary interchange under a change of time orientation of $(M,g)$. It is thus sufficient to focus in the following on the future boundary.
The next proposition is found in \cite{Sbie18}, Proposition 2.2.

\begin{proposition}\label{PropBoundaryChart}
Let $\iota : M \hookrightarrow \tilde{M}$ be a $C^0$-extension of a smooth time-oriented globally hyperbolic Lorentzian manifold $(M,g)$ with Cauchy hypersurface $\Sigma$  and let $\tilde{p} \in \partial^+ \iota(M)$. For every $\delta >0$ there exists a chart $\tilde{\varphi} : \tilde{U} \to(-\varepsilon_0, \varepsilon_0) \times  (-\varepsilon_1, \varepsilon_1)^{d} =: R_{\varepsilon_0, \varepsilon_1}$, $\varepsilon_0, \varepsilon_1 >0$ with the following properties
\begin{enumerate}[i)]
\item $\tilde{p} \in \tilde{U}$ and $\tilde{\varphi}(p) = (0, \ldots, 0)$
\item $|\tilde{g}_{\mu \nu} - m_{\mu \nu}| < \delta$, where $m_{\mu \nu} = \mathrm{diag}(-1, 1, \ldots , 1)$
\item There exists a Lipschitz continuous function $f : (-\varepsilon_1, \varepsilon_1)^d \to (-\varepsilon_0, \varepsilon_0)$ with the following property: 
\begin{equation}\label{PropF1}
\{(x_0,\underline{x}) \in (-\varepsilon_0, \varepsilon_0) \times (-\varepsilon_1, \varepsilon_1)^{d} \; | \: x_0 < f(\underline{x})\} \subseteq \tilde{\varphi} \big( \iota\big(I^+(\Sigma,M)\big)\cap \tilde{U}\big)
\end{equation} and 
\begin{equation}\label{PropF2}
\{(x_0,\underline{x}) \in (-\varepsilon_0, \varepsilon_0) \times (-\varepsilon_1, \varepsilon_1)^{d}  \; | \: x_0 = f(\underline{x})\} \subseteq \tilde{\varphi}\big(\partial^+\iota(M)\cap \tilde{U}\big) \;.
\end{equation}
Moreover, the set on the left hand side of \eqref{PropF2}, i.e. the graph of $f$, is achronal\footnote{With respect to \emph{smooth} timelike curves.} in $(-\varepsilon_0, \varepsilon_0) \times  (-\varepsilon_1, \varepsilon_1)^{d}$.
\end{enumerate}
\end{proposition}

Note that any past directed causal curve starting below the graph of $f$ remains below the graph of $f$, since if it crossed the graph of $f$ it would, via $\iota^{-1}$, give rise to a past directed past inextendible causal curve in $M$ which starts in $I^+(\Sigma,M)$ but does not intersect $\Sigma$ -- which contradicts $\Sigma$ being a Cauchy hypersurface.

We define
\begin{itemize}
\item $C^+_a := \big{\{} X \in \R^{d+1} \, | \, \frac{<X,e_0>_{\R^{d+1}}}{||X||_{\R^{d+1}}}   > a \big{\}}$
\item $C^-_a := \big{\{} X \in \R^{d+1} \, | \, \frac{<X,e_0>_{\R^{d+1}}}{||X||_{\R^{d+1}}}   < -a \big{\}}$
\item $C^c_a := \big{\{} X \in \R^{d+1} \, | \, -a < \frac{<X,e_0>_{\R^{d+1}}}{||X||_{\R^{d+1}}}   < a \big{\}}$\;.
\end{itemize}
Here, $C^+_a$ is the forward cone of vectors which form an angle of less than $\cos^{-1}(a)$ with the $x_0$-axis, and $C^-_a$ is the corresponding backwards cone. In Minkowski space, the forward and backward cones of timelike vectors correspond to the value $a = \cos(\frac{\pi}{4}) = \frac{1}{\sqrt{2}}$.

Since $\frac{5}{8} < \frac{1}{\sqrt{2}} < \frac{5}{6}$, one can choose $\delta >0$ in Proposition \ref{PropBoundaryChart} small enough such that in the chart $\tilde{\varphi}$  all vectors in $C^+_{\nicefrac{5}{6}}$ are future directed timelike, all vectors in $C^-_{\nicefrac{5}{6}}$ are past directed timelike, and all vectors in $C^c_{\nicefrac{5}{8}}$ are spacelike.

A first easy consequence of this is that we have the inclusion relations\footnote{See proof of Theorem 3.1, Step 1.2 in \cite{Sbie15} for the second inclusion -- the first follows directly.} 
\begin{equation}\label{EqInclusionRelationFuture}
\begin{split}
&\big( x + C^+_{\nicefrac{5}{6}}\big) \cap  \Reps \subseteq J^+(x, \Reps) \subseteq \big( x + C^+_{\nicefrac{5}{8}}\big) \cap  \Reps \\
&\big( x + C^-_{\nicefrac{5}{6}}\big) \cap  \Reps \subseteq J^-(x, \Reps) \subseteq \big( x + C^-_{\nicefrac{5}{8}}\big) \cap  \Reps  \;.
\end{split}
\end{equation}
A second consequence is that the $x_0$ coordinate is a time function  and thus if $\tilde{\gamma}$ is a causal curve in $R_{\varepsilon_0, \varepsilon_1}$ we may reparametrise it by the $x_0$ coordinate, i.e., $\tilde{\gamma}(s) = \big(s, \overline{\tilde{\gamma}}(s)\big)$. In this parametrisation we have $$\frac{<\dot{\tilde{\gamma}}(s), \partial_0>_{\R^{d+1}}}{||\dot{\tilde{\gamma}}(s)||_{\R^{d+1}}} \geq \frac{5}{8} \;,$$
and thus we obtain the uniform bound 
\begin{equation}\label{EqUniformBoundCausalCurve}
||\dot{\tilde{\gamma}}(s)||_{\R^{d+1}} \leq \frac{8}{5}\;.
\end{equation}

The next proposition follows from Proposition \ref{PropBoundaryChart} together with the proof of Theorem 2 in \cite{GalLinSbi17}, see also Theorem 3.2 in \cite{Sbie18}.

\begin{proposition}\label{PropGeodesicBoundaryChart}
Let $(M,g)$ be a smooth time-oriented globally hyperbolic Lorentzian manifold and $\iota : M \hookrightarrow \tilde{M}$ a $C^0$-extension. Assume that $\partial^+ \iota(M) \neq \emptyset$ and let $\tilde{p} \in \partial^+ \iota(M)$. Let $\tilde{\varphi} : \tilde{U} \to (-\varepsilon_0, \varepsilon_0) \times (-\varepsilon_1, \varepsilon_1)^d$ be a chart around $\tilde{p}$ as in Proposition \ref{PropBoundaryChart}. Then there exists a future directed timelike geodesic $\tau : [-1,0) \to M$ that is future inextendible in $M$ and such that $\tilde{\varphi} \circ \iota \circ \tau :[-1,0) \to (-\varepsilon_0, \varepsilon_0) \times (-\varepsilon_1, \varepsilon_1)^d$ maps into $\{(s,\underline{x}) \in (-\varepsilon_0, \varepsilon_0) \times (-\varepsilon_1, \varepsilon_1)^{d} \; | \: s < f(\underline{x})\}$ and has an endpoint on $\{(s,\underline{x}) \in (-\varepsilon_0, \varepsilon_0) \times (-\varepsilon_1, \varepsilon_1)^{d} \; | \: s = f(\underline{x})\}$.
\end{proposition}

The next lemma will be used in our later applications to obtain topological information from Lorentzian causality.

\begin{lemma}\label{LemCausalHomotopy}
Let $(M,g)$ be a time-oriented Lorentzian manifold with $g \in C^0$ and let $\iota : M \hookrightarrow \tilde{M}$ be a $C^0$-extension of $M$. Moreover, let $\gamma : [0,1] \to M$ be a future directed timelike curve and let $\tilde{U} \subseteq \tilde{M}$ be an open set. Assume that $\tilde{\gamma} := \iota \circ \gamma$ maps into $\tilde{U}$ and that $J^+(\tilde{\gamma}(0), \tilde{U}) \cap J^-(\tilde{\gamma}(1), \tilde{U}) \subset \subset \tilde{U}$ is compactly contained in $\tilde{U}$.

Let $\Gamma : [0,1] \times [0,1] \to M$ be a causal homotopy of $\gamma$ with fixed endpoints, i.e., \begin{enumerate}
\item $s \mapsto \Gamma(s ; r)$ is a future directed causal curve for all $r \in [0,1]$ with $\Gamma(0;r) = \gamma(0)$ and $\Gamma(1;r) = \gamma(1)$
\item $\Gamma(s;0) = \gamma(s)$ for all $s \in [0,1]$.
\end{enumerate}

Then $\iota \circ \Gamma$ maps into $\tilde{U}$.
\end{lemma} 

We briefly elaborate on the statement of this lemma. We recall that the causal relations are global in nature. Given a causal curve $\sigma$ in $M$ from $\gamma(0)$ to $\gamma(1)$, then $\iota \circ \sigma$ lies by definition in  $J^+(\tilde{\gamma}(0),\tilde{M}) \cap J^-(\tilde{\gamma}(1),\tilde{M})$, but in general it does not have to lie in $J^+(\tilde{\gamma}(0),\tilde{U}) \cap J^-(\tilde{\gamma}(1),\tilde{U})$, even if the latter set is compact in $\tilde{U}$.  In applications $\tilde{U}$ will be chosen to be a small neighbourhood of $\gamma$ and thus the lemma shows that causal curves that are causally homotopic to $\gamma$ belong to the \emph{local causality} of $\tilde{U}$. Related techniques have already been employed by the author in \cite{Sbie15}.

\begin{proof}
The proof is by continuity. Let $I:= \{ r \in [0,1] \; | \; s \mapsto (\iota \circ \Gamma) (s ; r) \textnormal{ maps into } \tilde{U}\}$. Since $\tilde{\gamma}$ maps into $\tilde{U}$ we have $0 \in I$ and thus $I $ is not empty. Also $I$ is open since $\tilde{U}$ is open. To show that $I$ is closed, let $r_n \in I$ be a sequence with $r_n \to r_\infty \in [0,1]$ for $n \to \infty$. Since $s \mapsto (\iota \circ \Gamma)(s ; r_n)$ is a future directed causal curve in $\tilde{U}$ from $\tilde{\gamma}(0)$ to $\tilde{\gamma}(1)$, it lies in $J^+(\tilde{\gamma}(0), \tilde{U}) \cap J^-(\tilde{\gamma}(1), \tilde{U})$. Since this set is compactly contained in $\tilde{U}$ by assumption, it follows that also $s \mapsto (\iota \circ \Gamma)(s; r_\infty)$ maps into $\tilde{U}$. This shows that $I = [0,1]$, which concludes the proof.
\end{proof}

\subsection{Lemma for bounding local holonomy in $C^{0,1}_{\loc}$-extensions}
\label{SecLemBound}
The following lemma is  fundamental to all our later applications.

\begin{lemma}\label{LemBoundParallelTransport}
Let $(\tilde{O}, \tilde{g})$ be a Lorentzian manifold with a $C^1$-regular Lorentzian metric $\tilde{g}$, which, moreover, is endowed with a global coordinate chart $\tilde{\varphi} : \tilde{O} \to \tilde{V} \subseteq \R^{d+1}$ with respect to which the metric components satisfy $|\tilde{g}_{\mu \nu}| \leq C_{\tilde{g}}$ and $|\partial_\kappa \tilde{g}_{\mu \nu}| \leq C_{\partial \tilde{g}}$. Let $\tilde{\gamma} : [0,T] \to \tilde{O}$ be a smooth curve with\footnote{Here, for a vector $X \in T_{\tilde{p}}\tilde{O}$, the norm $||X||_{\R^{d+1}} := \sqrt{\sum_{\mu = 0}^d (X^\mu)^2}$ denotes the standard Euclidean norm with respect to the global coordinate chart.}  
$||\dot{\tilde{\gamma}}||_{\R^{d+1}} \leq C_{\dot{\tilde{\gamma}}}$, let $X \in T_{\tilde{\gamma}(0)}\tilde{O}$, and let $P : T_{\tilde{\gamma}(0)}\tilde{O} \to T_{\tilde{\gamma}(T)}\tilde{O}$ denote the parallel transport map along $\tilde{\gamma}$.

Then $||P(X)||_{\R^{d+1}} \leq ||X||_{\R^{d+1}} \cdot e^{C(C_{\tilde{g}}, C_{\partial \tilde{g}}, C_{\dot{\tilde{\gamma}}}) \cdot T}$, where the constant $C(C_{\tilde{g}}, C_{\partial \tilde{g}}, C_{\dot{\tilde{\gamma}}})$ depends on $\tilde{\gamma}$ only via $C_{\dot{\tilde{\gamma}}}$.
\end{lemma}

This lemma will be applied to the region below the graph of $f$ in a chart $\tilde{U}$ as in Proposition \ref{PropBoundaryChart} to uniformly bound the parallel transport around a loop. Note that all that is needed to do so is a uniform bound on the Euclidean norm of the tangent vector  and the domain of definition of the curve. 

\begin{proof}
Let $X \in T_{\tilde{\gamma}(0)}\tilde{O}$ and let $s \mapsto X(s) \in T_{\tilde{\gamma}(s)} \tilde{O}$ denote the parallel transport of $X$ along $\tilde{\gamma}$. In the global coordinate chart we have $$ 0 = \Big(\frac{D}{ds} X(s)\Big)^\mu = \frac{d}{ds} X^\mu(s) + \Gamma^\mu_{\kappa \rho}(\tilde{\gamma}(s)) \dot{\tilde{\gamma}}^\kappa(s) X^\rho(s) = \frac{d}{ds} X^\mu(s) + A^\mu_{\;\; \rho}(s) X^\rho(s) \;,$$
where we have defined $A^\mu_{\;\; \rho}(s) :=   \Gamma^\mu_{\kappa \rho}(\tilde{\gamma}(s)) \dot{\tilde{\gamma}}^\kappa(s)$, an $s$-dependent linear map from $\R^{d+1}$ to $\R^{d+1}$.
By the assumptions we have $|\Gamma^\mu_{\kappa \rho}(\tilde{\gamma}(s)) \dot{\tilde{\gamma}}^\kappa(s)| \leq C(C_{\tilde{g}}, C_{\partial{\tilde{g}}}, C_{\dot{\tilde{\gamma}}})$ and thus we obtain for the operator norm of $A$, $||A(s)|| \leq  C(C_{\tilde{g}}, C_{\partial{\tilde{g}}}, C_{\dot{\tilde{\gamma}}})$. Furthermore, we compute
\begin{equation*}
\begin{split}
\frac{d}{ds} ||X(s)||_{\R^{d+1}} &= \frac{\sum_{\mu = 0}^d X^\mu(s) \frac{d}{ds} X^\mu(s)}{||X(s)||_{\R^{d+1}}}\leq \frac{||X(s)||_{\R^{d+1}} ||\frac{d}{ds} X^\mu(s)||_{\R^{d+1}}}{||X(s)||_{\R^{d+1}}}  = ||A(s) X(s)||_{\R^{d+1}} \\
&\leq C(C_{\tilde{g}}, C_{\partial{\tilde{g}}}, C_{\dot{\tilde{\gamma}}}) ||X(s)||_{\R^{d+1}} \;,
\end{split}
\end{equation*}
which gives $||X(s)||_{\R^{d+1}} \leq ||X(0)||_{\R^{d+1}} e^{C(C_{\tilde{g}}, C_{\partial{\tilde{g}}}, C_{\dot{\tilde{\gamma}}}) s}$ using Gronwall's inequality.
\end{proof}

\section{Cosmological singularities} \label{SecCos}

In this section we consider the following class of \emph{cosmological warped product spacetimes} $(M,g)$: Let $(\overline{M}, \overline{g})$ be a $3$-dimensional complete Riemannian manifold\footnote{Everything goes through  unchanged up to notation also for $d$-dimensional complete Riemannian manifolds, where $d \geq 1$.}, let $0 < b \leq \infty$ and let $a :(0,b) \to (0, \infty)$ be a smooth function.   We then set $M= (0,b) \times \overline{M}$ and 
\begin{equation}
\label{EqMetricWP} g = -dt^2 + a(t)^2 \overline{g}\;.
\end{equation} 
We fix a time-orientation on $(M,g)$ by stipulating that $\partial_t$ is future directed. Clearly, all spacetimes $(M,g)$ in our class are globally hyperbolic.  

Let $i,j,k$ denote local spatial coordinates on $(\overline{M},\overline{g})$. The Christoffel symbols of $g = -dt^2 + a(t)^2 \overline{g}$ are then given by
\begin{equation}\label{EqChristoffel}
\begin{aligned}
\Gamma^t_{tt} &= 0 \qquad  &&\Gamma^t_{ti} =0 \qquad  &&&\Gamma^t_{ij} = \dot{a}(t) a(t) \overline{g}_{ij} \\
\Gamma^i_{tt} &= 0 \qquad  &&\Gamma^i_{tj} = \frac{\dot{a}(t)}{a(t)} \delta^i_{\; \; j} \qquad  &&&\Gamma^i_{jk} = \overline{\Gamma}^i_{jk} \;,
\end{aligned}
\end{equation}
where $\overline{\Gamma}^i_{jk}$ denotes the Christoffel symbols of $(\overline{M}, \overline{g})$.
Let now $\gamma$ be a geodesic in $(M,g)$. It thus satisfies
\begin{align*} 
0&=\ddot{\gamma}^t + \Gamma^t_{ij} \dot{\gamma}^i \dot{\gamma}^j \\
0&= \ddot{\gamma}^i + 2\Gamma^i_{tj} \dot{\gamma}^t \dot{\gamma}^j + \Gamma^i_{jk} \dot{\gamma}^j \dot{\gamma}^k = \ddot{\gamma}^i + 2 \frac{\frac{d}{dt}a}{a} \dot{\gamma}^t \dot{\gamma}^i + \overline{\Gamma}^i_{jk} \dot{\gamma}^j \dot{\gamma}^k \;,
\end{align*}
which shows that if $\overline{\gamma}$ denotes the projection of $\gamma$ to $\overline{M}$, then $\overline{\nabla}_{\dot{\overline{\gamma}}} \dot{\overline{\gamma}} $ is proportional to $\dot{\overline{\gamma}}$. Thus, the projected geodesics are still geodesics in $(\overline{M}, \overline{g})$, although in general not affinely parametrised ones. 

Choose now polar normal coordinates $(\chi, \theta, \varphi) \in (0,\chi_0) \times \mathbb{S}^2$ on some open set $\overline{U} \subseteq \overline{M}$ such that in those coordinates the metric $\overline{g}$ takes the form 
\begin{equation}\label{EqMetricPolarNC}
\overline{g} = d \chi^2 + \sum_{A,B = 1}^2\overline{g}_{AB}(\chi, \theta, \varphi) \,dx^A \otimes dx^B\;,
\end{equation} 
where $x^1 = \theta$ and $x^2=\varphi$. Then the radial curves of constant $\theta$ and $\varphi$ are geodesics. It thus follows that if a geodesic in $(0,b) \times \overline{U} \subseteq M$ has initial tangent vector which lies in $\mathrm{span} \{\partial_t, \partial_\chi\}$ then it stays in the $\{t, \chi\}$-plane\footnote{Note that the curves of constant $\chi, \theta, \varphi$ are also geodesic, which follows from $\Gamma^t_{tt} = 0$.}. We have thus shown the following\footnote{Of course, one can also use \eqref{EqChristoffel} to show that the extrinsic curvature of the $\{t, \chi\}$-planes vanishes in order to infer the following proposition.}

\begin{proposition}\label{PropTotGeod}
Let $(M,g)$ be a cosmological warped product spacetime in the above class and let $(\chi, \theta, \varphi) \in (0,\chi_0) \times \mathbb{S}^2$ be local polar normal coordinates on some open set $\overline{U} \subseteq \overline{M}$. Then the $\{t, \chi\}$-planes in $(0,b) \times \overline{U}$ are totally geodesic.
\end{proposition}

Results regarding the future $C^0$-inextendibility of a subclass of cosmological warped product spacetimes have also been established by Galloway and Ling in \cite{GalLin16}. Here we prove

\begin{theorem}\label{ThmFLRWFuture}
Let $(M,g)$ be a cosmological warped product spacetime in the above class with  $b = \infty$ and let the scale factor $ a : (0,\infty) \to (0,\infty)$ satisfy $\int_1^\infty \frac{a(t)}{ \sqrt{a(t)^2 + 1}}\, dt = \infty$.
Then $(M,g)$ is future $C^0$-inextendible.
\end{theorem}

Note that the condition in the above theorem is in particular satisfied if $a(t)$ is bounded away from $0$ for large $t$.

\begin{proof}
We only need to establish the future timelike geodesic completeness of $(M,g)$. The theorem then follows from Theorem 3.6 in \cite{GalLinSbi17}.

So let $\gamma : (0,d) \to M$ be a future directed future inextendible timelike geodesic that is parametrised by proper time, $\tau \mapsto \big(t(\tau), \overline{\gamma}(\tau)\big)$. We need to show that $d = \infty$. Let $\tau_0 \in (0,d)$. We choose local polar normal coordinates $(\chi, \theta, \varphi) \in (0,\chi_0) \times \mathbb{S}^2$ on some open neighbourhood $\overline{U} \subseteq \overline{M}$ of $\overline{\gamma}(\tau_0)$ such that $\dot{\overline{\gamma}}(\tau_0)$ is proportional to $\partial_\chi$ (or it vanishes). Then $\dot{\gamma}(\tau_0) \in \mathrm{span} \{\partial_t, \partial_\chi\}$ and by Proposition \ref{PropTotGeod} $\gamma$ stays in the $\{t,\chi\}$-plane in $(0,\infty) \times \overline{U}$. Note that by \eqref{EqMetricWP} and \eqref{EqMetricPolarNC} $\partial_\chi$ is a Killing vector field \emph{in the $\{t,\chi\}$-plane} and thus we have $g|_{\{t,\chi\}\mathrm{-plane}}(\dot{\gamma}, \partial_\chi) = a(t)^2 \dot{\gamma}^\chi = \kappa \in \R$. Hence, we obtain
\begin{equation}\label{EqSpeedSpatial}
a^2(t) \overline{g}(\dot{\overline{\gamma}}, \dot{\overline{\gamma}}) = a(t)^2 \dot{\gamma}^\chi \dot{\gamma}^\chi = \frac{\kappa^2}{a(t)^2}
\end{equation}
for the connected neighbourhood of $\tau_0 \in (0,d)$ such that $\overline{\gamma}$ remains in $\overline{U}$. We can now cover the image of $\overline{\gamma}$ in $\overline{M}$ with local polar normal coordinates as above and repeat the argument. In the overlap covered by two such charts the constants $\kappa$ in \eqref{EqSpeedSpatial} have to agree, since \eqref{EqSpeedSpatial} is independent of the choice of coordinates. We thus obtain \eqref{EqSpeedSpatial} for all $\tau \in (0,d)$. Moreover, together with  $$-1 = g(\dot{\gamma}, \dot{\gamma}) = -(\dot{\gamma}^t)^2 + a(t)^2\overline{g}(\dot{\overline{\gamma}}, \dot{\overline{\gamma}})  = -(\dot{\gamma}^t)^2 + \frac{\kappa^2}{a(t)^2} $$ we obtain $\Big(\frac{dt}{d\tau}\Big)^2 = 1 + \frac{\kappa^2}{a(t)^2} = \frac{a(t)^2 + \kappa^2}{a(t)^2}$. Note that for $\tau \to d$ we must have $t(\tau) \to \infty$, since if $t(\tau) \to t_0 < \infty$, then it follows from \eqref{EqSpeedSpatial} together with the completeness of $(\overline{M}, \overline{g})$ that $\overline{\gamma}(\tau)$ has a limit point in $\overline{M}$ for $\tau \to d$ -- and thus $\gamma(\tau)$ has a limit point in $M$ for $\tau \to d$, contradicting the future inextendibility of $\gamma$. 

Since $\frac{dt}{d\tau} >0$ we obtain $\frac{d \tau}{dt}  = \frac{a(t)}{\sqrt{a(t)^2 + \kappa^2}}$. 
If $|\kappa| < 1$ then $\frac{a(t)}{\sqrt{a(t)^2 + \kappa^2}} \geq \frac{a(t)}{\sqrt{a(t)^2 + 1}}$ and if $|\kappa| \geq 1$, then $\frac{1}{|\kappa|} \frac{a(t)}{\sqrt{a(t)^2 + 1}} = \frac{a(t)}{\sqrt{\kappa^2 a(t)^2 + \kappa^2}} \leq \frac{a(t)}{\sqrt{a(t)^2 + \kappa^2}}$. It thus follows by assumption that $\int_1^\infty \frac{a(t)}{\sqrt{a(t)^2 + \kappa^2}} \, dt = \infty$ and thus $d = \infty$.
\end{proof}

We now discuss the inextendibility to the past and show that if  $\lim_{t \to 0} a(t) = 0$ and $\frac{1}{a(t)}$ is integrable near $0$, then a holonomy singularity is present. 

\begin{theorem}\label{ThmFLRWPast}
Let $(M,g)$ be a cosmological warped product spacetime in the above class satisfying  $\lim_{t \to 0} a(t) = 0$ and $\int_0^{\nicefrac{b}{2}} \frac{1}{a(t')} \, dt' < \infty$. Then $(M,g)$ is past $C^{0,1}_{\loc}$-inextendible.
\end{theorem}

We need the following well-known fact which in particular expresses that the class of spacetimes considered in Theorem \ref{ThmFLRWPast} exhibit past particle horizons.

\begin{lemma}\label{LemHorizon}
Let $(M,g)$ be a cosmological warped product spacetime in the above class satisfying $\int_0^{\nicefrac{b}{2}} \frac{1}{a(t')} \, dt' < \infty$ and let $\sigma : [-1,0) \to M$ be a past directed past inextendible causal curve in $M$. Then $\lim_{\tau \to 0} t\big(\sigma(\tau)\big) = 0$ and $\lim_{\tau \to 0} \overline{\sigma}(\tau)$ exists in $\overline{M}$.
\end{lemma}

\begin{proof}
The past inextendibility  of $\sigma$ together with the completeness of $(\overline{M}, \overline{g})$ imply that $\lim_{\tau \to 0} t\big(\sigma(\tau)\big) = 0$. To see the second claim, we parametrise $\sigma$ by the $t$-coordinate, i.e., $t \mapsto \sigma(t) = (t, \overline{\sigma}(t))$. We then have $0 \geq g(\dot{\sigma}(t), \dot{\sigma}(t)) = -1 + a(t)^2 \overline{g}(\dot{\overline{\sigma}}(t), \dot{\overline{\sigma}}(t))$, which yields $$\frac{1}{a(t)} \geq \sqrt{\overline{g}(\dot{\overline{\sigma}}(t), \dot{\overline{\sigma}}(t))} \;.$$
The assumption $\int_0^{\nicefrac{b}{2}} \frac{1}{a(t')} \, dt' < \infty$ together with the completeness of $(\overline{M}, \overline{g})$ implies that $\lim_{t \to 0} \overline{\sigma}(t)$ exists in $\overline{M}$. 
\end{proof}

\begin{proof}[Proof of Theorem \ref{ThmFLRWPast}:]
In order to use the results from Section \ref{SecFundResults} without changing from future to past, let us change the time orientation on $(M,g)$, i.e., \underline{we redefine $-\partial_t$ to be future directed}, and we thus have to show that $(M,g)$ is future $C^{0,1}_{\loc}$-inextendible.  The proof is by contradiction and proceeds in four steps.
\newline

\underline{\textbf{Step 1:}} Assume $\iota : M \hookrightarrow \tilde{M}$ is a $C^{0,1}_{\loc}$-extension with $\partial^+ \iota(M) \neq \emptyset$. Let $\tilde{p} \in \partial^+ \iota(M)$. Then by Proposition \ref{PropBoundaryChart} there exists a chart $\tilde{\varphi} : \tilde{U} \to(-\varepsilon_0, \varepsilon_0) \times  (-\varepsilon_1, \varepsilon_1)^{d} =: R_{\varepsilon_0, \varepsilon_1}$ as in Proposition \ref{PropBoundaryChart} with $\delta >0$ so small that all vectors in $C^+_{\nicefrac{5}{6}}$ are future directed timelike, all vectors in $C^-_{\nicefrac{5}{6}}$ are past directed timelike, and all vectors in $C^c_{\nicefrac{5}{8}}$ are spacelike. After making the chart slightly smaller if necessary, we can also assume that in this chart the Lorentzian metric $\tilde{g}$ on $\tilde{U}$ satisfies a global Lipschitz condition
\begin{equation*}
|\tilde{g}_{\mu \nu}(x) - \tilde{g}_{\mu \nu}(y)| \leq \Lambda ||x - y||_{\R^{d+1}} 
\end{equation*}
for all $x,y \in \Reps$, where $\Lambda >0$ is a constant. In the region $\{(x_0,\underline{x}) \in \Reps \; | \: x_0 < f(\underline{x})\}$ below the graph of $f$ the metric components $\tilde{g}_{\mu \nu}$ are smooth and satisfy the bounds $|\partial_\kappa \tilde{g}_{\mu \nu}| \leq \Lambda$. By Proposition \ref{PropGeodesicBoundaryChart} there exists a future directed timelike  geodesic $\gamma : [-\mu,0) \to M$ that is future inextendible in $M$ and such that $\tilde{\varphi} \circ \iota \circ \gamma :[-\mu,0) \to \Reps$ maps below the graph of $f$ and has a future endpoint on the graph of $f$. We also define $\tilde{\gamma} := \iota \circ \gamma$.
\newline

\underline{\textbf{Step 2:}} \underline{\textbf{Setting up a coordinate system on $M$ adapted to $\gamma$:}} 

By Lemma \ref{LemHorizon} we know that $\lim_{s \to 0} t(\gamma(s)) = 0$  and that $\lim_{s \to 0} \overline{\gamma}(s)$ exists in $\overline{M}$. In particular it is easy to see from \eqref{EqMetricWP} that $\gamma : [-\mu, 0 ) \to M$ must be future incomplete. We can thus assume without loss of generality that it is affinely parametrised.
As in the proof of Theorem \ref{ThmFLRWFuture}, using Proposition \ref{PropTotGeod}, we now choose polar normal coordinates $(\chi, \theta, \varphi) \in (0,\chi_0) \times \mathbb{S}^2$ on some open neighbourhood $\overline{U} \subseteq \overline{M}$ of $\lim_{s \to 0} \overline{\gamma}(s)$ such that, after making $\mu>0$ smaller if necessary, $\gamma|_{[-\mu,0)}$ lies completely in the $\{t, \chi\}$-plane of the coordinate chart on $(0,b) \times \overline{U} \subseteq M$   and such that, without loss of generality, $\lim_{s \to 0} \chi(\overline{\gamma}(s)) = \frac{\chi_0}{2}$. 

Let us denote the $\{t,\chi\}$-plane by $M_{t, \chi}$. The metric restricted to $M_{t, \chi}$ is $g|_{\{t,\chi\}} = -dt^2 + a(t)^2 d\chi^2$. The affinely parametrised timelike geodesic takes the coordinate form $s \mapsto \big(\gamma^t(s), \gamma^\chi(s)\big)$. Using that $\partial_\chi$ is a Killing vector field in $M_{t, \chi}$ we obtain that $\gamma$ satisfies the equations
\begin{equation*}
\begin{split}
\R \ni \kappa &= g(\dot{\gamma}, \partial_\chi) = a(t)^2 \dot{\gamma}^\chi \\
-1 &= - (\dot{\gamma}^t)^2 + a(t)^2 \big(\dot{\gamma}^\chi\big)^2 = - (\dot{\gamma}^t)^2 + \frac{\kappa}{a(t)^2} \;.
\end{split}
\end{equation*}
We thus obtain $\dot{\gamma} = - \frac{\sqrt{a(t)^2 + \kappa^2}}{a(t)} \partial_t + \frac{\kappa}{a(t)^2} \partial_\chi$. Without loss of generality we can assume that the polar normal coordinates $(\chi, \theta, \varphi)$ were chosen such that $\kappa \geq 0$\footnote{I.e., such that the projected geodesic $\overline{\gamma}$ in $\overline{M}$ is an outgoing radial geodesic in the normal coordinates.}, which we will do from now on.

\underline{\textbf{Setting up a frame field adapted to $\gamma$:}}

We define an orthonormal frame field on $M_{t, \chi}$, depending on $\kappa \geq 0$, by \begin{align*}
f_{0,\kappa} &:= -\frac{\sqrt{a(t)^2 + \kappa^2}}{a(t)} \partial_t + \frac{\kappa}{a(t)^2} \partial_\chi \\
f_{1,\kappa} &:= -\frac{\kappa}{a(t)} \partial_t + \frac{\sqrt{a(t)^2 + \kappa^2}}{a(t)^2} \partial_\chi \;.
\end{align*}
We have $g|_{\{t, \chi\}}(f_{0, \kappa}, f_{0, \kappa}) = -1$, $g|_{\{t, \chi\}}(f_{1, \kappa}, f_{1, \kappa}) = 1$, and $g|_{\{t, \chi\}}(f_{0, \kappa}, f_{1, \kappa}) = 0$. Since along $\gamma$ we have $f_{0, \kappa} = \dot{\gamma}$, it follows directly that $f_{0, \kappa}$ and $f_{1, \kappa}$ are parallely propagated along $\gamma$.

We also define a null frame, depending on $\kappa \geq 0$, by
\begin{equation}
\label{EqDefNullFrame}
\begin{split}
\ell_\kappa &:= f_{0,\kappa} + f_{1, \kappa} =-\frac{1}{a(t)} \big[ \sqrt{a(t)^2 + \kappa^2} + \kappa\big] \partial_t + \frac{1}{a(t)^2} \big[\sqrt{a(t)^2 + \kappa^2} + \kappa\big] \partial_\chi \\
\underline{\ell}_\kappa &:= f_{0,\kappa} - f_{1, \kappa} =-\frac{1}{a(t)} \big[ \sqrt{a(t)^2 + \kappa^2} - \kappa\big] \partial_t - \frac{1}{a(t)^2} \big[\sqrt{a(t)^2 + \kappa^2} - \kappa\big] \partial_\chi \;.
\end{split}
\end{equation}
Of course also $\ell_\kappa$ and $\underline{\ell}_\kappa$ are parallely propagated along $\gamma$.

\underline{\textbf{Setting up a family of loops along which holonomy will be computed:}}

We introduce a new $\hat{t}$ coordinate by $\hat{t}(t) := \int_0^t \frac{1}{a(t')} \, dt'$ and let $\hat{b} := \int_0^b \frac{1}{a(t')} \, dt'$. Thus, $dt = a(t) d\hat{t}$ and the metric on $M_{t, \chi}$ reads   $$ g|_{\{t,\chi\}} = a(t)^2 (-d\hat{t}^2 + d\chi^2) $$ in the new coordinates $(\hat{t}, \chi) \in (0, \hat{b})\times (0,\chi_0)$.
We define the null coordinates $v = \hat{t} + \chi$ and $u = \hat{t} - \chi$. Since $\gamma$ is timelike it follows that $ [-\mu, 0 ) \ni s \mapsto v(\gamma(s))$ and 
$[-\mu, 0) \ni s \mapsto u(\gamma(t))$ are strictly decreasing functions and thus \begin{equation}
\label{EqBoundV}v(\gamma(-\mu)) > \lim_{s \to 0} v(\gamma(s)) = \frac{\chi_0}{2}
\end{equation}
and $u(\gamma(s)) < \lim_{s \to 0} u(\gamma(s)) = -\frac{\chi_0}{2}$. If necessary, we now make $\mu >0$ even smaller such that the function $\chi$ is bounded away from $0$ and $\chi_0$ on the set $$\Big( \{-\frac{\chi_0}{2} \leq u \leq u(\gamma(-\mu))\}  \cap \{\frac{\chi_0}{2} \leq v \leq v(\gamma(-\mu))\} \Big) \setminus \{u = -\frac{\chi_0}{2}, v = \frac{\chi_0}{2}\}  \subset (0, \hat{b}) \times (0,\chi_0) \;, $$ see also Figure \ref{FigChartTChi}. Moreover, after making $\mu>0$ even smaller we can also ensure that\footnote{Recall that the notation $\subset \subset$ denotes compact inclusion.} $$\Big(\tilde{\varphi}\big(\tilde{\gamma}(-\mu)\big) + C^+_{\nicefrac{5}{8}}\Big) \cap \Big(\lim_{s \to 0}\tilde{\varphi}\big(\tilde{\gamma}(s)\big) + C^-_{\nicefrac{5}{8}}\Big) \subset\subset \Reps \;,$$ see also Figure \ref{FigBoundaryChart}.
\begin{figure}[h]
\centering
\begin{minipage}{.5\textwidth}
  \centering
 \def\svgwidth{9cm}
   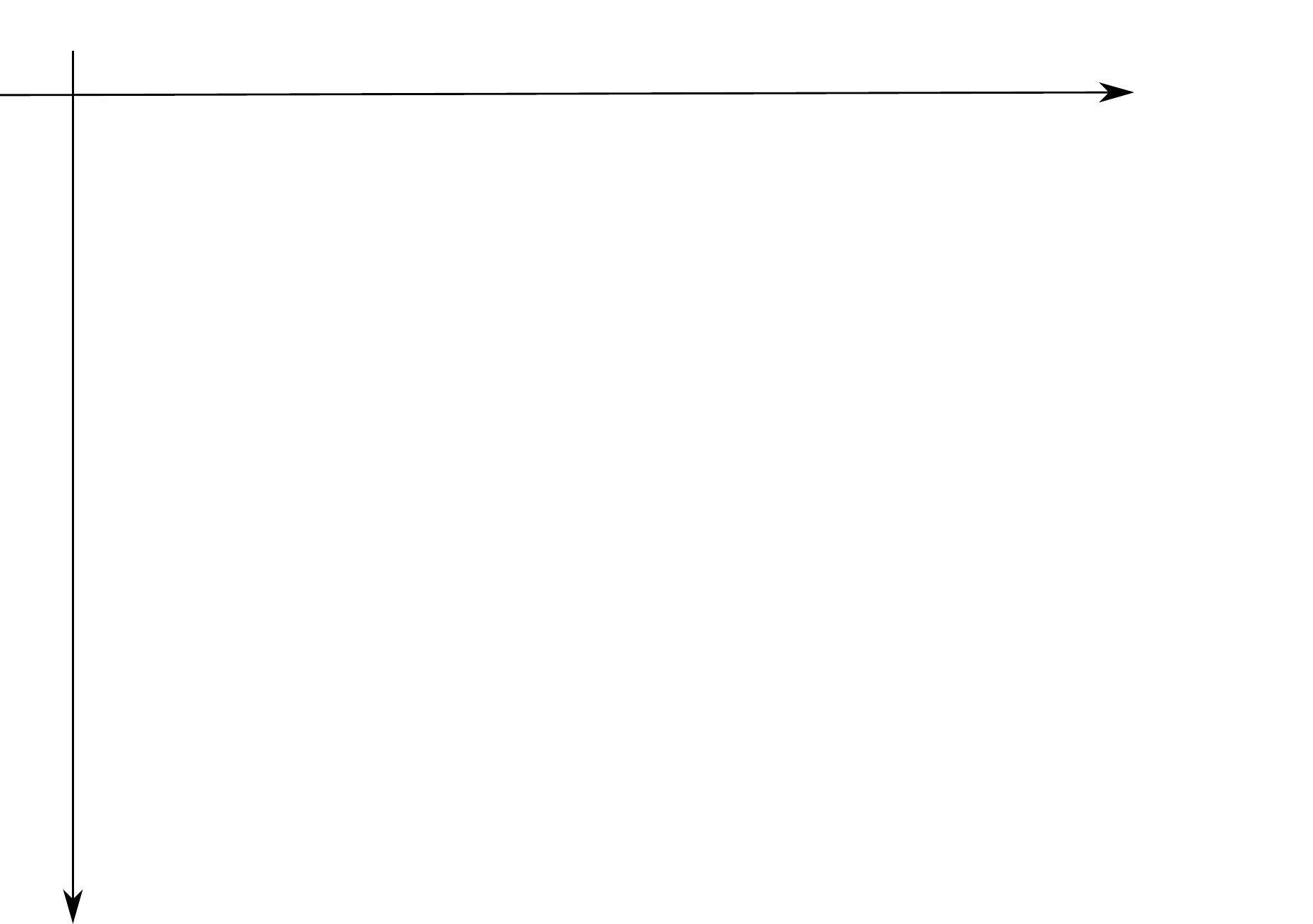
      \caption{The $\{t,\chi\}$-plane in $M$.} \label{FigChartTChi}
\end{minipage}%
\begin{minipage}{.5\textwidth}
  \centering
  \def\svgwidth{7cm}
  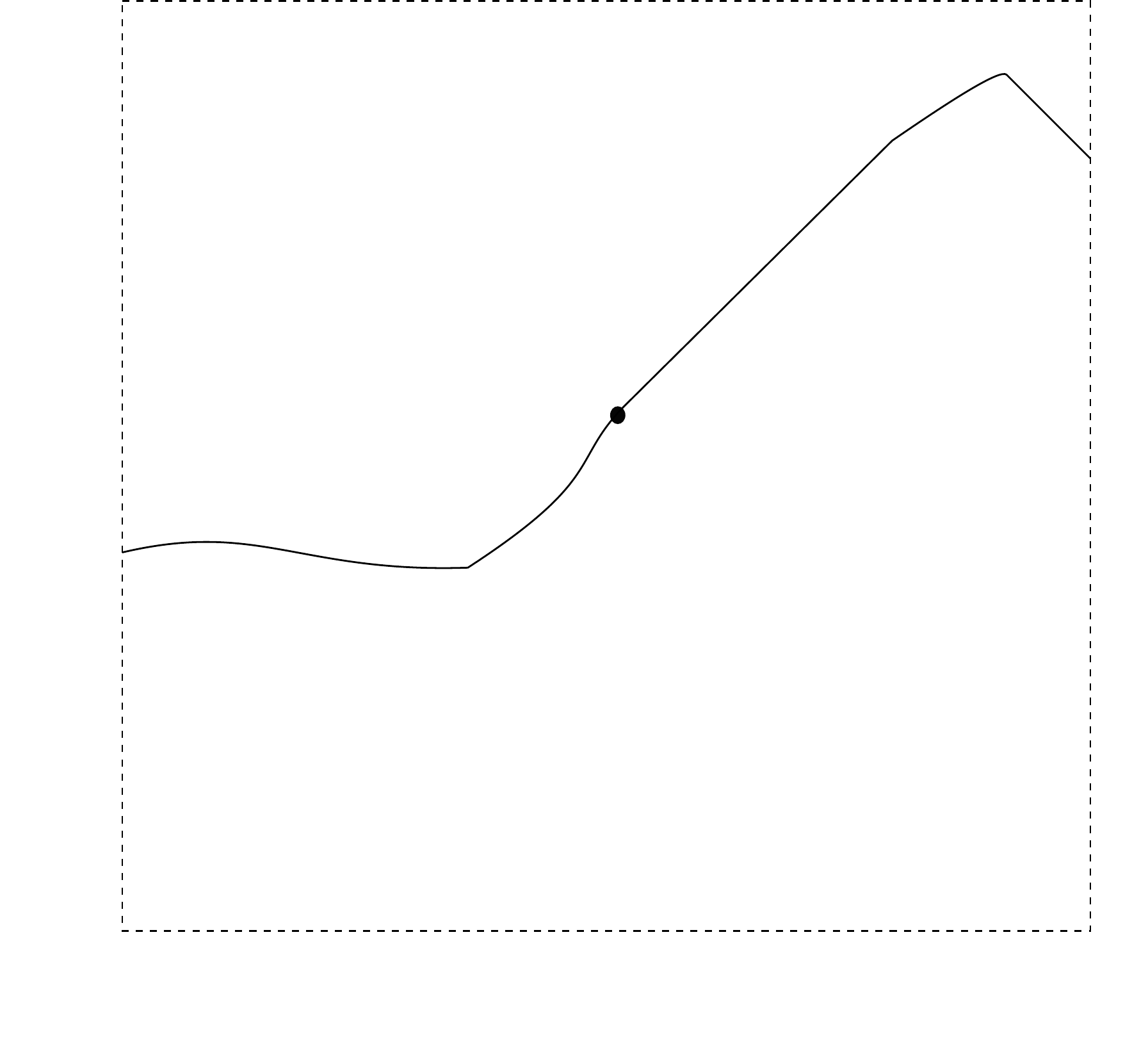
     \caption{The boundary chart in the extension $\tilde{M}$.} \label{FigBoundaryChart}
\end{minipage}
\end{figure}
Let now $-\mu < -\mu_0 < 0$. We construct a causal homotopy $\Gamma_{\mathrm{right}, \mu_0} : [-\mu, -\mu_0] \times [-\mu, -\mu_0] \to M_{t,\chi}$ with fixed endpoints of $\gamma|_{[-\mu, -\mu_0]}$ by
\begin{equation*}
\Gamma_{\mathrm{right}, \mu_0}  (s;r) = \begin{cases} \big( u(\gamma(2s + \mu)), v(\gamma(-\mu))\big) \qquad &\textnormal{ for } - \mu \leq s \leq \frac{-\mu + r}{2} \\
\big(u(\gamma(r)), v(\gamma(2s - r))\big) \qquad &\textnormal{ for } \frac{-\mu + r}{2} \leq s \leq r \\
\gamma(s) &\textnormal{ for } s \geq r \;,
\end{cases}
\end{equation*}
where we have used $(u,v)$-coordinates in the definition. Then for each $r \in [-\mu, -\mu_0]$ the curves $ s \mapsto \Gamma_{\mathrm{right}, \mu_0} (s ; r)$ are future directed causal curves with $\Gamma_{\mathrm{right}, \mu_0} (-\mu;r) = \gamma(-\mu)$ and $\Gamma_{\mathrm{right}, \mu_0} (-\mu_0; r) = \gamma(-\mu_0)$, and we have $\Gamma_{\mathrm{right}, \mu_0} (s; -\mu) = \gamma|_{[-\mu, -\mu_0]}(s)$. See also Figure \ref{FigChartTChi}.

We also define the broken null geodesics $\sigma_{\mathrm{right}, \mu_0} : [-\mu, -\mu_0] \to M_{t,\chi}$ by $\sigma_{\mathrm{right}, \mu_0}(s) := \Gamma_{\mathrm{right}, \mu_0}(s; -\mu_0)$. The family of loops, based at $\gamma(-\mu)$ and depending on $\mu_0$, is then given by $\sigma_{\mathrm{right}, \mu_0}$ followed by the reversal of $\gamma|_{[-\mu, -\mu_0]}$.
\newline

\underline{\textbf{Step 3: Uniform boundedness of holonomy by virtue of the $C^{0,1}_{\loc}$-extension.}}

Recalling the inclusion relations \eqref{EqInclusionRelationFuture} we note that the assumptions of Lemma \ref{LemCausalHomotopy} with $\Gamma = \Gamma_{\mathrm{right}, \mu_0}$ are met and thus we conclude that $\tilde{\Gamma}_{\mathrm{right}, \mu_0} := \iota \circ \Gamma_{\mathrm{right}, \mu_0}$ maps into $\tilde{U}$ -- and in fact into the region below the graph of $f$. In particular $\tilde{\sigma}_{\mathrm{right}, \mu_0} := \iota \circ \sigma_{\mathrm{right}, \mu_0}$ maps into the region of $\tilde{U}$ that is below the graph of $f$ for all $-\mu < -\mu_0 < 0$. Using that the $x_0$-coordinate on $\tilde{U}$ is a time-coordinate, we can reparametrise $\tilde{\gamma}|_{[-\mu, -\mu_0]}$ and $\tilde{\sigma}_{\mathrm{right}, \mu_0}$ by $x_0$. The size of the domain of definition of the curves in this new parametrisation is obviously uniformly bounded in $\mu_0$ by $2\varepsilon_0$. Moreover, as in \eqref{EqUniformBoundCausalCurve} we also obtain a uniform bound on the coordinate velocity of the curves. It now follows from Lemma \ref{LemBoundParallelTransport} that there is a constant $C>0$ such that 
\begin{equation}\label{EqUnifBoundPT}
|| \big(P^{-1}_{\tilde{\gamma}|_{[-\mu,-\mu_0]} }\circ P_{\tilde{\sigma}_{\mathrm{right}, \mu_0}} \big)(X) ||_{\R^{4}} \leq C \cdot ||X||_{\R^4} \;,
\end{equation}
where $X \in T_{ \tilde{\gamma}(-\mu) }\tilde{U}$, $|| \cdot ||_{\R^4}$ denotes the Euclidean norm induced by the coordinate chart $\tilde{\varphi}$, and $P_\tau$ denotes parallel transport along a curve $\tau$.
\newline

\underline{\textbf{Step 4: Showing that the holonomy along those loops in $M$ is unbounded.}}

We now compute the parallel transport of $\underline{\ell}_\kappa$ along $\sigma_{\mathrm{right}, \mu_0}$. For this, let us set $e_0 := \partial_t$ and $e_1 := \frac{1}{a(t)} \partial_\chi$. It then follows from \eqref{EqChristoffel} that
\begin{equation*}
\nabla_{e_0} e_0 = 0, \qquad \nabla_{e_0} e_1 = 0, \qquad \nabla_{e_1} e_0 = \frac{\dot{a}(t)}{a(t)} e_1, \qquad \nabla_{e_1} e_1 = \frac{\dot{a}(t)}{a(t)} e_0 \;.
\end{equation*}
Defining $\underline{L} := e_0 + e_1$ and $L := e_0 - e_1$ this gives
\begin{equation}\label{EqPTL}
\begin{aligned}
\nabla_L \Big( \frac{L}{a(t)} \Big) &= 0\;, \qquad \qquad &&\nabla_{\underline{L}}\Big( \frac{\underline{L}}{a(t)}\Big) =0\;, \\
\nabla_{L}\big(a(t) \underline{L}\big) &=0\;, \qquad \qquad &&\nabla_{\underline{L}} \big(a(t) L\big) =0 \;.
\end{aligned}
\end{equation}
We furthermore have $L = - \frac{\sqrt{a(t)^2 + \kappa^2} - \kappa}{a(t)} \ell_\kappa$ and $\underline{L} = - \frac{\sqrt{a(t)^2 + \kappa^2} + \kappa}{a(t)} \underline{\ell}_\kappa$, which together with \eqref{EqPTL} finally gives
\begin{equation}
\label{EqPTell}
\begin{aligned}
\nabla_{\ell_\kappa} \Big(\frac{\sqrt{a(t)^2 + \kappa^2} - \kappa}{a(t)^2} \ell_\kappa\Big) &= 0 \;, \qquad \qquad &&\nabla_{\underline{\ell}_\kappa} \Big( \frac{\sqrt{a(t)^2 + \kappa^2} + \kappa}{a(t)^2} \underline{\ell}_\kappa\Big) = 0\;, \\
\nabla_{\ell_\kappa} \Big( \big[ \sqrt{a(t)^2 + \kappa^2} + \kappa\big] \underline{\ell}_\kappa\Big) &=0 \;, \qquad \qquad &&\nabla_{\underline{\ell}_{\kappa}} \Big(\big[ \sqrt{a(t)^2 + \kappa^2} - \kappa\big]\ell_\kappa\Big) =0 \;.
\end{aligned}
\end{equation}
Let $\hat{t}_{\mathrm{right},\mu_0}$ be the value of the $\hat{t}$-coordinate at the point where the tangent of $\sigma_{\mathrm{right}, \mu_0}$ changes from $\ell_\kappa$ to $\underline{\ell}_\kappa$, and let $\hat{t}_{\mu_0} = \hat{t}(\gamma(-\mu_0))$ and $\hat{t}_{\mu} = \hat{t}(\gamma(-\mu))$, see also Figure \ref{FigHolonomyTChi}. We denote the value of the $t$-coordinate at $\hat{t}_{\mathrm{right},\mu_0}$ by $t_{\mathrm{right}, \mu_0}$, etc.. 
\begin{figure}[h]
\centering
 \def\svgwidth{7cm}
   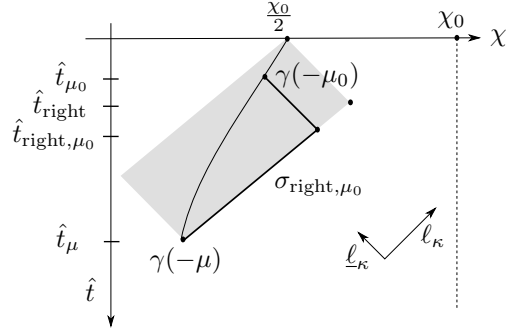
      \caption{Parallel transport along $\sigma_{\mathrm{right},\mu_0}$ in $M$.} \label{FigHolonomyTChi}
\end{figure}
Recall that the $\{t,\chi\}$-plane in $M$ is totally geodesic and thus parallel transport in $M$ of vectors tangent to $M_{t,\chi}$ may be computed in $M_{t,\chi}$. Using \eqref{EqPTell} we now obtain that the parallel transport of the vector $\underline{\ell}_\kappa$ from $\gamma(-\mu)$ along $\sigma_{\mathrm{right},\mu_0}$ to the point where the tangent of $\sigma_{\mathrm{right},\mu_0}$  changes from $\ell_\kappa$ to $\underline{\ell}_\kappa$,  is given by $$\frac{\sqrt{a(t_{\mathrm{right}, \mu_0})^2 + \kappa^2} + \kappa}{\sqrt{a(t_\mu)^2 + \kappa^2} + \kappa} \cdot \underline{\ell}_\kappa \;.$$ Continuing the parallel transport along $\sigma_{\mathrm{right},\mu_0}$ to the point $\gamma(-\mu_0)$ we obtain
\begin{equation*}
P_{\sigma_{\mathrm{right}, \mu_0}} (\underline{\ell}_\kappa) =  \frac{\sqrt{a(t_{\mu_0})^2 + \kappa^2} + \kappa}{a(t_{\mu_0})^2} \cdot \frac{a(t_{\mathrm{right}, \mu_0})^2}{\sqrt{a(t_\mu)^2 + \kappa^2} + \kappa} \cdot \underline{\ell}_\kappa \;.
\end{equation*}
Note that we have $\hat{t}_{\mathrm{right},\mu_0} > \hat{t}_{\mathrm{right}} := \frac{1}{2}\big(v(\gamma(-\mu)) - \frac{\chi_0}{2}\big) >0$ by \eqref{EqBoundV}, and thus $t_{\mathrm{right}, \mu_0}$ is uniformly bounded away from $0$ in $\mu_0$.
Thus for $\mu_0 \to 0$ we obtain $P_{\sigma_{\mathrm{right}, \mu_0}} (\underline{\ell}_\kappa) \sim \frac{\sqrt{a(t_{\mu_0})^2 + \kappa^2} + \kappa}{a(t_{\mu_0})^2} \cdot \underline{\ell}_\kappa$. Recalling that $\underline{\ell}_\kappa$ is parallel along $\gamma$ we finally find
\begin{equation*}
\big(P^{-1}_{\gamma|_{[-\mu_0,-\mu]}} \circ P_{\sigma_{\mathrm{right}, \mu_0}}\big) (\underline{\ell}_\kappa) \sim \frac{\sqrt{a(t_{\mu_0})^2 + \kappa^2} + \kappa}{a(t_{\mu_0})^2} \cdot \underline{\ell}_\kappa 
\end{equation*}
for $\mu_0 \to 0$. This, however, is in contradiction to \eqref{EqUnifBoundPT} derived in Step 3: We set $X:= \iota_* \Big(\underline{\ell}_\kappa|_{\gamma(-\mu)} \Big)= \underline{\ell}_\kappa^i \partial_i $ with $\underline{\ell}_\kappa^i \in \R$ and thus obtain 
\begin{equation*}
\begin{split}
C ||\iota_* \Big(\underline{\ell}_\kappa|_{\gamma(-\mu)}\Big)||_{\R^4} &\geq  ||\big(P^{-1}_{\tilde{\gamma}|_{[-\mu,-\mu_0]} }\circ P_{\tilde{\sigma}_{\mathrm{right}, \mu_0}} \big)(\iota_* \underline{\ell}_\kappa)||_{\R^4} \\
&= ||\iota_* \Big( \big(P^{-1}_{\gamma|_{[-\mu_0,-\mu]}} \circ P_{\sigma_{\mathrm{right}, \mu_0}}\big) (\underline{\ell}_\kappa)\Big)||_{\R^4} \\
&\sim \frac{\sqrt{a(t_{\mu_0})^2 + \kappa^2} + \kappa}{a(t_{\mu_0})^2} || \iota_* \Big(\underline{\ell}_\kappa|_{\gamma(-\mu)}\Big)||_{\R^4} \;.
\end{split}
\end{equation*}
Recall that we have $\kappa \geq 0$. For $\kappa = 0$ we have $\frac{\sqrt{a(t_{\mu_0})^2 + \kappa^2} + \kappa}{a(t_{\mu_0})^2} \sim \frac{1}{a(t_{\mu_0})}$ for $\mu_0 \to 0$, and for $\kappa >0$ we have $\frac{\sqrt{a(t_{\mu_0})^2 + \kappa^2} + \kappa}{a(t_{\mu_0})^2} \sim \frac{1}{a(t_{\mu_0})^2}$ for $\mu_0 \to 0$.
This concludes the proof of Theorem \ref{ThmFLRWPast}.
\end{proof}

\begin{remark} \label{RemFLRW}
\begin{enumerate}
\item Note that the parallel transport map $\big(P^{-1}_{\gamma|_{[-\mu_0,-\mu]}} \circ P_{\sigma_{\mathrm{right}, \mu_0}}\big)$ corresponds to an infinite boost in the direction of $-f_{1,\kappa}$ for $\mu_0 \to 0$.

\item Analogously to $\sigma_{\mathrm{right}, \mu_0}$ one can define $$\sigma_{\mathrm{left}, \mu_0} (s) := \begin{cases} \big(u(\gamma(-\mu), v(\gamma(2s + \mu))\big) \qquad &\textnormal{ for } - \mu \leq s \leq \frac{-\mu - \mu_0}{2} \\
\big(u(\gamma(2s +\mu_0)), v(\gamma(-\mu_0))\big) \qquad &\textnormal{ for } \frac{-\mu - \mu_0}{2} \leq s \leq -\mu_0 \;. \end{cases} $$
Using \eqref{EqPTell} one finds that the parallel transport of $\ell_\kappa$ along $\sigma_{\mathrm{left}, \mu_0}$ is given by 
\begin{equation}\label{EqPTLEFT}
P_{\sigma_{\mathrm{left}, \mu_0}} \ell_\kappa = \frac{\sqrt{a(t_\mu)^2 + \kappa^2} - \kappa}{a(t_{\mu_0})^2} \cdot \frac{a(t_{\mathrm{left}, \mu_0})^2}{\sqrt{a(t_\mu)^2 + \kappa^2} - \kappa} \cdot \ell_\kappa\;,
\end{equation}
where $t_{\mathrm{left}, \mu_0}$ is defined analogously to $t_{\mathrm{right}, \mu_0}$ and is also bounded uniformly in $\mu_0$ away from $0$. For $\kappa = 0$ we have $\frac{\sqrt{a(t_\mu)^2 + \kappa^2} - \kappa}{a(t_{\mu_0})^2}  \sim \frac{1}{a(t_{\mu_0})}$ for $\mu_0 \to 0$ and thus \eqref{EqPTLEFT} gives again an infinite boost. However, for $\kappa >0$ we observe that $\sqrt{a(t_{\mu_0})^2 + \kappa^2} = \kappa + \frac{a(t_{\mu_0})^2}{2\kappa} + \mathcal{O}\big(a(t_{\mu_0})^4\big) $ and thus \eqref{EqPTLEFT} shows $P_{\sigma_{\mathrm{right}, \mu_0}} \ell_\kappa \sim \ell_\kappa$ for $\mu_0 \to 0$, i.e., we obtain a finite boost in the limit. This is intimately related to the following point:
\item There are three standards of finite energy at the singularity for observers following timelike geodesics in the $\{t,\chi\}$-plane, given by $\kappa >0, \kappa <0, \kappa = 0$:\footnote{The computations below show that if an inertial observer with say $\kappa >0$ approaches the singularity, then he measures the velocity/energy of other inertial observers with $\kappa >0$ crossing his worldline  as finite, while inertial observers with $\kappa = 0$ or $\kappa <0$ seem to reach infinite velocities/energies. Analogously for inertial observers with $\kappa = 0$ or $\kappa <0$.}
 Recall that the vector field $f_{0,\kappa} = -\frac{\sqrt{a(t)^2 + \kappa^2}}{a(t)} \partial_t + \frac{\kappa}{a(t)^2} \partial_\chi $ is the affine velocity vector field along geodesics with parameter $\kappa$. It is now an easy computation to verify that for $\kappa, \rho>0$ we have for $t \to 0$ 
\begin{align*}
g(f_{0,\kappa}, f_{0, 0}) &\sim g(f_{0,-\kappa}, f_{0,0}) \sim - \frac{1}{a(t)} \\
g(f_{0,\kappa}, f_{0, -\rho}) &\sim -\frac{1}{a(t)^2} \\
g(f_{0,\kappa}, f_{0, \rho}) &\sim g(f_{0,-\kappa}, f_{0, -\rho}) \sim -1 \;.
\end{align*}
\end{enumerate}
\end{remark}

\begin{remark} \label{RemSchwarz}
We remark that the Schwarzschild singularity is also a holonomy singularity: Consider the interior $(M,g)$ of the Schwarzschild solution, where $M = \R \times (0,2m) \times \mathbb{S}^2$ with standard $(t,r,\theta, \varphi)$-coordinates and $g = -(1 - \frac{2m}{r}) \, dt^2 + \frac{1}{1 - \frac{2m}{r}} \, dr^2 + r^2 (d\theta^2 + \sin^2 \theta \, d\varphi^2) $. We define an orthonormal basis by $e_0 := -(\frac{2m}{r} - 1)^{\nicefrac{1}{2}} \, \partial_r$, $e_1:= (\frac{2m}{r} - 1)^{-\nicefrac{1}{2}} \, \partial_t$, $e_2:= \frac{1}{r} \, \partial_\theta$, $e_3 := \frac{1}{r \sin\theta} \, \partial_\varphi$. Note that this frame is parallel along $e_0$ and that the $\{t,r\}$-plane as well as the $\{r, \theta\}$-plane are totally geodesic. We also have future particle horizons. For simplicity we restrict our consideration here to an observer with affine velocity vector $e_0$. 

Let us consider first a null basis $L = e_0 + e_1$ and $\underline{L} = e_0 - e_1$ for the $\{t,r\}$-plane. A straightforward computation now yields $\nabla_L( \frac{\sqrt{r}}{\sqrt{2m - r}} L) = 0$ and $\nabla_L( \frac{\sqrt{2m - r}}{\sqrt{r}}\underline{L}) = 0$, which shows that light propagating in the $L$-direction is infinitely red-shifted when approaching $r=0$ as observed by our observer, as expected from the infinite expansion in the $t$-direction. Note that this is different to the situation of the cosmological singularities we just discussed, where we encounter an infinite blue-shift of radiation approaching the singularity. However, now the parallel transport of $\underline{L}$ along $L$ blows up when approaching $r=0$. 

Considering the situation in the $\{r,\theta\}$-plane, we set $\ell := e_0 + e_2$ and $\underline{\ell} := e_0 - e_2$. A computation gives $\nabla_\ell (\frac{1}{r} \ell) = 0$, which shows that we have an infinite blue-shift of radiation travelling in $\ell$-direction when approaching $r=0$ as observed by our observer, in qualitatively the same manner as in the cosmological singularities discussed earlier.

With a bit of additional work one can now extend the method of proof of Theorem \ref{ThmFLRWPast} also to the Schwarzschild singularity, showing $C^{0,1}_{\loc}$-inextendibility. But of course this statement is strictly weaker than the $C^0$-inextendibility proven in \cite{Sbie15}. However, the blow-up of the local causal holonomy is of independent interest.
\end{remark}

Theorem \ref{ThmFLRWFuture} and Theorem \ref{ThmFLRWPast} together directly give the following
\begin{corollary} \label{CorFLRW}
Let $(M,g)$ be a cosmological warped product spacetime with $b=\infty$ and let the scale factor $ a : (0,\infty) \to (0,\infty)$ satisfy  $\lim_{t \to 0} a(t) = 0$, $\int_1^\infty \frac{a(t)}{ \sqrt{a(t)^2 + 1}}\, dt = \infty$ and $\int_0^1 \frac{1}{a(t')} \, dt' < \infty$. Then $(M,g)$ is $C^{0,1}_{\loc}$-inextendible.
\end{corollary}

The physically most interesting examples of cosmological warped product spacetimes, and thus the physically most relevant application of Theorem \ref{ThmFLRWFuture}, Theorem \ref{ThmFLRWPast}, and Corollary \ref{CorFLRW}, are to the class of isotropic and homogeneous cosmological models, known as \emph{FLRW spacetimes}. Here, $(\overline{M}_K, \overline{g}_K)$ is a $3$-dimensional complete Riemannian manifold of constant sectional curvature $K$. 
By redefining the scale factor function $a(t)$ one may restrict to the three cases $K = -1, 0, +1$. In all three cases one can locally introduce polar normal coordinates\footnote{See for example \cite{LeeRiem}, Chapter 10.} $(\chi, \theta, \varphi) \in (0,\chi_0) \times \mathbb{S}^2$ on $\overline{M}_K$ such that the $3$-metric in these coordinates reads $$\overline{g}_K = d\chi^2 + f_K(\chi)^2(d\theta^2 + \sin^2\theta \, d\varphi^2) \quad \textnormal{ with }\quad  \begin{cases} f_K(\chi) = \sinh \chi \qquad &\textnormal{ for } K = -1 \\
f_K(\chi) = \chi &\textnormal{ for } K = 0 \\
f_K(\chi) = \sin \chi &\textnormal{ for } K = +1 \;. \end{cases}$$
Thus, in the coordinates $(t,\chi,\theta, \varphi)$ the spacetime metric reads
$$g  = -dt^2 + a(t)^2\big[ d\chi^2 + f_K(\chi)^2(d\theta^2 + \sin^2\theta \, d\varphi^2)\big]\;.$$
Note that various topologies are possible for $\overline{M}_K$. The classical, simply connected choices are $\overline{M}_{-1} = \overline{M}_0 = \R^3$ with almost global coordinates $(\chi, \theta, \varphi) \in (0,\infty) \times (0,\pi) \times (0,2\pi)$, and  $\overline{M}_1 = \mathbb{S}^3$ with almost global coordinates $(\chi, \theta, \varphi) \in (0,\pi) \times (0,\pi) \times (0,2\pi)$.

\section{Spherically symmetric weak null singularities}

\subsection{General theorem} \label{SecGT}

We consider the following class of spacetimes $(M,g)$:
Let $$M := \underbrace{(-\infty,0) \times (-\infty, 0)}_{=: Q} \times \mathbb{S}^2 \qquad \textnormal{ and } \qquad \overline{M} :=\underbrace{(-\infty, 0] \times (-\infty, 0]}_{=: \overline{Q}} \times \mathbb{S}^2 \;,$$
with standard $u,v$ coordinates on the first two factors of $M$ and $\overline{M}$, and let
$$g = -\frac{\Omega^2}{2} \big( du \otimes dv + dv \otimes du\big) + r^2 \mathring{\gamma} \;,$$
where $\mathring{\gamma}$ denotes the standard round metric on $\mathbb{S}^2$. For the metric components we assume that
\begin{itemize}
\item $\Omega : Q \to (0,\infty)$  is smooth and extends continuously to $\overline{Q}$ as a positive function.\\[-15pt]
\item $r : Q \to (0,\infty)$ is smooth and extends continuously to $\overline{Q}$ as a positive function.\\[-15pt]
\item  $\lim_{v \to 0} \partial_v r (u,v) = -\infty $ for all $u \in (-\infty, 0)$. \\[-15pt]
\item For each $u_0 \in (-\infty, 0)$ there exists a $v_0(u_0)  \in (-\infty, 0)$ such that $\partial_u r (u_0,v) <0$ for all $v \in [v_0(u_0), 0)$.
\end{itemize}
We fix the time-orientation by stipulating that $\partial_u + \partial_v$ is future directed. Clearly $(M,g)$ is globally hyperbolic. See also Figure \ref{FigHomotopyNull} on page \pageref{FigHomotopyNull}  for a Penrose diagram. By assumption $(\overline{M}, g)$ is a future $C^0$-extension of $(M,g)$. The hypersurface $\{v = 0\}$ in $\overline{M}$ is known as a weak null singularity.

Our first result shows that there is no $C^{0,1}_{\loc}$-extension through $\{v=0\}$. The precise statement is
\begin{theorem}\label{ThmOneNullSing}
Let $(M,g)$ be as above. Then there exists no $C^{0,1}_{\loc}$-extension $\iota : M \hookrightarrow \tilde{M}$ with the property that there is a future directed future inextendible timelike geodesic $\gamma : [-1,0) \to M$ with $\lim_{s \to 0} v(\gamma(s)) = 0$ and such that $\lim_{s \to 0} (\iota \circ \gamma)(s)$ exists in $\tilde{M}$.
\end{theorem}
\begin{remark}
\begin{enumerate}
\item Note that the theorem does not require the Hawking mass to go to infinity for $v \to 0$.
\item Also note that the theorem in particular rules out the existence of $C^{0,1}_{\loc}$-extensions $\iota : M \hookrightarrow \tilde{M}$ with the property that there is a future directed future inextendible timelike geodesic $\gamma : [-1,0) \to M$ which leaves into $\tilde{M} \setminus \iota(M)$ through $\{u=0\} \cap \{v=0\}$\footnote{That means $\lim_{s \to 0} v(\gamma(s)) = 0 = \lim_{s \to 0} u(\gamma(s))$ and $\lim_{s \to 0} (\iota \circ \gamma)(s)$ exists in $\tilde{M}$.}, although we have not made any blow-up assumption at $\{u=0\} \cap \{v=0\}$. This is due to the fact, established in the proof of Theorem \ref{ThmOneNullSing}, that any $C^0$-extension that extends through $\{u=0\} \cap \{v=0\}$ also necessarily extends through a bit of $ \{ -\infty < u < 0\} \cap \{v=0\}$, see Step 2 and, in particular, Step 2.2 of the proof.
\item The Theorem remains valid for a much larger class of spacetimes. For the proof, as written, to go through it suffices that in addition to the first two bullet points one assumes that for each $u_0 \in (-\infty, 0)$ there exists a sequence $v_n \in (-\infty, 0)$ with $v_n \to 0$ for $n \to \infty$ such that $(\rd_u r \cdot \rd_v r) (u_0, v_n) >0$ and, moreover, such that $|\rd_u r (u_0, v_n)| \to \infty$ or $|\rd_v r(u_0, v_n)| \to \infty$ for $n \to \infty$. We would expect that one can in addition remove the assumption on the sign of $\rd_u r \cdot \rd_v r$, but this would require a modification of the proof. All of the physically interesting applications known to the author are however covered by the assumptions made in the theorem.
\end{enumerate}
\end{remark}

The next theorem covers the case that $\{u = 0\} \subseteq \tilde{M}$ is also a weak null singularity. It is more or less a direct consequence of Theorem \ref{ThmOneNullSing}.
\begin{theorem} \label{ThmTwoNullSing}
Let $(M,g)$ as above satisfy in addition 
\begin{itemize}
\item $\lim_{u \to 0} \partial_u r(u,v) = -\infty$ for all $v \in (-\infty,0)$.\\[-15pt]
\item For each $v_0 \in (-\infty, 0) $ there exists a $ u_0(v_0) \in (-\infty, 0)$ such that $\partial_v r (u,v_0)<0$ for all $u \in [u_0(v_0),0)$.
\end{itemize}
Then $(M,g)$ is future $C^{0,1}_{\loc}$-inextendable.
\end{theorem}

\begin{proof}[Proof of Theorem \ref{ThmTwoNullSing}:]
Assume $\iota : M \hookrightarrow \tilde{M}$ is a future $C^{0,1}_{\loc}$-extension. Then by Proposition \ref{PropGeodesicBoundaryChart} there is a future directed and future inextendible timelike geodesic $\gamma : [-1,0) \to M$ such that $\iota \circ \gamma$ has a future limit point in $\tilde{M}$. The future inextendibility of $\gamma$ implies that we have $\lim_{s \to 0} v(\gamma(s)) = 0$ or $\lim_{s \to 0} u(\gamma(s)) = 0$ (or in fact both). We then obtain a contradiction to the statement of Theorem \ref{ThmOneNullSing} (or its dual where $u$ and $v$ are interchanged).
\end{proof}

Before we begin with the proof of Theorem \ref{ThmOneNullSing} let us collect the expressions for the Christoffel symbols of $(M,g)$ in the coordinates $(u,v, x^A)$, where $x^A$ denotes a coordinate system on $\mathbb{S}^2$. A direct computation gives
\begin{equation}
\label{EqChristoffelNullSing}
\begin{aligned}
\Gamma^u_{uu} &= \partial_u \log \Omega^2 \ \qquad \qquad \hspace{2cm} &&\Gamma^u_{AB} = \frac{2}{\Omega^2} r \partial_vr  \cdot \mg_{AB} \\
\Gamma^u_{uv} &= \Gamma^u_{vv} = \Gamma^u_{Au} = \Gamma^u_{Av} = 0 \\[5mm]
\Gamma^v_{vv} &= \partial_v \log \Omega^2 \qquad \quad &&\Gamma^v_{AB} = \frac{2}{\Omega^2} r \partial_ur  \cdot \mg_{AB} \\
\Gamma^v_{vu} &= \Gamma^v_{uu} = \Gamma^v_{Av} = \Gamma^v_{Au} = 0 \\[5mm]
\Gamma^A_{Bu} &= \frac{1}{r} \partial_u r \cdot \delta^A_{\; \; B} &&\Gamma^A_{Bv} = \frac{1}{r} \partial_v r \cdot \delta^A_{\; \; B} \qquad \qquad &&&\Gamma^A_{BC} = \mathring{\Gamma}^A_{BC} \\
\Gamma^A_{uu} &= \Gamma^A_{uv} = \Gamma^A_{vv} = 0 \;,
\end{aligned}
\end{equation}
where $\mathring{\Gamma}^A_{BC} $ denotes the Christoffel symbols of $(\mathbb{S}^2, \mg)$.

\begin{proof}[Proof of Theorem \ref{ThmOneNullSing}:]
The proof is by contradiction. 

\underline{\textbf{Step 1:}} Assume $\iota : M \hookrightarrow \tilde{M}$ is a $C^{0,1}_{\loc}$-extension and let $\gamma : [-\mu,0) \to M$ be an affinely\footnote{It follows easily from the $C^0$-extension $(\overline{M}, g)$ that a timelike geodesic as in Theorem \ref{ThmOneNullSing} has to be future incomplete. We can thus assume without loss of generality that it is affinely parametrised on the interval $[-\mu,0)$.} parametrised future directed timelike geodesic with $\lim_{s \to 0} v(\gamma(s)) = 0$ and $\lim_{s \to 0} \tilde{\gamma}(s) = \tilde{p} \in \tilde{M}$, where $\tilde{\gamma} := \iota \circ \gamma$. Then by Proposition \ref{PropBoundaryChart} there exists a chart $\tilde{\varphi} : \tilde{U} \to(-\varepsilon_0, \varepsilon_0) \times  (-\varepsilon_1, \varepsilon_1)^{d} =: R_{\varepsilon_0, \varepsilon_1}$ as in Proposition \ref{PropBoundaryChart} with $\delta >0$ so small that all vectors in $C^+_{\nicefrac{5}{6}}$ are future directed timelike, all vectors in $C^-_{\nicefrac{5}{6}}$ are past directed timelike, and all vectors in $C^c_{\nicefrac{5}{8}}$ are spacelike. After making the chart slightly smaller if necessary, we can also assume that in this chart the Lorentzian metric $\tilde{g}$ on $\tilde{U}$ satisfies a global Lipschitz condition
\begin{equation*}
|\tilde{g}_{\mu \nu}(\tilde{x}) - \tilde{g}_{\mu \nu}(\tilde{y})| \leq \Lambda ||\tilde{x} - \tilde{y}||_{\R^{d+1}} 
\end{equation*}
for all $\tilde{x},\tilde{y} \in \Reps$, where $\Lambda >0$ is a constant. In the region $\{(\tilde{x}_0,\underline{\tilde{x}}) \in \Reps \; | \: \tilde{x}_0 < f(\underline{\tilde{x}})\}$ below the graph of $f$ the metric components $\tilde{g}_{\mu \nu}$ are smooth and satisfy the bounds $|\partial_\kappa \tilde{g}_{\mu \nu}| \leq \Lambda$. Let $\mu>0$ be so small that $\tilde{\gamma}$ is contained in $\tilde{U}$ and such that $$\Big(C^+_{\nicefrac{5}{8}} + \tilde{\varphi}(\tilde{\gamma}(-\mu))\Big) \cap  C^-_{\nicefrac{5}{8}} \subset \subset \Reps \;.$$
\newline

\underline{\textbf{Step 2:}} \hspace{5mm} \underline{\textbf{Step 2.1:}} We construct a timelike curve leaving $M$ that lies in the $\{u,v\}$-plane.

Using the spherical symmetry of $(M,g)$ we can choose standard $(\theta, \varphi)$ coordinates on $\mathbb{S}^2$ such that $\gamma$ lies in the $\{\theta = \frac{\pi}{2}\}$-plane. Then $\gamma(s) = \big(\gamma^u(s), \gamma^v(s), \frac{\pi}{2}, \gamma^\varphi(s)\big)$ satisfies \begin{equation}
\label{EqG}-1 = -\Omega^2 \dot{\gamma}^u \dot{\gamma}^v + r^2 \big(\dot{\gamma}^\varphi\big)^2 \qquad \qquad \textnormal{ and } \qquad \qquad \R \ni \kappa = g(\partial_\varphi, \dot{\gamma}) = r^2 \dot{\gamma}^\varphi \;.
\end{equation}
We can choose the $\varphi$ coordinate such that we have $\kappa \geq 0$. Since $\gamma$ is future directed timelike we have $\dot{\gamma}^u, \dot{\gamma}^v >0$, and thus $\lim_{s \to 0} \gamma^u(s) \in (-\infty, 0]$ exists. The continuous extension of $r$ then implies that $r_\infty := \lim_{s \to 0} r(\gamma(s)) >0$ also exists. In particular $r$ is bounded away from $0$ uniformly along $\gamma$. It then follows from $\dot{\gamma}^\varphi = \frac{\kappa}{r^2}$ that $\lim_{s \to 0} \gamma^\varphi(s) =: \varphi_\infty$ also exists and, without loss of generality, we can assume that $\varphi_\infty = \pi$.

If $\kappa = 0$ then we are already done. So let us assume $\kappa >0$. It follows from \eqref{EqG} that $-1 = -\Omega^2 \dot{\gamma}^u \dot{\gamma}^v + \frac{\kappa^2}{r^2}$. 
We now accelerate  the curve slightly in the $\varphi$-direction  such that it is still timelike. For $s \in [-\mu, 0)$ and $\varepsilon>0$ we have
\begin{equation}
\label{EqAcc}
\Big{|} \frac{(\kappa + \varepsilon)^2}{r^2(\gamma(s))} - \frac{\kappa^2}{r^2(\gamma(s))} \Big{|} \leq \Big{|} \frac{(\kappa + \varepsilon)^2}{r^2(\gamma(s))} - \frac{(\kappa + \varepsilon)^2}{r^2_\infty} \Big{|} + \Big{|} \frac{(\kappa + \varepsilon)^2}{r^2_\infty} - \frac{\kappa^2}{r_\infty^2}\Big{|} + \Big{|} \frac{\kappa^2}{r_\infty^2} - \frac{\kappa^2}{r^2(\gamma(s))} \Big{|} \;.
\end{equation}
Choosing first $\varepsilon>0$ so small such that the second term on the right hand side of \eqref{EqAcc} is  less than $\frac{1}{6}$, and then $\mu>0$ so small such that for all $s \in [-\mu, 0)$ we have that the first and third term on the right hand side are also less than $\frac{1}{6}$, we obtain that the right hand side of \eqref{EqAcc} is less than $\frac{1}{2}$. Hence, defining the curve $$ [-\mu, 0) \ni s \mapsto \sigma(s) := \big(\gamma^u(s), \gamma^v(s), \frac{\pi}{2}, \gamma^\varphi(-\mu) + \int_{-\mu}^s \frac{\kappa +\varepsilon}{r^2(\gamma(s'))} \, ds' \big) $$
we obtain $g(\dot{\sigma}, \dot{\sigma}) = -\Omega^2 \dot{\gamma}^u \dot{\gamma}^v + \frac{(\kappa + \varepsilon)^2}{r^2} < -\frac{1}{2}$ for all $s \in [-\mu, 0)$, and thus $\sigma$ is timelike. Since $\sigma$ moves slightly faster in $\varphi$ than $\gamma$ there is now a $-\mu < -\mu_0 < 0$ such that $\sigma^\varphi(-\mu_0) = \varphi_\infty = \pi$. Moreover, $\sigma^\varphi$ is strictly increasing and thus $$(\sigma^\varphi)^{-1} \Big|_{[\gamma^\varphi(-\mu), \pi)} : [\gamma^\varphi(-\mu), \pi) \to [-\mu, -\mu_0)$$ is  strictly increasing, bijective, and continuous.  Let $-\mu < -\mu_1 <0$. We now define a causal homotopy $\Gamma_{\mu_1} : [-\mu, -\mu_1] \times [-\mu, -\mu_1] \to M$ of $\gamma|_{[-\mu, -\mu_1]}$ with fixed endpoints by
$$\Gamma_{\mu_1} (s ; \lambda) = \begin{cases} \sigma(s) \qquad \qquad &\textnormal{ for } -\mu \leq s \leq (\sigma^\varphi)^{-1} \big(\gamma^\varphi(\lambda)\big) \\
\big(\gamma^u(s), \gamma^v(s), \frac{\pi}{2}, \gamma^\varphi(\lambda)\big) \qquad \qquad &\textnormal{ for } (\sigma^\varphi)^{-1} \big(\gamma^\varphi(\lambda) \big)\leq s \leq \lambda \\
\gamma(s) &\textnormal{ for } \lambda \leq s \leq -\mu_1 \;. \end{cases}$$
For each $\lambda \in [-\mu, -\mu_0]$ the curve $s \mapsto \Gamma_{\mu_1} (s; \lambda)$ is future directed timelike from $\gamma(-\mu)$ to $\gamma(-\mu_1)$. Note that the homotopy only homotopes the $\varphi$-component of $\gamma$, see Figure \ref{FigPhiHomotop}; the projection of $\Gamma_{\mu_1}( \cdot \, ; \lambda)$ to $Q$ traces out the projection of $\gamma|_{[-\mu, -\mu_0]}$ to $Q$ for all $\lambda$.\footnote{Let us point out that the spherical symmetry of $g$ is important for our construction of the homotopy to yield causal curves: the spherical symmetry implies \begin{equation*} \begin{split} g|_{\Gamma_{\mu_1}(s;\lambda)} &\big(\partial_s \Gamma_{\mu_1} (s;\lambda), \partial_s \Gamma_{\mu_1}(s;\lambda)\big)  \\
&= g|_{\Gamma_{\mu_1}(s;\lambda)}  \big( \dot{\Gamma}_{\mu_1}^u (s;\lambda) \partial_u + \dot{\Gamma}_{\mu_1}^v(s;\lambda) \partial_v + \dot{\Gamma}_{\mu_1}^\varphi(s;\lambda) \partial_\varphi, \dot{\Gamma}_{\mu_1}^u (s;\lambda) \partial_u + \dot{\Gamma}_{\mu_1}^v(s;\lambda) \partial_v + \dot{\Gamma}_{\mu_1}^\varphi(s;\lambda) \partial_\varphi\big) \\
&\overset{!}{=}g|_{\gamma(s)}  \big( \dot{\Gamma}_{\mu_1}^u (s;\lambda) \partial_u + \dot{\Gamma}_{\mu_1}^v(s;\lambda) \partial_v + \dot{\Gamma}_{\mu_1}^\varphi(s;\lambda) \partial_\varphi, \dot{\Gamma}_{\mu_1}^u (s;\lambda) \partial_u + \dot{\Gamma}_{\mu_1}^v(s;\lambda) \partial_v + \dot{\Gamma}_{\mu_1}^\varphi(s;\lambda) \partial_\varphi\big) \;.
\end{split}
\end{equation*}
}
\begin{figure}[h]
\centering
 \def\svgwidth{8cm}
   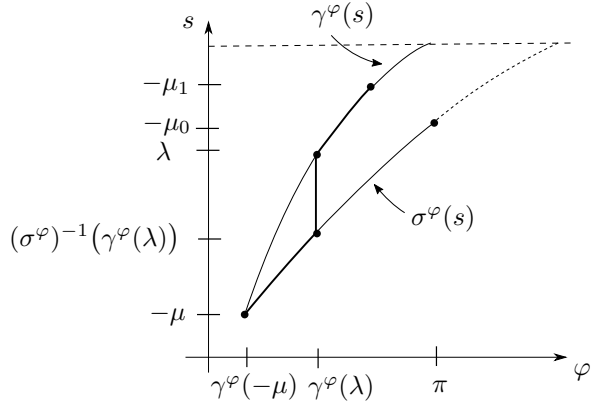
      \caption{The homotopy in the $\varphi$-component of $\Gamma_{\mu_1}$.} \label{FigPhiHomotop}
\end{figure}
By Lemma \ref{LemCausalHomotopy} $\iota \circ \Gamma_{\mu_1}$ maps into $\tilde{U}$ for all $-\mu_1 \in [-\mu, 0)$. In particular the $\tilde{\varphi} \circ \iota$-image of the curves $\sigma_{\mu_1}(s) := \Gamma_{\mu_1}(s; -\mu_1)$ maps into $\Big(C^+_{\nicefrac{5}{8}} + \tilde{\varphi}(\tilde{\gamma}(-\mu))\Big) \cap  C^-_{\nicefrac{5}{8}} \subset \subset \Reps$ for all $-\mu_1 \in [-\mu,0)$.  We define a future directed and future inextendible timelike curve $\sigma_0 : [-\mu,0) \to M$ by $$\sigma_0(s) := \begin{cases} \sigma(s) \qquad \qquad &\textnormal{ for } s \in [-\mu, -\mu_0] \\
\big(\gamma^u(s), \gamma^v(s), \frac{\pi}{2}, \pi\big) \qquad \qquad &\textnormal{ for } s \in [-\mu_0, 0) \;. \end{cases}$$
Note that $\sigma_0$ lies in the closure of $$\Big{\{}\sigma_{\mu_1} \big( [ -\mu, -\mu_1]) \; | \; -\mu < -\mu_1 <0 \Big{\}} \subseteq M \;.$$ The continuity of $\iota$ now implies that $\varphi \circ \iota \circ \sigma_0$ lies in the closure of $\Big(C^+_{\nicefrac{5}{8}} + \tilde{\varphi}(\tilde{\gamma}(-\mu))\Big) \cap C^-_{\nicefrac{5}{8}} \subset \subset \Reps$. Since it also lies below the graph of $f$ and $\sigma_0$ is future inextendible in $M$ it follows that it has a limit point on the graph of $f$.\footnote{Although it is not needed for the remainder of the argument, one can show that $\lim_{s \to 0} (\iota \circ \sigma_0)(s) = \tilde{p}$ as expected. For this let $\tilde{O} \subseteq \tilde{U}$ be a neighbourhood of $\tilde{p}$ and let $-\mu < -\hat{\mu} <0$ be so small that $$ \tilde{\varphi}^{-1} \Big(\overline{ \Big(C^+_{\nicefrac{5}{8}} + \tilde{\varphi}(\tilde{\gamma}(-\hat{\mu}))\Big) \cap C^-_{\nicefrac{5}{8}}}\Big) \subseteq \tilde{O} \;.$$ We can now repeat the construction of the causal homotopy $\Gamma_{\mu_1}$ but with $\mu$ replaced by $\hat{\mu}$. Note that $\mu_0$ will also be replaced by a $-\mu_0 < -\hat{\mu}_0 < 0$. Then the same line of reasoning leads to $(\iota \circ \sigma_0) \big([-\hat{\mu}_0,0)\big) \subseteq \tilde{\varphi}^{-1}\Big(\overline{\Big(C^+_{\nicefrac{5}{8}} + \tilde{\varphi}(\tilde{\gamma}(-\hat{\mu}))\Big) \cap C^-_{\nicefrac{5}{8}}}\Big) \subseteq \tilde{O}$, which shows $\lim_{s \to 0} (\iota \circ \sigma_0)(s) = \tilde{p}$.} Note that $\sigma_0|_{[-\mu_0, 0)}$ lies in the $\{u,v\}$-plane as desired.
\newline

\textbf{\underline{Step 2.2:}} We construct a null curve $[v_0, 0)  \ni s \overset{\tau}{\mapsto} (u_0, s, \frac{\pi}{2}, \pi)$, where $-\infty < u_0 < 0$, such that $\lim_{s \to 0} (\iota \circ \tau)(s) \in \tilde{M}$ exists.

Note that this step serves two purposes. On the one hand it is convenient for the construction in Step 3 to work with a radially outgoing null geodesic. On the other hand we have ensured that even if $0 = \lim_{s \to 0} u(\sigma_0(s)) = \lim_{s \to 0} u(\gamma(s))$ we can also find a curve that `leaves $M$ through $\{v= 0\} \cap \{-\infty < u < 0\}$'\footnote{The precise definition of `a curve $\sigma : [-1,0) \to M$ leaving $M$ through $\{v = 0\} \cap \{-\infty < u < 0\}$' is that $\lim_{s \to 0} v(\sigma(s)) = 0$, $\lim_{s \to 0} u (\sigma(s)) \in (-\infty, 0)$, and $\lim_{s \to 0} (\iota \circ \sigma)(s)$ exists in $\tilde{M}$.}, i.e., one cannot extend only through $\{v = 0\} \cap \{ u =0\}$ without extending at the same time through a bit of $\{v = 0\} \cap \{-\infty < u < 0\}$.

The first part of the argument is analogous to the one in Step 2 of the proof of Theorem \ref{ThmFLRWPast}. For $-\mu_0 < -\mu_1 <0$ we define a causal homotopy $\Gamma_{\mu_1} : [-\mu_0, -\mu_1] \times [-\mu_0, -\mu_1] \to M$  of $\sigma_0|_{[-\mu_0, -\mu_1]}$ with fixed endpoints by
\begin{equation*}
\Gamma_{\mu_1}(s ; \lambda) = \begin{cases} \big( u(\sigma_0(-\mu_0)), v(\sigma_0(2s + \mu_0)), \frac{\pi}{2}, \pi\big) \qquad \qquad &\textnormal{ for } s \in [-\mu_0, \frac{-\mu_0 + \lambda}{2}] \\
\big(u(\sigma_0(2s - \lambda)), v(\sigma_0(\lambda)), \frac{\pi}{2}, \pi\big) &\textnormal{ for } s \in [\frac{-\mu_0 + \lambda}{2}, \lambda] \\
\sigma_0(s) &\textnormal{ for } s \in [\lambda , - \mu_1] \;.
\end{cases}
\end{equation*}
See also Figure \ref{FigHomotopyNull}. By Lemma \ref{LemCausalHomotopy} $\iota \circ \Gamma_{\mu_1}$ maps into $\tilde{\varphi}^{-1}\Big(\Big(C^+_{\nicefrac{5}{8}} + \tilde{\varphi}(\tilde{\gamma}(-{\mu}))\Big) \cap C^-_{\nicefrac{5}{8}}\Big) \subseteq \tilde{U}$ for all $-\mu_0 < -\mu_1 <0$. Defining $\tau(s) := \big(u(\sigma_0(-\mu_0)), v(\sigma_0(s)), \frac{\pi}{2}, \pi\big)$ for $s \in [-\mu_0, 0)$ we obtain that $\iota \circ \tau$ maps into $\tilde{\varphi}^{-1}\Big(\Big(C^+_{\nicefrac{5}{8}} + \tilde{\varphi}(\tilde{\gamma}(-{\mu}))\Big) \cap C^-_{\nicefrac{5}{8}}\Big) \subseteq \tilde{U}$ and has a limit point on the graph of $f$. We can thus set $v_0 := v(\sigma_0(-\mu_0))$ and $u_0 := u(\sigma_0(-\mu_0))$.
\begin{figure}[h]
\centering
 \def\svgwidth{7cm}
   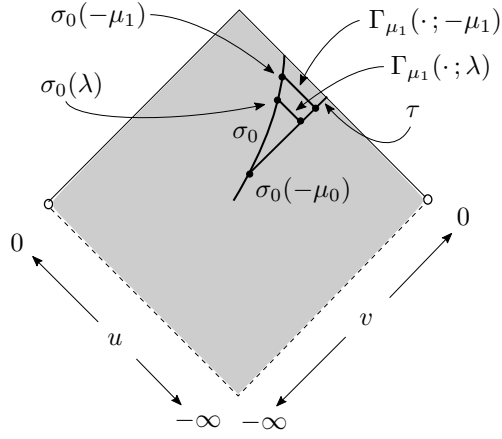
      \caption{The homotopy in the $\{u,v\}$-plane.} \label{FigHomotopyNull}
\end{figure}
\newline

\underline{\textbf{Step 3:}}  Let $B_\lambda (u_0, \frac{\pi}{2}, \pi) \subseteq (-\infty, 0) \times (0, \pi) \times (0, 2\pi)$ be the coordinate ball of radius $\lambda >0$ around $(u = u_0, \theta = \frac{\pi}{2}, \varphi = \pi)$. We show that $\lambda >0$ can be chosen so small that\footnote{Note that this notation unfortunately changes the ordering of the coordinates: by $[v_0, 0) \times B_\lambda(u_0, \frac{\pi}{2}, \pi)$ we denote the set of all $(u,v, \theta, \varphi) \in (-\infty, 0) \times (-\infty, 0) \times (0, \pi) \times (0, 2\pi)$ such that $(u, \theta, \varphi) \in B_\lambda (u_0, \frac{\pi}{2}, \pi)$ and $v \in [v_0, 0)$.}
\begin{equation*}
\begin{split}
&\iota \Big([v_0, 0) \times B_\lambda(u_0, \frac{\pi}{2}, \pi)\Big) \subseteq \tilde{U} \\
& \frac{\partial}{\partial x^\mu} = \frac{\partial \tilde{x}^\alpha}{\partial x^\mu} \frac{\partial}{\partial \tilde{x}^\alpha} \qquad \textnormal{ with } \qquad \Big{|}\frac{\partial \tilde{x}^\alpha}{\partial x^\mu}\Big{|} \leq C
\end{split}
\end{equation*}
for some $C \geq 0$. Here, $x^\mu$ denotes the set of coordinates $(u,v, \theta, \varphi)$ and $\tilde{x}^\alpha$ denotes the set of coordinates induced by $\tilde{\varphi} \circ \iota$.
\newline

We begin by recalling from Step 2.2 that $\tilde{\tau}\big([v_0, 0)\big) $ is compactly contained in $\tilde{U}$, where $\tilde{\tau} := \iota \circ \tau$. Using the continuity of $\iota$ we can choose $\lambda_1 >0$ so small that $\iota\Big( [v_0, \frac{19}{20} v_0] \times B_{\lambda_1}(u_0, \frac{\pi}{2}, \pi) \Big) \subset \subset \tilde{U}$. We also note that it follows from \eqref{EqChristoffelNullSing} that
\begin{equation}\label{EqParallelTAlongV}
\nabla_{\partial_v} \Big( \frac{1}{\Omega^2} \partial_v\Big) = 0\;, \qquad \quad \nabla_{\rd_v} \rd_u = 0\;, \qquad \quad \nabla_{\rd_v} \Big( \frac{1}{r} \rd_\theta\Big) = 0\;, \qquad \quad \nabla_{\rd_v} \Big(\frac{1}{r \sin \theta} \rd_\varphi \Big) = 0 \;.
\end{equation}
We define $e_1 := \frac{1}{r} \rd_\theta$, $e_2:= \frac{1}{r \sin \theta} \rd_\varphi$, $e_3 := \rd_u$, $e_4 := \frac{1}{\Omega^2} \rd_v$, which is thus a parallel frame field along the integral curves of $\partial_v$. We clearly have
\begin{equation}
\label{EqRelY}
\frac{\partial}{\rd x^\mu} = Y_\mu^{\; \; \nu} e_\nu \qquad \textnormal{ with } \qquad |Y_\mu^{\; \; \nu}| \leq C_0
\end{equation}
for some $C_0 > 0$. By compactness we have
\begin{equation}
\label{EqRelXTilde}
\iota_* e_\mu = \tilde{X}_\mu^{\; \; \alpha} \frac{\rd}{\rd \tilde{x}^\alpha} \qquad \textnormal{ with } \qquad |\tilde{X}_{\mu}^{\; \; \alpha} | \leq C_1 \qquad \textnormal{ on } \qquad \iota\Big([v_0, \frac{19}{20} v_0] \times B_{\lambda_1}(u_0, \frac{\pi}{2}, \pi)\Big) 
\end{equation}
for some $C_1 > 0$. For $\underline{x} = (u, \theta, \varphi) \in B_{\lambda_1}(u_0, \frac{\pi}{2}, \pi)$ let $v(\underline{x}) := \sup\big{\{} \hat{v} \in [v_0,0) \; | \; \iota\big([v_0, \hat{v}) \times \{\underline{x}\} \big) \subseteq \tilde{U} \big{\}}$.
\newline

\underline{\textbf{Step 3.1:}} We show that on $\bigcup_{\underline{x} \in B_{\lambda_1}(u_0, \frac{\pi}{2}, \pi)} \Big(\big[v_0, v(\underline{x})\big) \times \{\underline{x}\}\Big) \subseteq [v_0, 0) \times B_{\lambda_1}(u_0, \frac{\pi}{2}, \pi)$ we have $\Big|\frac{\rd \tilde{x}^\alpha}{\rd x^\mu} \Big| \leq \hat{C}$.
\newline
 
To show this let $\ux  \in B_{\lambda_1}(u_0, \frac{\pi}{2}, \pi)$ and let $\tau_{\ux} : [v_0, v(\ux)) \to M$, $\tau_{\ux} (s) = (s, \ux)$ be the outgoing null geodesics. Then $\tilde{\tau}_{\ux} := \iota \circ \tau_{\ux}$ is a causal curve mapping into $\tilde{U}$ and since $\tilde{x}_0$ is a time function on $\tilde{U}$ we can reparametrise $\tilde{\tau}_{\ux}$ by $\tilde{x}_0$ to obtain, in the $\tilde{x}^\alpha$ coordinates, the curve
\begin{equation*}
\begin{split}
\tilde{\tau}_{\mathrm{rep}, \ux} : (-\varepsilon_0, \varepsilon_0) \supseteq I_{\ux} &\to (-\varepsilon_0, \varepsilon_0) \times (-\varepsilon_1, \varepsilon_1)^3 \\
\tilde{\tau}_{\mathrm{rep}, \ux}(s) &= \big(s, \overline{\tilde{\tau}}_{\mathrm{rep}, \ux}(s)\big) \;.
\end{split}
\end{equation*}
Clearly we have $|I_{\ux}| \leq 2 \varepsilon_0$ and, as in \eqref{EqUniformBoundCausalCurve}, we have $||\dot{\tilde{\tau}}_{\mathrm{rep}, \ux} ||_{\R^4} \leq \frac{8}{5}$, both uniform in $\ux \in B_{\lambda_1}(u_0, \frac{\pi}{2}, \pi)$.
It now follows from Lemma \ref{LemBoundParallelTransport} that the parallel transport map $P_{\tilde{\tau}_{\mathrm{rep}, \ux}} (s_0, s_1) : T_{\tilde{\tau}_{\mathrm{rep}, \ux}(s_0)} \tilde{U} \to T_{\tilde{\tau}_{\mathrm{rep}, \ux}(s_1)} \tilde{U}$ is uniformly bounded in the $\tilde{x}^\alpha$-coordinates by a constant $C_2 >0$ , independent of $\ux \in B_{\lambda_1}(u_0, \frac{\pi}{2}, \pi)$. Since $e_\mu$ is parallel along $\tau_{\ux}$ and since parallel transport commutes with isometries we thus obtain 
\begin{equation}
\label{EqIsoPaTComm}
\big(\iota_* e_\mu)|_{\tilde{\tau}_{\mathrm{rep}, \ux} (s_1)} = \iota_* \Big( P_{\tau_{\mathrm{rep}, \ux}} (s_0, s_1) \big(e_\mu|_{\tau_{\mathrm{rep}, \ux}(s_0)}\big)\Big) = P_{\tilde{\tau}_{\mathrm{rep}, \ux}} (s_0, s_1) \Big( \big(\iota_* e_\mu\big)|_{\tilde{\tau}_{\mathrm{rep}, \ux}(s_0)}\Big) \;,
\end{equation}
where we have denoted $\iota^{-1} \circ \tilde{\tau}_{\mathrm{rep}, \ux}$ by $\tau_{\mathrm{rep}, \ux}$. It now follows from \eqref{EqRelXTilde} together with \eqref{EqIsoPaTComm} that
\begin{equation}
\label{EqPushForwardEMuBounded}
||\big( \iota_* e_\mu\big)|_{\tilde{\tau}_{\mathrm{rep}, \ux}} (s_1) ||_{\R^4} \leq 2 C_1 \cdot C_2  \quad \textnormal{ for all } s_1 \in I_{\ux}\;.
\end{equation} 
Moreover, we have $$\frac{\rd \tilde{x}^\alpha}{\rd x^\mu} \frac{\rd}{\rd \tilde{x}^\alpha} = \iota_* \Big( \frac{\rd}{\rd x^\mu}\Big) = \iota_* \Big( Y_{\mu}^{\; \; \nu} e_\nu\Big) = \iota_* \big( Y_{\mu}^{\; \; \nu}\big)  \cdot \big(\iota_* e_{\nu}\big) \;, $$
and thus $\frac{\rd \tilde{x}^\alpha}{\rd x^\mu} = \iota_*(Y_{\mu}^{\; \; \nu}) \cdot (\iota_* e_\nu)^\alpha$, which, together with \eqref{EqRelY} and \eqref{EqPushForwardEMuBounded} proves the claim in Step 3.1.
\newline

\underline{\textbf{Step 3.2:}} We now choose $0 < \lambda < \lambda_1$ so small that $$\overline{B_{2 \hat{C} \lambda}(\tilde{\tau})} := \Big{\{} \tilde{x} = (\tilde{x}_0, \ldots, \tilde{x}_3) \in \Reps \simeq \tilde{U} \; \Big| \; d\Big( \tilde{x}, (\tilde{\varphi} \circ \tilde{\tau})\big([v_0, 0)\big)\Big) \leq 2\hat{C} \lambda \Big{\}} $$
is compactly contained in $\tilde{U}$. Here, $d(\cdot, \cdot)$ denotes the Euclidean coordinate distance function  on $\Reps = (-\varepsilon_0, \varepsilon_0) \times (-\varepsilon_1, \varepsilon_1)^3$. Let $$J =\Big{\{} \hat{v} \in (v_0, 0) \; \Big| \; \iota\Big( [v_0, \hat{v}) \times \overline{B_{\lambda}(u_0, \frac{\pi}{2}, \pi)}\Big) \subseteq \overline{B_{2 \hat{C} \lambda}(\tilde{\tau})} \Big{\}} \;. $$
We show by continuity that $J = (v_0, 0)$. By the choice of $0 < \lambda < \lambda_1$ we have $\iota\Big( [v_0, \frac{19}{20} v_0] \times \overline{B_{\lambda}(u_0, \frac{\pi}{2}, \pi)}\Big) \subseteq \tilde{U}$. Thus, by Step 3.1 this gives $|| \iota_* \frac{\rd}{\rd x^\mu} ||_{\R^4} \leq 2 \hat{C}$ on $\iota\Big( [v_0, \frac{19}{20} v_0] \times \overline{B_{\lambda}(u_0, \frac{\pi}{2}, \pi)}\Big) $. Thus integrating the integral curves of $\iota_* \frac{\partial}{\rd u}$, $\iota_* \frac{\rd}{\rd \theta}$, $\iota_* \frac{\rd}{\rd \varphi}$ from $\tilde{\tau}\big( [v_0, \frac{19}{20} v_0]\big)$ we obtain that we have in fact $$\iota\Big( [v_0, \frac{19}{20} v_0] \times \overline{B_{\lambda}(u_0, \frac{\pi}{2}, \pi)}\Big)  \subseteq \overline{B_{2 \hat{C} \lambda}(\tilde{\tau})} \;.$$
Thus, $J$ is non-empty. To show the openness of $J$ let $\hat{v} \in J$. Then by the continuity of $\iota$ and the openness of $\tilde{U}$ there exists $\varepsilon >0$ such that  $\iota\Big( [v_0, \hat{v} + \varepsilon) \times \overline{B_{\lambda}(u_0, \frac{\pi}{2}, \pi)}\Big) \subseteq \tilde{U}$. It then follows again from Step 3.1 that on $\iota\Big( [v_0, \hat{v} + \varepsilon) \times \overline{B_{\lambda}(u_0, \frac{\pi}{2}, \pi)}\Big) $ we have $|| \iota_* \frac{\rd}{\rd x^\mu} ||_{\R^4} \leq 2 \hat{C}$ and the same argument as before shows then that $\iota\Big( [v_0, \hat{v} + \varepsilon) \times \overline{B_{\lambda}(u_0, \frac{\pi}{2}, \pi)}\Big) \subseteq \overline{B_{2 \hat{C} \lambda}(\tilde{\tau})}$. The closedness of $J$ is immediate by its definition. We thus obtain $J = (v_0, 0)$, which, together with Step 3.1 concludes the proof of Step 3.
\newline

\underline{\textbf{Step 4:}} We define a family $\sigma_{v_1}$ of loops in $[v_0, 0) \times B_\lambda (u_0, \frac{\pi}{2}, \pi)$ and show that the assumption of $\iota : M \hookrightarrow \tilde{M}$ being a $C^{0,1}_{\loc}$-extension implies that the holonomy along those loops is uniformly bounded.
\newline

Let $v_1 \in (v_0, 0)$. In $(v,u, \theta, \varphi)$-coordinates on $M$ we define the following points
\begin{equation*}
\begin{aligned}
&A := (v = v_0, u= u_0, \frac{\pi}{2}, \pi + \frac{\lambda}{2}) \qquad \qquad &&B:= (v = v_0, u = u_0, \frac{\pi}{2}, \pi - \frac{\lambda}{2}) \\
&C(v_1):=( v=v_1, u=u_0, \frac{\pi}{2}, \pi - \frac{\lambda}{2}) \qquad \qquad &&D(v_1) := (v = v_1, u = u_0, \frac{\pi}{2}, \pi + \frac{\lambda}{2}) \;,
\end{aligned}
\end{equation*}
and also the following curves:
\begin{equation*}
\begin{aligned}
&\overrightarrow{AB} : [0, \lambda] \to M\;, \qquad \qquad &&\overrightarrow{AB}(s) :=(v_0, u_0, \frac{\pi}{2}, \pi + \frac{\lambda}{2} - s) \\
&\overrightarrow{BC(v_1)} : [v_0, v_1] \to M \;, \qquad \qquad &&\overrightarrow{BC(v_1)}(s) := (s, u_0, \frac{\pi}{2}, \pi - \frac{\lambda}{2}) \\
&\overrightarrow{C(v_1)D(v_1)} : [0, \lambda] \to M\;, \qquad \qquad &&\overrightarrow{C(v_1)D(v_1)}(s) := (v_1, u_0, \frac{\pi}{2}, \pi - \frac{\lambda}{2} + s ) \\
&\overrightarrow{D(v_1) A} : [v_0, v_1] \to M \;, \qquad \qquad &&\overrightarrow{D(v_1)A}(s) := (v_1 + v_0 - s, u_0, \frac{\pi}{2}, \pi + \frac{\lambda}{2}) \;.
\end{aligned}
\end{equation*}
See also Figure \ref{FigLoopsNull}.
\begin{figure}[h]
\centering
 \def\svgwidth{7cm}
   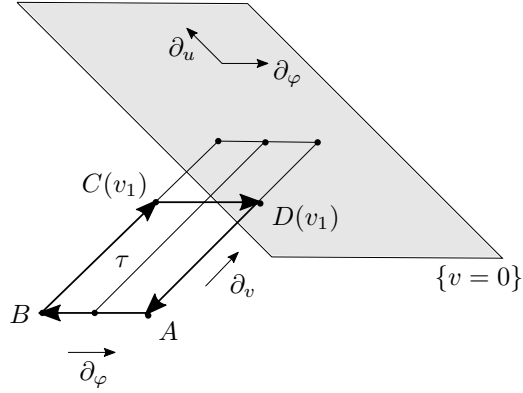
      \caption{The family of loops $\sigma_{v_1}$ in the $\{v, \varphi\}$-plane.} \label{FigLoopsNull}
\end{figure}
Let us denote by $\sigma_{v_1} := \overrightarrow{D(v_1)A}\; * \; \overrightarrow{C(v_1) D(v_1)} \; * \; \overrightarrow{BC(v_1)} \; * \; \overrightarrow{AB} $ the concatenation of the four curve segments and, as usual, by an over-set tilde the composition with the isometric embedding $\iota$ into $\tilde{M}$. We clearly have that the size of the domains of $\widetilde{\overrightarrow{D(v_1)A}}$, $\widetilde{\overrightarrow{C(v_1) D(v_1)} }$, $\widetilde{\overrightarrow{BC(v_1)} }$, $\widetilde{\overrightarrow{AB}}$ is uniformly bounded for $v_1 \in (v_0, 0)$, and, since the tangent vectors of the curves are either $\iota_*\big(\partial_v\big)$ or $\iota_*\big(\partial_\varphi\big)$ we have by Step 3 that the tangent vectors in the $\tilde{x}^\alpha$-coordinates are also uniformly bounded in $v_1$. It thus follows from Lemma \ref{LemBoundParallelTransport} that the parallel transport map $P_{\tilde{\sigma}_{v_1}} : T_{\iota(A)} \tilde{U} \to T_{\iota(A)} \tilde{U}$ along the loop $\tilde{\sigma}_{v_1}$ is also uniformly bounded in $v_1 \in (v_0, 0)$ with respect to the Euclidean norm induced by the $\tilde{x}^\alpha$-coordinates on $\tilde{U}$.
\newline

\underline{\textbf{Step 5:}} We show that the parallel transport map $P_{{\sigma}_{v_1}} : T_{A} M \to T_A M$ along the loop $\sigma_{v_1}$ is unbounded in the Euclidean norm induced by the coordinates $(v,u, \theta, \varphi)$ as $v_1 \to 0$. Since parallel transport commutes with isometries this is in contradiction to Step 4 and thus concludes the proof of Theorem \ref{ThmOneNullSing}.
\newline

We denote by $P_{\overrightarrow{AB}} : T_AM \to T_BM$ the parallel transport map along $\overrightarrow{AB}$, and similarly $P_{\overrightarrow{BC(v_1)}}$, $P_{\overrightarrow{C(v_1)D(v_1)}}$, and $P_{\overrightarrow{D(v_1)A}}$. As usual we denote by $||P_{\overrightarrow{AB}}||_{\R^4}$ the operator norm of the matrix representation of $P_{\overrightarrow{AB}}$ with respect to the basis $(\partial_u, \partial_v, \partial_\theta, \partial_\varphi)$ for $T_AM$ and $T_BM$, and analogously for the other parallel transport maps. We trivially have
\begin{equation}
\label{EqPTAB}
||P_{\overrightarrow{AB}}||_{\R^4} \leq C  \qquad \textnormal{ and } \qquad ||(P_{\overrightarrow{AB}})^{-1}||_{\R^4} \leq C \;.
\end{equation}
It also follows directly from \eqref{EqParallelTAlongV} (together with the continuous positive extension of $r$ and $\Omega^2$ to $\overline{Q}$) that we have
\begin{equation}
\label{EqPTNULL}
\begin{split}
||P_{\overrightarrow{BC(v_1)}}||_{\R^4} \leq C  \qquad &\textnormal{ and } \qquad ||(P_{\overrightarrow{BC(v_1)}})^{-1}||_{\R^4} \leq C \;, \\
||P_{\overrightarrow{D(v_1)A}}||_{\R^4} \leq C  \qquad &\textnormal{ and } \qquad ||(P_{\overrightarrow{D(v_1)A}})^{-1}||_{\R^4} \leq C 
\end{split}
\end{equation}
uniformly in $v_1 \in (v_0, 0)$. We now compute the parallel transport along $\overrightarrow{C(v_1)D(v_1)}$. It follows from \eqref{EqChristoffelNullSing} that \underline{for $\theta = \frac{\pi}{2}$ } we have
\begin{equation} \label{EqPTPHI}
\begin{aligned}
\nabla_{\rd_\varphi} \rd_u &= \frac{\rd_u r}{r} \, \rd_\varphi\;, \qquad \qquad &&\nabla_{\rd_\varphi} \rd_v = \frac{\rd_v r}{r} \, \rd_\varphi \;,\\
\nabla_{\rd_\varphi} \rd_\theta &=0 \;,\qquad \qquad &&\nabla_{\rd_\varphi} \rd_{\varphi} = \frac{2r \rd_v r}{\Omega^2} \, \rd_u + \frac{2r \rd_u r}{\Omega^2} \, \rd_v \;.
\end{aligned}
\end{equation}
By the assumption of the theorem there exists a $v_{\nicefrac{1}{2}} \in (v_0, 0)$ such that for $v_{\nicefrac{1}{2}} < v_1 <0$ we have $\rd_u r(u_0, v_1) <0$ and $\rd_v r(u_0, v_1) < 0$.  \textbf{The remaining computations are valid only for $\theta = \frac{\pi}{2}$ and $v \in (v_{\nicefrac{1}{2}}, 0)$. }

 We define the null vectors
$$\ell_{(u,v)} := \frac{\rd_v r}{\Omega^2} \rd_u + \frac{\rd_u r}{\Omega^2} \rd_v + \frac{\sqrt{\rd_u r \cdot \rd_v r}}{\Omega r} \rd_\varphi \qquad \textnormal{ and } \qquad \underline{\ell}_{(u,v)} := \frac{\partial_v r}{ \Omega^2} \rd_u + \frac{\rd_ur}{\Omega^2} \rd_v - \frac{\sqrt{\rd_u r \cdot \rd_v r}}{\Omega r} \rd_\varphi \;. $$
A direct computation using \eqref{EqPTPHI} gives $$\nabla_{\rd_\varphi} \ell_{(u,v)} = \sqrt{\frac{2m}{r} - 1} \cdot \ell_{(u,v)} \qquad \textnormal{ and } \qquad \nabla_{\rd_\varphi}  \underline{\ell}_{(u,v)} = - \sqrt{\frac{2m}{r} - 1}  \cdot \underline{\ell}_{(u,v)} \;, $$
where we have used $\sqrt{\frac{2m}{r} - 1} = \frac{2\sqrt{\rd_u r \cdot \rd_v r}}{\Omega}$, where $m := \frac{r}{2}\big( 1 + \frac{4 \rd_u r \cdot \rd_v r}{\Omega^2}\big)$ is the Hawking mass. We thus obtain\footnote{Together with $\nabla_{\rd_\varphi} \rd_\theta = 0$ and $\nabla_{\rd_\varphi} \big( - \partial_u r \cdot \partial_v + \partial_v r \cdot \partial_u\big) = 0$, which follows directly from \eqref{EqPTPHI}, this solves the parallel transport map along $\rd_\varphi$ for $\theta = \frac{\pi}{2}$ and $v \in (v_{\nicefrac{1}{2}}, 0)$.   Note that $-\rd_u r \cdot \rd_v + \rd_v r \cdot \rd_u$ is a spacelike vector, since $\rd_u r \cdot \rd_v r > 0$. }
\begin{equation}
\label{EqPTLLBAR}
\nabla_{\rd_\varphi}\Big(e^{-\sqrt{\frac{2m}{r} - 1} \cdot \varphi }\cdot \ell_{(u,v)} \Big)= 0 \qquad \textnormal{ and } \qquad \nabla_{\rd_\varphi}\Big(e^{\sqrt{\frac{2m}{r} - 1} \cdot \varphi }\cdot \underline{\ell}_{(u,v)}\Big) = 0 \;.
\end{equation}
Using $\rd_\varphi = \frac{r}{\sqrt{\frac{2m}{r} - 1}} \big(\ell_{(u,v)} - \underline{\ell}_{(u,v)}\big)$ and \eqref{EqPTLLBAR} we can now compute the parallel transport of $\rd_\varphi$ along $\overrightarrow{C(v_1)D(v_1)}$:
\begin{equation*}
\begin{split}
P_{\overrightarrow{C(v_1)D(v_1)}} \partial_\varphi &= \frac{r}{\sqrt{\frac{2m}{r} - 1}} \Big( e^{-\sqrt{\frac{2m}{r} - 1} \cdot \lambda} \cdot \ell_{(u,v)} - e^{\sqrt{\frac{2m}{r} - 1} \cdot \lambda} \cdot \underline{\ell}_{(u,v)}\Big) \\
&= \frac{r}{\sqrt{\frac{2m}{r} - 1}} \Big( \underbrace{\big[ e^{-\sqrt{\frac{2m}{r} - 1} \cdot \lambda} - e^{\sqrt{\frac{2m}{r} - 1} \cdot \lambda}\big]}_{= - 2 \sinh (\sqrt{\frac{2m}{r}-1} \cdot \lambda)} \cdot \big[ \frac{\rd_v r}{\Omega^2} \rd_u + \frac{\rd_u r}{\Omega^2} \rd_v\big] \\
&\qquad \qquad \qquad + \big[ e^{-\sqrt{\frac{2m}{r} - 1} \cdot \lambda} + e^{\sqrt{\frac{2m}{r} - 1} \cdot \lambda}\big] \cdot \frac{\sqrt{\frac{2m}{r} - 1}}{2r} \rd_\varphi\Big) \;.
\end{split}
\end{equation*}
Since we have $\frac{\sinh x}{x} \geq 1$ it follows that
\begin{equation}
\label{EqPPHIINFINITE}
||P_{\overrightarrow{C(v_1)D(v_1)}} \partial_\varphi||_{\R^4} \geq |2 r \lambda| \cdot \Big|\frac{ \rd_v r}{\Omega^2}\Big| \to \infty
\end{equation}
for $v_1 \to 0$ by assumption. 

Let now $X(v_1) :=\big((P_{\overrightarrow{AB}})^{-1} \circ (P_{\overrightarrow{BC(v_1)}})^{-1} \big) (\rd_\varphi) \in T_{A}M$. We thus have 
\begin{equation} 
\label{EqBound1}
||X(v_1)||_{\R^4} \leq ||(P_{\overrightarrow{AB}})^{-1} ||_{\R^4} \cdot ||(P_{\overrightarrow{BC(v_1)}})^{-1} ||_{\R^4} \underbrace{||\rd_\varphi||_{\R^4}}_{=1} \end{equation} 
and
\begin{equation*}
\begin{split}
\frac{||P_{\sigma_{v_1}} X(v_1)||_{\R^4}}{||X(v_1)||_{\R^4}} &= \frac{||(P_{\overrightarrow{D(v_1)A}} \circ P_{\overrightarrow{C(v_1)D(v_1)}} \circ P_{\overrightarrow{BC(v_1)}} \circ P_{\overrightarrow{AB}} )(X(v_1))||_{\R^4}}{||X(v_1)||_{\R^4}} \\
&\geq \frac{||(P_{\overrightarrow{D(v_1)A}} \circ P_{\overrightarrow{C(v_1)D(v_1)}})(\rd_\varphi)||_{\R^4}}{||(P_{\overrightarrow{AB}})^{-1} ||_{\R^4} \cdot ||(P_{\overrightarrow{BC(v_1)}})^{-1} ||_{\R^4} \cdot  ||\rd_\varphi||_{\R^4}} \\
&\geq \frac{||P_{\overrightarrow{C(v_1)D(v_1)}}(\rd_\varphi)||_{\R^4}}{||(P_{\overrightarrow{D(v_1)A}})^{-1}||_{\R^4} \cdot ||(P_{\overrightarrow{AB}})^{-1} ||_{\R^4} \cdot ||(P_{\overrightarrow{BC(v_1)}})^{-1} ||_{\R^4} \cdot ||\rd_\varphi||_{\R^4}} \\
&\geq \frac{||P_{\overrightarrow{C(v_1)D(v_1)}}(\rd_\varphi)||_{\R^4}}{C^3 \cdot  ||\rd_\varphi||_{\R^4}} \;, 
\end{split}
\end{equation*}
where we have used \eqref{EqBound1} in the first inequality, $||Ax|| \geq \frac{||x||}{||A^{-1}||}$ in the second, and  \eqref{EqPTAB} and \eqref{EqPTNULL} in the third. It now follows from  \eqref{EqPPHIINFINITE} that $||P_{\sigma_{v_1}}||_{\R^4} \to \infty$ for $v_1 \to 0$, which concludes Step 5.
\end{proof}

\subsection{Reissner-Nordstr\"om-Vaidya spacetimes} \label{SecRNV}

The Reissner-Nordstr\"om-Vaidya (RNV) spacetime $(M,g)$ is given by $M= \R \times (0,\infty) \times \mathbb{S}^2$ with canonical $(v,r)$-coordinates on the first two factors and 
\begin{equation}
\label{EqMetricRNVR}
g = -\Big(1 - \frac{2 \varpi (v)}{r} + \frac{e^2}{r^2}\Big) \, dv^2 + dv \otimes dr + dr \otimes dv + r^2 \, \mathring{\gamma} \;,
\end{equation} where $e>0$ and $\varpi : \R \to (0, \infty)$ is a smooth non-decreasing function, \cite{BonVai70}. Together with the Maxwell field $F := \frac{2e}{r^2} \, dv \wedge dr$ and $\rho := \frac{1}{r^2} \partial_v \varpi \geq 0$, and defining the stress-energy tensor $T^{\mathrm{em}}_{\mu \nu} = F_{\mu \lambda} F_{\nu}^{\; \; \lambda} - \frac{1}{2} g_{\mu \nu} F_{\lambda \rho} F^{\lambda \rho}$ of the electromagnetic field and the stress-energy tensor $ T^{\mathrm{dust}}_{\mu \nu} = \rho \cdot \partial_\mu v \partial_\nu v$ of dust, it solves the Einstein-Maxwell-null-dust equations
\begin{equation*}
\begin{aligned}
R_{\mu \nu} - \frac{1}{2} g_{\mu \nu} R &= 2(T^{\mathrm{em}}_{\mu \nu} + T^{\mathrm{dust}}_{\mu \nu} ) \\
dF &= 0 \;, \qquad \nabla^{\mu} F_{\mu \nu} = 0 \\
\nabla^\mu  T^{\mathrm{dust}}_{\mu \nu}  &= 0 \;, \qquad g^{-1}(dv, dv) = 0 \;.
\end{aligned}
\end{equation*}
A time orientation is fixed by stipulating that $-\partial_r$ is future directed. We assume $\varpi(v) = \varpi(\infty) - \beta v^{-p}$, where $\beta >0$, $p>1$, and $\varpi(\infty) > e >0$, and restrict our considerations to $v \geq v_0$, where $v_0 >0$ is so large that $\varpi(v_0) > e$. Thus, for late affine time $v$ the RNV spacetime we are considering models the continuous influx of null-dust into a sub-extremal Reissner-Nordstr\"om black hole decaying with a tail $\rho \sim v^{-(p+1)}$.   A Penrose diagram is given in Figure \ref{FigPenroseRNV}. 
\begin{figure}[h]
\centering
 \def\svgwidth{8cm}
   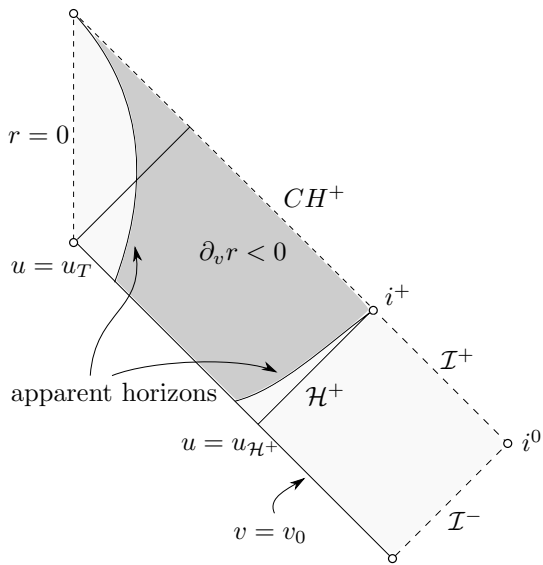
      \caption{Penrose diagram of the Reissner-Nordstr\"om-Vaidya spacetime.} \label{FigPenroseRNV}
\end{figure}
In the following we introduce spherically symmetric double null coordinates $(u,v)$ on the RNV spacetime, see also Appendix B of \cite{PoiIs90}. Such coordinates are in particular used to infer the  structure of the above Penrose diagram. Note however that in general no closed form solution exists for the second null coordinate $u$. We construct it by the method of characteristics, prescribing initial values on $\{v = v_0\}$.  We then focus on the region $\uh < u < u_T$, which is defined by the property that outgoing null rays are eventually trapped. One can in fact show that every outgoing null ray starting inside the black hole is eventually trapped, see Proposition \ref{PropTrapped}, and thus $u_T$ corresponds to the left endpoint of the domain of definition of the null coordinate $u$, cf.\ Figure \ref{FigPenroseRNV}. In Proposition \ref{PropRNV} we show that for any $\tilde{u} \in (\uh, u_T)$ the region $(\uh, \tilde{u}) \times (v_0, \infty)$ satisfies the assumptions of Theorem \ref{ThmOneNullSing}, and thus the part of $\CH$ which is covered by the null coordinate $u$ is $C^{0,1}_{\loc}$-inextendible in the precise sense of Theorem \ref{ThmOneNullSing}.  One can now construct another null coordinate $u$ in the same way with initial values prescribed on $\{ v = v_1 > v_0\}$ to cover a larger part of $\CH$. In this way it follows that all of $\CH$ is $C^{0,1}_{\loc}$-inextendible in the sense of Theorem \ref{ThmOneNullSing}.

We now begin. Let $f = 1 - \frac{2 \varpi(v)}{r} + \frac{e^2}{r^2}$. Then $\ell := - \frac{\rd}{\rd r}$ is future directed ingoing null and $\underline{\ell}:= \frac{\rd}{\rd v} + \frac{f}{2} \frac{\rd}{\rd r}$ is future directed outgoing null. We construct a null coordinate $u$ using the method of characteristics: on $\{v = v_0\}$ we set $u|_{\{v = v_0\}} = -r$ and extend it by requiring that it is constant along the integral curves of $\underline{\ell}$.\footnote{Note that a priori this does not yield a globally defined function -- and indeed the top left corner of the Penrose diagram will not be covered.} The integral curves of $\underline{\ell}$ are determined by \begin{equation} 
\label{EqEqIntegralCurves}
\dot{v} = 1 \quad \textnormal{ and } \quad \dot{r} = \frac{f}{2}\end{equation}
 and thus, since these are curves of constant $u$, we obtain
\begin{equation}
\label{EqVROnConstantU}
\frac{\partial r}{\partial v}\Big|_u = \frac{\dot{r}}{\dot{v}} = \frac{f}{2} \;.
\end{equation}
We write\footnote{Recall that we are restricting to $v \geq v_0$ such that $\varpi(v) > e$, thus $f$ has indeed two roots.} \begin{equation}
\label{EqFTwoDef}
f(v,r) = 1 - \frac{2\varpi(v)}{r} + \frac{e^2}{r^2} = \frac{1}{r^2}( r - r_+(v))(r - r_-(v))$$ with $$r_+(v) = \varpi(v) + \sqrt{\varpi(v)^2 - e^2} \quad \textnormal{ and } \quad r_-(v) = \varpi(v) - \sqrt{\varpi(v)^2 - e^2} \;.
\end{equation} 
Since $\varpi'(v) > 0$, we clearly have $r_+'(v) >0$ and moreover $r_-'(v) = \varpi'(v) - \frac{\varpi(v) \varpi'(v)}{\sqrt{\varpi(v)^2 - e^2}} < 0$. Defining $r_+(\infty)$ and $r_-(\infty)$ analogously with $\varpi(v)$ replaced by $\varpi(\infty)$, we thus obtain $r_+(v) \nearrow r_+(\infty)$ and $r_-(v) \searrow r_-(\infty)$ for $v \to \infty$.
\begin{figure}[h]
\centering
 \def\svgwidth{6cm}
   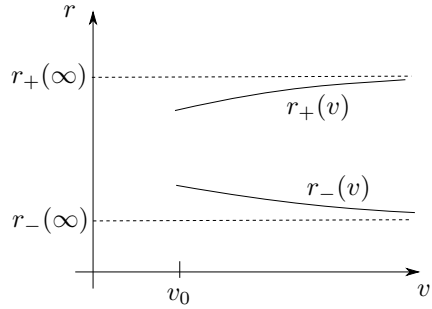
      \caption{Functional dependence of roots $r_+(v)$ and $r_-(v)$ on $v$.} \label{FigRV}
\end{figure}
Defining the regions
\begin{equation*}
\begin{aligned}
I&:= \{ (v,r) \in [v_0, \infty) \times (0, \infty) \; | \; r_+(v) < r\} \\
II&:=\{ (v,r) \in [v_0, \infty) \times (0, \infty) \; | \; r_-(v) < r< r_+(v) \} \\
III &:= \{ (v,r) \in [v_0, \infty) \times (0, \infty) \; | \; 0 < r<r_-(v) \} \;,
\end{aligned}
\end{equation*}
we thus have $f(v,r) >0$ on $I \cup III$ and $f(v,r) <0$ and $II$. Furthermore, we have $f(v,r) = 0$ on the transition curves $v \mapsto (v, r_-(v))$ and $v \mapsto (v, r_+(v))$, which are the apparent horizons.

\textbf{Claim 1:} The integral curves of $\underline{\ell}$ exist for all $v \geq v_0$.
\begin{proof}
For the ODE \eqref{EqEqIntegralCurves} to break down in finite affine time $v \geq v_0$ we must either have that $r$ goes to zero or infinity in finite affine time. Since we have $\frac{\rd r}{\rd v}\Big|_u = \frac{f}{2} > 0$ in $III$, we cannot have $r \to 0$. Moreover, it follows that $r$ is uniformly bounded from below along each integral curve, which directly implies that $f$ is uniformly bounded from above, and thus $r$ cannot go to infinity in finite affine time.
\end{proof}

Recall that $r(u,v)$ is the value of $r$ at affine time $v$ along the integral curve of $\underline{\ell}$ that passes through $(v = v_0, r = -u)$.

\textbf{Claim 2:}  If $r(u,v_1) > r_+(\infty)$  for some $v_1 \geq v_0$, then $r(u,v) \to \infty$ for $v \to \infty$.
\begin{proof}
Note first that by \eqref{EqVROnConstantU} and since $r(u,v_1) > r_+(\infty)$ we have that $r(u,v)$ is monotonically growing in $v$.   Let $r_1 := r(u,v_1)$. Then for $r \geq r_1$ and $v \geq v_1$ we have
\begin{equation*}
f(v,r) = \frac{1}{r^2}(r - r_+(v))(r-r_-(v)) \geq \frac{1}{r^2} (r - r_+(\infty))( r - r_+(\infty)) = (1 - \frac{r_+(\infty)}{r})^2 \geq (1 - \frac{r_+(\infty)}{r_1})^2 > 0 \;.
\end{equation*} 
It then follows from \eqref{EqVROnConstantU} that $r(v,u) \to \infty$ for $v \to \infty$.
\end{proof}

\textbf{Claim 3:} If $\frac{\rd r}{\rd v}\Big|_u (u,v_1) <0$ (which is the case if, and only if, $r(u,v_1) \in (r_-(v_1), r_+(v_1))$), then $\frac{\rd r}{\rd v}\Big|_u (u,v) <0$ for all $v \geq v_1$ and $\lim_{v \to \infty} r(u,v) = r_-(\infty)$.
\begin{proof}
To prove the first part, assume to the contrary that\footnote{We will adopt from now on the convention that $\rd_v = \frac{\rd}{\rd v}\Big|_u$ is the partial derivative in the $(u,v)$-coordinate system.} $\rd_vr (u, v_2) \geq 0 $ for some $v_2 > v_1$. Let $v_3 \in (v_1, v_2)$ be the smallest value such that $\rd_v r (u,v_3) = 0$. This implies $r_+(v) > r(u,v) > r_-(v)$ for $v \in (v_1, v_3)$ and $r(u,v_3) = r_-(v_3)$. But we also have $\rd_v r(u,v_3) = 0 > \rd_v r_-(v_3)$, which implies the contradiction $ r_- (v) > r(u,v)$ for $v < v_3$ close enough to $v_3$.

For the second part note that the first part implies that $r(u,v)$ is strictly monotonically decreasing for $v \geq v_1$ and also $r(u,v) \in (r_-(v), r_+(v))$ for all $v \geq v_1$. We cannot have $\lim_{v \to \infty} r(u,v) = r_0 > r_-(\infty)$, since $\lim_{v \to \infty} f(v,r) <0$ for all $r \in (r_-(\infty), r_+(\infty))$, and thus for late $v$ we would have that $\partial_v r(u,v) \leq c <0$, which is a contradiction. It thus follows that $\lim_{v \to \infty} r(u,v) = r_-(\infty)$.
\end{proof}

Let us now define $$\uh := \sup \{ u \in (-\infty, 0) \; | \; \lim_{v \to \infty} r(u', v) = \infty \quad \forall u' \leq u\}\;.$$
By Claim 2 we have $-\infty < \uh$. Choosing $u \in (-r_+(v_0), -r_-(v_0)) $ and using Claim 3 we see that $\lim_{v \to \infty} r(u,v) = r_-(\infty)$. This shows $\uh < 0$. The hypersurface $u = \uh$ is the event horizon, cf.\ Figure \ref{FigPenroseRNV}.\footnote{Although not needed in the following, note that it is also obvious that the domain of definition of $u$ covers all of region $I$, since tracing backwards the integral curves of $\underline{\ell}$ starting in $I$ they have to stay in region $I$ by Claim 3, have uniformly bounded velocity and thus intersect $\{v = v_0\}$.}

\textbf{Claim 4}: We have $\lim_{v \to \infty} r(\uh, v) = r_+(\infty)$.
\begin{proof}
By Claim 3 we must have $\rd_v r (u,v) \geq 0$ for all $u < \uh$, $v \geq v_0$. This, together with $\lim_{v \to \infty} r(u,v) = \infty$ for $u < \uh$ implies $r(u,v) \geq r_+(v)$ for all $u < \uh$, $v \geq v_0$. By continuity we also have $r(\uh,v) \geq r_+(v)$ for all $v \geq v_0$. We thus obtain $\liminf_{v \to \infty} r(\uh, v) \geq r_+(\infty)$.

We now show $\limsup_{v \to \infty} r(\uh, v) \leq r_+(\infty)$. Assuming to the contrary that $\limsup_{v \to \infty} r(\uh, v) > r_+(\infty)$, there exists $v_1 \geq v_0$ such that $r(\uh, v_1) > r_+(\infty)$. By continuity there exists $\delta >0$ such that $r(u,v_1) > r_+(\infty)$ for all $u \in [\uh, \uh + \delta]$. But then Claim 2 shows $\lim_{v \to \infty} r(u,v) = \infty$ for all $u \in [\uh, \uh + \delta]$, in contradiction to the definition of $\uh$.
\end{proof}

Taking a $\rd_u$ derivative of \eqref{EqVROnConstantU} we obtain 
\begin{equation}
\label{EqDRU}
\rd_v (\rd_u r) = \frac{1}{r^2} \big( \varpi(v) - \frac{e^2}{r}\big)\cdot \rd_u r \;.
\end{equation}
Since we have $\rd_ur (u, v_0) = -1$, it follows directly that $\rd_u r (u,v) <0$ for all $u \in (-\infty, 0)$ and $v \geq v_0$, since $ \rd_ur$ cannot pass through zero -- since then the solution would be identically zero by the uniqueness of solutions of \eqref{EqDRU}.

We also define the asymptotic surface gravity of the event horizon by 
\begin{equation}
\label{EqDefSGE}
\kappa_+(\infty) := \frac{r_+(\infty) - r_-(\infty)}{2 r_+^2(\infty)} = \frac{1}{2}\rd_r f(\infty, r_+(\infty)) = \frac{1}{r_+^2(\infty)} \big(\varpi(\infty) - \frac{e^2}{r_+(\infty)}\big) \;,
\end{equation}
where the last two equalities follow easily from \eqref{EqFTwoDef}.

\textbf{Claim 5:} We have $\limsup_{v \to \infty} r(u,v) < r_+(\infty)$ for all $u > \uh$.
\begin{proof}
We use the red-shift effect along the event horizon for the proof. Assume this is not the case. Then there exists $u_1 > \uh$ with $\limsup_{v \to \infty} r(u_1, v) \geq r_+(\infty)$. Since $\partial_u r <0$ we have $r(u_1,v) < r(\uh, v)$ for all $v \geq v_0$ and thus $\limsup_{v \to \infty} r(u_1,v) \leq r_+(\infty)$. This gives $\limsup_{v \to \infty} r(u_1, v) = r_+(\infty)$.  Moreover, we must have $\rd_v r(u_1, v) \geq 0$, since otherwise we would obtain a contradiction from Claim 3. This shows $\lim_{v \to \infty}r(u_1, v) = r_+(\infty)$. Now using $\rd_u r <0$ again, which implies $r(u_1, v) < r(u,v) < r(\uh, v)$ for all $u \in (\uh, u_1)$ and $v \geq v_0$, we obtain $\lim_{v \to \infty} r(u,v) = r_+(\infty)$ for all $u \in [\uh, u_1]$.

We can now choose $v_1 \geq v_0$ so large that $\frac{1}{r(u_1,v)^2}\big(\varpi(v) - \frac{e^2}{r(u_1,v)}\big) \geq \frac{1}{2} \kappa_+(\infty)$ for all $v \geq v_1$. Together with $\rd_u r <0$ this gives $\frac{1}{r(u,v)^2}\big(\varpi(v) - \frac{e^2}{r(u,v)}\big) \geq \frac{1}{2} \kappa_+(\infty)$ for all $u \in [\uh, u_1]$ and $v \geq v_1$. Thus, \eqref{EqDRU} gives
$$\rd_u r(u,v) \leq \rd_u r(u, v_1) \cdot e^{\frac{1}{2} \kappa_+(\infty) (v - v_1)} \leq - \underbrace{\min_{u \in [\uh, u_1]} |\rd_u r(u, v_1)|}_{>0} \cdot e^{\frac{1}{2} \kappa_+(\infty) (v - v_1)} $$ for all $u \in [\uh, u_1] $ and $v \geq v_1$, which, after integration in $u$, is a contradiction to $\lim_{v \to \infty} r(u,v) = r_+(\infty)$ for all $u \in [\uh, u_1]$.
\end{proof}

Let $\Delta_1 >0$ be such that $r(\uh + \Delta_1, v_0) \in (r_-(v_0), r_+(v_0))$. We in particular show now that there is a small trapped neighbourhood towards the future of the event horizon, cf.\ Figure \ref{FigPenroseRNV}.

\textbf{Claim 6:} For all $u \in (\uh, \uh + \Delta_1]$ there exists a $v_1(u) \geq v_0$ such that $\rd_v r(u,v) <0$ for all $v \geq v_1(u)$ and $\lim_{v \to \infty} r(u,v) = r_-(\infty)$.
\begin{proof}
Clearly $\rd_v r(\uh + \Delta_1, v_0) <0$ and by Claim 3 we have $\rd_v r(\uh + \Delta_1, v) <0$ for all $v \geq v_0$ and $\lim_{t \to \infty} r(\uh + \Delta_1, v ) = r_-(\infty)$. Using $\partial_u r <0$ and Claim 5 for all $u \in (\uh, \uh + \Delta_1)$ there exists an $\varepsilon(u) >0$ and a sequence $v_n \to \infty$ (depending on $u$) such that
$$r(\uh + \Delta_1, v_n) < r(u,v_n) < r(\uh, v_n) - \varepsilon(u) \quad \textnormal{ for all } n \in \N \;.$$
It thus follows that for $n$ large enough we have $\rd_v r(u,v_n) <0$ and then by Claim 3 $\rd_v r(u, v) <0$ for all $v \geq v_n$. The second part of the claim follows again from Claim 3.
\end{proof}

We now define $$u_T := \sup \{ u \in (\uh, 0) \; | \; \forall \;  u' \in (\uh, u)  \; \; \exists v_1(u') \geq v_0 \textnormal{ such that } \rd_v r(u',v) <0 \textnormal{ for all } v \geq v_1(u') \} \;.$$
By Claim 6 we know that $u_T > \uh$ and by Claim 3 that $\lim_{v \to \infty} r(u,v) = r_-(\infty)$ for all $u \in (\uh, u_T)$. We now show that for $\uh < \tilde{u} < u_T$ the region $(\uh, \tilde{u}) \times (v_0, \infty) \times \mathbb{S}^2 \subseteq M$ satisfies the assumptions of Theorem \ref{ThmOneNullSing}.

Analogously to definition \eqref{EqDefSGE} we define the surface gravity of the Cauchy horizon
$$\kappa_-(\infty) := \frac{r_-(\infty) - r_+(\infty)}{2r_-^2(\infty)} = \frac{1}{r_-^2(\infty)}\big(\varpi(\infty) - \frac{e^2}{r_-(\infty)}\big) $$
and also define $V(v) := -e^{\kappa_-(\infty) v}$ and $V_0 := V(v_0)$.
In the $(u,v)$-coordinates the metric $g$ takes the form 
\begin{equation}
\label{EqMetricUV}
g = - \frac{\Omega^2}{2} (du \otimes dv + dv \otimes du) + r^2(u,v) \mathring{\gamma} \;,
\end{equation}
and in the $(u,V)$-coordinates it then reads
\begin{equation*}
\begin{split}
g &= - \frac{\Omega^2}{2 \kappa_-(\infty) \cdot V} (du \otimes dV + dV \otimes du) + r^2(u,V) \mathring{\gamma} \\
&=: - \frac{\overline{\Omega}^2}{2} (du \otimes dV + dV \otimes du) + r^2(u,V) \mathring{\gamma} \;,
\end{split}
\end{equation*}
where we have defined $\overline{\Omega}^2 = - \frac{\Omega^2}{\kappa_-(\infty)} e^{-\kappa_-(\infty) v}$.
\begin{proposition} \label{PropRNV}
Let $\uh < \tilde{u} < u_T$. Then
\begin{enumerate}
\item The functions $r(u,V) : (\uh, \tilde{u}] \times (V_0, 0) \to (0, \infty)$ and $\overline{\Omega}^2(u,V) :  (\uh, \tilde{u}] \times (V_0, 0) \to (0, \infty)$ extend continuously to $ (\uh, \tilde{u}] \times (V_0, 0] $ as positive functions.
\item For all $u \in (\uh, \tilde{u})$ we have $\lim_{V \to 0} \partial_V r (u,V) = -\infty$.
\end{enumerate} 
It then follows that after a trivial rescaling of the $(u,V)$-coordinates the patch $(\uh, \tilde{u}) \times (V_0, 0) \times \mathbb{S}^2 $ of the RNV spacetime $M$ satisfies the assumptions\footnote{Note that we have already shown that $\partial_u r <0$.} of Theorem \ref{ThmOneNullSing} and is thus $C^{0,1}_{\loc}$-inextendible through $\{V=0\} \cap (\uh, \tilde{u}]$ in the sense of Theorem \ref{ThmOneNullSing} for any $\tilde{u} \in (\uh, u_T)$.
\end{proposition}

We will derive in fact much more precise asymptotics for $r$ and $\partial_V r$ as $V \to 0$.
\begin{proof}
\textbf{Step 1:} Let $u_0 \in (\uh, u_T)$ and let $v_1 \geq v_0$ be so large that $\rd_v r (u_0, v_1) < 0$. Let $\delta >0$ with $\uh < u_0 - \delta < u_0 + \delta < u_T$ be so small that $\rd_v r (u,v_1) <0$ for all $u \in [u_0 - \delta, u_0 + \delta]$. Furthermore we choose $v_2 \geq v_1$ so large that $$\frac{r(u_0 - \delta, v) - r_+(\infty)}{2r^2(u_0 - \delta, v)} \leq \frac{1}{2} \kappa_-(\infty) <0 \qquad \textnormal{ for all } v \geq v_2 \;,$$
which is possible since we know that $\lim_{v \to \infty} r(u,v) = r_-(\infty)  $ for $u \in (\uh, u_T)$.  Since we have $\rd_u r <0$ and $\rd_v r< 0$ in $[u_0 - \delta, u_0 + \delta] \times \{ v \geq v_1\}$ it follows that $$\frac{r(u , v) - r_+(\infty)}{2r^2(u, v)} \leq \frac{1}{2} \kappa_-(\infty) <0 \qquad \textnormal{ for all } (u,v) \in [u_0 - \delta, u_0 + \delta] \times \{v \geq v_2\}\;.$$
Recalling
\begin{equation}\label{EqDVRS}
\begin{split}
\rd_v r &= \frac{1}{2} \Big(1 - \frac{2 \varpi(v)}{r} + \frac{e^2}{r^2}\Big) \\
&= \frac{1}{2} \Big( 1 - \frac{2\varpi(\infty)}{r} + \frac{e^2}{r^2}\Big) + \frac{\varpi(\infty) - \varpi(v)}{r} \\
&= \frac{1}{2r^2} (r - r_+(\infty))(r - r_-(\infty)) + \frac{\beta v^{-p}}{r} \;,
\end{split}
\end{equation}
we obtain $$\rd_v r(u,v) \leq \frac{1}{2} \kappa_-(\infty) \big( r(u,v) - r_-(\infty)\big) + \frac{ \beta v^{-p}}{r_-(\infty)} \qquad \textnormal{ for all } (u,v) \in [u_0 - \delta, u_0 + \delta] \times \{v \geq v_2\}\;. $$
This gives $$r(u,v) - r_-(\infty) \leq \underbrace{e^{\frac{\kappa_-(\infty)}{2} (v - v_2)} \Big[ \int\limits_{v_1}^v e^{-\frac{\kappa_-(\infty)}{2} (v' - v_2)} \cdot \frac{\beta (v')^{-p}}{r_-(\infty)} \, dv'} + \big(r(u,v_2) - r_-(\infty)\big)\Big] $$
for all $ (u,v) \in [u_0 - \delta, u_0 + \delta] \times \{v \geq v_2\}$. A standard integration by parts argument\footnote{See for example Lemma 2.5 in \cite{FouSbi20}.} gives that the underbraced term equals $- \frac{2}{\kappa_-(\infty)} \frac{\beta}{r_-(\infty)} v^{-p} + \mathcal{O}(v^{-(p+1)})$ and thus we obtain
\begin{equation}
\label{EqUniformCR}
\big(r(u,v) - r_-(\infty)\big) \leq - \frac{2}{\kappa_-(\infty)} \frac{\beta}{r_-(\infty)} v^{-p} + C \cdot v^{-(p+1)}
\end{equation}
for all $ (u,v) \in [u_0 - \delta, u_0 + \delta] \times \{v \geq v_2\}$, where $C>0$. Note that the locally uniform convergence in $u$ in particular implies that $r(u,V)$ extends continuously to $V=0$ for $\uh < u < u_T$.

\textbf{Step 2:} We now derive asymptotics for $\rd_v r$ for $v \to \infty$. Let again $u_0 \in (\uh, u_T)$ and  $v_1 \geq v_0$ be so large that $\rd_v r (u_0, v) < 0$ for all $v \geq v_1$. Differentiating \eqref{EqDVRS} in $v$ we obtain
\begin{equation} \label{EqDVVR}
\begin{split}
\rd_v(\rd_v r) (u_0, v)&= \frac{1}{2} \rd_rf(\infty, r) \cdot \rd_v r (u_0, v)- \frac{\beta v^{-p}}{r^2} \rd_v r (u_0, v)- \frac{p \beta v^{-(p+1)}}{r} (u_0,v) \\
&\leq 2 \kappa_-(\infty) \cdot \rd_v r (u_0,v) - \frac{p \beta v^{-(p+1)}}{2r_-(\infty)}
\end{split}
\end{equation}
for all $v \geq v_2 $ with $v_2 \geq v_1$ large enough. This gives
\begin{equation*}
\begin{split}
\rd_v r(u_0, v) &\leq e^{2 \kappa_-(\infty)\cdot (v - v_2)} \Big[ - \int\limits_{v_2}^v e^{- 2 \kappa_-(\infty) \cdot (v' - v_2)} \cdot \frac{ p \beta (v')^{-(p+1)}}{2r_-(\infty)} \, dv' + \rd_vr(u_0, v_2)\Big] \\
&= \frac{p \beta}{4 \kappa_-(\infty) r_-(\infty)} v^{-(p+1)} + \mathcal{O}\big( v^{-(p+2)}\big) 
\end{split}
\end{equation*}
with the same standard integration by parts argument as before. Hence, we have $\rd_v r(u_0, v) \lesssim - v^{-(p+1)}$. Using $\rd_v r = - \kappa_-(\infty) e^{\kappa_- v} \rd_V r$ we obtain $$\rd_V r (u_0, v) = - \frac{1}{\kappa_-(\infty)} e^{-\kappa_-(\infty) v} \rd_v r(u_0,v) \lesssim - e^{-\kappa_-(\infty)v } v^{-(p+1)}\;,$$
from which $\lim_{V \to 0} \partial_V r (u_0,V) = -\infty$ follows.

\textbf{Step 3:} It remains to show that $\overline{\Omega}^2$ extends continuously to $V=0$. Using the inverse of \eqref{EqMetricRNVR}, i.e., the inverse of the metric $g$ in $(v,r)$-coordinates, we compute $g^{-1}(dr, dr) = f$. Using on the other hand the inverse of \eqref{EqMetricUV}, i.e., the inverse of $g$ in $(u,v)$-coordinates, we obtain $g^{-1}(dr, dr) = -\frac{4}{\Omega^2} \rd_u r \cdot \rd_v r$. Hence, we obtain $- \frac{4}{\Omega^2} \rd_u r \cdot \rd_v r = f$, and together with $\rd_v r = \frac{f}{2}$ this gives $$\Omega^2 = - 2\rd_u r \;.$$
Equation \eqref{EqDRU} directly gives $$\rd_v \log (-\rd_u r) = \frac{1}{r^2} \big(\varpi(v) - \frac{e^2}{r} \big) \;,$$
which integrates to
\begin{equation*}
\begin{split}
\log(-\rd_u r) (u,v) &= \log(-\rd_u r) (u,v_2) + \int\limits_{v_2}^v \frac{1}{r^2(u,v')} \big( \varpi(v') - \frac{e^2}{r(u,v')} \big) \, dv' \\
&= \log(-\rd_u r) (u,v_2) + \int\limits_{v_2}^v  \big( \kappa_-(\infty) + \epsilon( u,v') \big) \, dv' \;,
\end{split}
\end{equation*}
where $\epsilon( u,v')  = \frac{1}{r^2(u,v')} \big( \varpi(v') - \frac{e^2}{r(u,v')} \big)  - \kappa_-(\infty)$. 
We thus obtain $$-2 \rd_u r(u,v) = - 2\rd_u r(u,v_2) \cdot e^{\kappa_-(\infty) \cdot (v - v_2)} \cdot e^{\int_{v_2}^v \epsilon(u,v') \, dv'} \;,$$ which gives 
$$\Omega^2(u,v) e^{-\kappa_-(\infty)\cdot  v} = \Omega^2(u,v_2) \cdot e^{-\kappa_-(\infty) \cdot v_2} \cdot e^{\int_{v_2}^v \epsilon(u,v') \, dv'} \;.$$
Let now $u_0, \delta,$ and $v_2$ be as in Step 1. By \eqref{EqUniformCR} and $|\varpi(\infty) - \varpi(v')| \lesssim (v')^{-p}$ we have $|\epsilon ( u,v') | \leq C (v')^{-p}$ uniformly for all $ (u,v') \in [u_0 - \delta, u_0 + \delta] \times \{v' \geq v_2\}$, where $C>0$.  Note that $\epsilon$ is continuous and since $p>1$ it is uniformly integrable. Thus using dominated convergence it follows that $\overline{\Omega}^2(u,V) = - \frac{\Omega^2}{\kappa_-(\infty)} e^{-\kappa_- v}$ extends continuously to $V=0$, which concludes the proof.
\end{proof}

\begin{proposition} \label{PropTrapped}
We have $u_T = 0$.
\end{proposition}

\begin{proof}
We first show that we have $\lim_{v \to \infty} r(u,v) = r_-(\infty)$ for all $u \in (\uh, 0)$. Let $\Delta_1$ be as above Claim 6, i.e., such that $r(\uh + \Delta_1, v_0) \in (r_-(v_0), r_+(v_0))$. By Claim 6 it remains to show $\lim_{v \to \infty} r(u,v) = r_-(\infty)$ for $u \in [\uh + \Delta_1, 0)$. If we have $r(u,v_0) \in (r_-(v_0), r_+(v_0))$ then this follows from Claim 3 and if $r(u,v_0) \in (0, r_-(v_0)]$ then there are two possibilities: Either $r(u,v)$ enters the region $(r_-(v), r_+(v))$ at some later time, and then we use Claim 3 again, or we have $r(u,v) \in (0, r_-(v)]$ for all $v$. In this latter case it then follows that $\partial_v r(u,v) \geq 0$ for all $v$ and thus $\lim_{v \to \infty} r(u,v) \in (0, r_-(\infty)]$ exists. If we had $\lim_{v \to \infty} r(u,v) < r_+(\infty)$, then the right hand side of \eqref{EqVROnConstantU} would be positively bounded away from zero which would lead to the contradiction $\lim_{v \to \infty} r(u,v) = \infty$. Thus we have $\lim_{v \to \infty} r(u,v) = r_-(\infty)$.

We now show that for each $u_0 \in (\uh, 0)$ there exists a $v_1 \geq v_0$ such that $\rd_vr(u_0,v_1) <0$ (and thus also $\rd_v r(u_0,v) <0$ for all $v \geq v_1$ by Claim 3). Assume this were not the case, then $\rd_vr(u_0,v) \geq 0$ for all $v \geq v_0$. Using the evolution equation for $\rd_v r$ from the  first line of \eqref{EqDVVR} we can then estimate
$$\rd_v(\rd_v r)(u_0,v) \leq \frac{1}{2} \kappa_-(\infty) \rd_v r(u_0, v)-\frac{p \beta v^{-(p+1)}}{2 r_-(\infty)}\;,$$
for $v \geq v_2$ with $v_2$ sufficiently large. This, however, gives as before $$\rd_v r(u_0, v) \leq  \frac{p \beta }{\kappa_-(\infty) \cdot  r_-(\infty)} v^{-(p+1)} + \mathcal{O}(v^{-(p+2)}) \;,$$
which is negative for large $v$ and thus a contradiction to $\rd_vr (u_0,v) \geq 0$ for all $v \geq v_0$.
\end{proof}

\subsection{Dafermos-Luk-Oh spacetimes and a $C^{0,1}_{\loc}$-result of strong cosmic censorship} \label{SecDLO}

Our next application of the Theorems \ref{ThmOneNullSing} and \ref{ThmTwoNullSing} is to spacetimes arising from sufficiently small spherically symmetric perturbations of asymptotically flat two-ended subextremal Reissner-Nordstr\"om (RN) initial data for the Einstein-Maxwell-scalar field system, \emph{which we call Dafermos-Luk-Oh (DLO) spacetimes}. The mathematical study of spherically symmetric perturbations of the interior of a subextremal RN black hole was initiated by Dafermos in \cite{Daf03}, \cite{Daf05a} where he in particular showed that if the perturbation decays sufficiently quickly along the event horizon that the metric then extends continuously to the Cauchy horizon in a small neighbourhood of timelike infinity. This sufficiently quick decay for spherically symmetric perturbations of RN initial data was later established by Dafermos and Rodnianski in \cite{DafRod05}. Moreover, \cite{Daf05a} also showed that if the scalar field is in addition \emph{pointwise} polynomially lower bounded along the event horizon, then the Hawking mass and $-\rd_vr$ blow up on the Cauchy horizon near timelike infinity. Using a Cauchy stability argument and using the previous results Dafermos continued to show in \cite{Daf14} that for sufficiently small perturbations of two-ended subextremal RN initial data the Penrose diagram of the future development is as in Figure \ref{FigPenroseDLO}, i.e., the Cauchy horizon is given by a bifurcate null hypersurface as in exact subextremal RN. Moreover, the metric extends continuously to the Cauchy horizon. 
\begin{figure}[h]
\centering
 \def\svgwidth{8cm}
   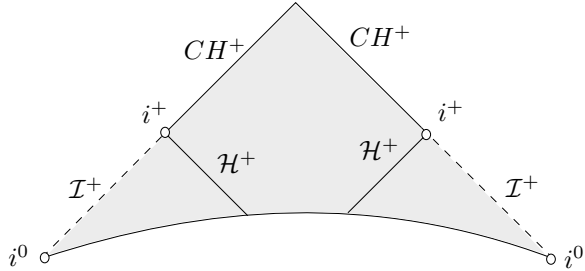
      \caption{Penrose diagram of the DLO spacetimes.} \label{FigPenroseDLO}
\end{figure}
In \cite{LukOh19II} Luk and Oh established that for \emph{generic} small spherically symmetric perturbations the scalar field obeys an \emph{integrated} lower bound along the event horizon. `Generic' means here that the initial data is contained in an open and dense set relative to appropriate topologies of the initial data space. The integrated lower bound being weaker than the pointwise lower bound assumed in \cite{Daf05a}, Luk and Oh continued to show in \cite{LukOh19I} that it is still strong enough to ensure the blow-up of the energy of the scalar field near the Cauchy horizon which they then use to infer the generic $C^2$-inextendibility of the DLO spacetimes.\footnote{It should be mentioned that in \cite{LukOh19I} Luk and Oh do not only consider globally small perturbations of subextremal RN initial data that lead to what we call here an DLO spacetime, but they also treat `admissible' large deviations from exact RN which possibly lead to the closing off of the Cauchy horizon, transitioning into a spacelike singularity in the interior. We do not discuss the latter case in detail here, but Theorem \ref{ThmOneNullSing} can still be used to infer the $C^{0,1}_{\loc}$-inextendibility of the Cauchy horizon in the sense of Theorem \ref{ThmOneNullSing}. However, we do not investigate the low-regularity inextendibility of the spacelike part of the boundary in this paper.}
In particular Luk and Oh also show the pointwise blow up of $-\rd_v r$ on the Cauchy horizon, which allows us to infer the following
\begin{theorem} \label{ThmInteriorDLO}
Consider a sufficiently small spherically symmetric perturbation of asymptotically flat two-ended subextremal RN initial data for the Einstein-Maxwell-scalar field system such that the future development is given by Figure \ref{FigPenroseDLO}. Assume the initial data is generic and small in the sense of \cite{LukOh19II}. Then the interior of the black hole\footnote{I.e. the set in Figure \ref{FigPenroseDLO} bounded by the two $CH^+$ to the future and the two $\Hp$ and part of the initial Cauchy hypersurface to the past.} is future $C^{0,1}_{\loc}$-inextendible.
\end{theorem}
\begin{proof}
This follows directly from Theorem \ref{ThmTwoNullSing}: By Theorem 5.5 and Remark 5.6 in \cite{LukOh19I} the metric in the double null gauge $(U_{\mathcal{CH}_2^+}, V_{\mathcal{CH}_1^+})$ extends continuously to the bifurcate Cauchy horizon $\{V_{\mathcal{CH}_1^+}=1\} \cup \{U_{\mathcal{CH}_2^+} = 1\}$, where we have used the notation from \cite{LukOh19I}. Moreover, (5.19) in \cite{LukOh19I} then implies that $\lim\limits_{V_{\mathcal{CH}_1^+} \to 1} \rd_{V_{\mathcal{CH}_1^+}} r = -\infty$, and similarly for the left Cauchy horizon. The statement that $\rd_{U_{\mathcal{CH}_2^+}} r <0$ near the right Cauchy horizon follows trivially from the Raychaudhuri equation (2.4) in \cite{LukOh19I} and the fact that $\rd_{U_{\mathcal{CH}_2^+}} r <0$ on the initial data hypersurface in the right exterior. Similarly for the left Cauchy horizon. Thus the assumptions of Theorem \ref{ThmTwoNullSing} are met (after a trivial reparametrisation of the null coordinates).
\end{proof}
In order to prove the $C^2$-formulation of strong cosmic censorship for DLO spacetimes, i.e., their generic $C^2$-inextendibility, Luk and Oh use the $C^2$-regularity of an assumed extension to infer that there is a \emph{radial null geodesic} which leaves the DLO spacetime and enters the extension. They then proceed by showing that if the null geodesic approached the Cauchy horizon then the Ricci-curvature contracted twice with the affine velocity vector of this null geodesic would blow up along the null geodesic -- and if the null geodesic approached null infinity or $i^+$, then $r$ would go to infinity along the null geodesic or the null geodesic would be affine complete. All of these are contradictions to the assumption that the radial null geodesic enters a $C^2$-extension. 

The following theorem improves the $C^2$-formulation of strong cosmic censorship for DLO spacetimes proven by Luk and Oh to a $C^{0,1}_{\loc}$-formulation.

\begin{theorem} \label{ThmSCC}
Consider a sufficiently small spherically symmetric perturbation of asymptotically flat two-ended subextremal RN initial data for the Einstein-Maxwell-scalar field system such that the future development is given by Figure \ref{FigPenroseDLO}. Assume the initial data is generic and small in the sense of \cite{LukOh19II}. Then the future development is future $C^{0,1}_{\loc}$-inextendible.
\end{theorem}

Let us remark that we follow here the setting of \cite{LukOh19II} and only establish the generic \emph{future inextendibility} of the physically more interesting \emph{future} development. One would expect\footnote{Private communication with J.\ Luk.} that with a bit more work one can extend \cite{LukOh19II} to also show that the past Cauchy horizon of the past development becomes singular \emph{generically}. The time-dual statement of Theorem \ref{ThmSCC} would then also yield the generic past $C^{0,1}_{\loc}$-inextendibility -- and thus the generic $C^{0,1}_{\loc}$-inextendibility of the whole spacetime, cf.\ Lemma \ref{LemFuturePastExt}.

Note that Theorem \ref{ThmSCC} goes beyond  Theorem \ref{ThmInteriorDLO} by the statement that also the black hole \emph{exterior} is future $C^{0,1}_{\loc}$-inextendible. The methods used by Luk and Oh for proving the future $C^2$-inextendibility of the exterior do not transfer to locally Lipschitz regularity. Theorem \ref{ThmGeod} together with Proposition \ref{PropLukOhComplete} in the appendix show however that the estimates obtained by Luk and Oh in \cite{LukOh19II} can be used to establish the timelike geodesic completeness of the exterior in the sense of Theorem \ref{ThmGeod}, which we can then use as an obstruction to $C^{0,1}_{\loc}$-extensions of the exterior.

\begin{proof}
The proof  is by contradiction, so let $\iota : M \hookrightarrow \tilde{M}$ be a future $C^{0,1}_{\loc}$-extension and let $\tilde{p} \in \partial^+ \iota (M)$, where $(M,g)$ denotes the generic DLO spacetime under consideration. By Proposition \ref{PropBoundaryChart} there exists a chart $\tilde{\varphi} : \tilde{U} \to(-\varepsilon_0, \varepsilon_0) \times  (-\varepsilon_1, \varepsilon_1)^{d} =: R_{\varepsilon_0, \varepsilon_1}$ as in Proposition \ref{PropBoundaryChart} with $\delta >0$ so small that all vectors in $C^+_{\nicefrac{5}{6}}$ are future directed timelike, all vectors in $C^-_{\nicefrac{5}{6}}$ are past directed timelike, and all vectors in $C^c_{\nicefrac{5}{8}}$ are spacelike. By Proposition \ref{PropGeodesicBoundaryChart} there exists a future directed timelike  geodesic $\gamma : [-\mu,0) \to M$ that is future inextendible in $M$ and such that $\tilde{\varphi} \circ \iota \circ \gamma :[-\mu,0) \to \Reps$ maps below the graph of $f$ and has a future endpoint on the graph of $f$. We also define $\tilde{\gamma} := \iota \circ \gamma$. 

If there exists a point on $\gamma$ that lies on the event horizons or in the interior, then we clearly have that $\gamma$ approaches to the future the bifurcate Cauchy horizon. Then the restriction of $\iota : M \hookrightarrow \tilde{M}$ to the black hole interior gives rise in particular to a future $C^{0,1}_{\loc}$-extension of the black hole interior\footnote{Take a future directed smooth timelike curve of constant $x_0$ in the chart $\tilde{\varphi}$ starting at a point of $\tilde{\gamma}$ which is contained in the black hole interior. This curve then intersects the graph of $f$ in, as a result by definition, a future boundary point of the interior.}, which is a contradiction to Theorem \ref{ThmInteriorDLO}.

If $\gamma$ starts in one of the black hole exteriors then by Proposition \ref{PropLukOhComplete} $\gamma$ is future complete or crosses the event horizon. We have already ruled out that it can cross the event horizon, so  it must be future complete. We then reparametrise $\tilde{\gamma}$ by the $x_0$-coordinate in the chart $\tilde{\varphi}$ to obtain a curve which we denote again by $\tilde{\gamma} : [s_0, s_1) \to \tilde{U} \subseteq \tilde{M}$, with $-\varepsilon_0 < s_0 < s_1 < \varepsilon_0$. Using the uniform bound $|\tilde{g}_{\mu \nu} - m_{\mu \nu}| < \delta$ and \eqref{EqUniformBoundCausalCurve} we compute $$\int_{s_0}^{s_1} \sqrt{- \tilde{g}(\dot{\tilde{\gamma}}(s), \dot{\tilde{\gamma}}(s)} \, ds \leq 2 \varepsilon_0 \cdot C(\delta, ||\dot{\tilde{\gamma}}||_{\R^{4}}) < \infty\;,$$
which is a contradiction to the future completeness of $\gamma$. Thus $(M,g)$ is future $C^{0,1}_{\loc}$-inextendible.
\end{proof}

\appendix

\section{Timelike geodesic completeness of the exterior of spherically symmetric black holes} 

Here we present a very general criterion for the exterior of a spherically symmetric black hole spacetime to be future timelike geodesically complete in the sense that any future directed and future inextendible timelike geodesic that starts in the exterior is  future complete or enters the black hole region. This result is in particular needed to prove the future $C^{0,1}_{\loc}$-inextendibility of the Dafermos-Luk-Oh (DLO) spacetimes in Theorem \ref{ThmSCC}. Note that the timelike geodesic completeness of the exterior of the DLO spacetimes does not seem to follow directly from \cite{LukOh19II} since the weak stability estimates in Theorem 8.6, in particular (8.37), degenerate near the horizon and do not give asymptotic stability. Thus, a perturbative argument, deducing the geodesic completeness of the exterior of the DLO spacetimes from the geodesic completeness of the exterior of a subextremal Reissner-Nordstr\"om solution, seems to fail. We circumvent this obstacle by changing from the double null coordinates used in \cite{LukOh19II} to $(v,r)$-coordinates, in which the weak stability estimates give orbital stability. Supplemented by a few further conditions in the exterior, see \eqref{EqAlign1} - \eqref{EqAlign5}, which can be verified directly from the estimates obtained in \cite{LukOh19II}, this suffices to infer the geodesic completeness of the exterior.
\newline

We consider a general spherically symmetric spacetime $(M,g)$ in $(v,r)$-coordinates, where $$M = \underbrace{[0, \infty) \times (0, \infty)}_{=: Q} \times \mathbb{S}^2$$ with standard $(v,r)$-coordinates on the first two factors and 
\begin{equation}
\label{EqMetricVR}
g = - f(v,r) \, dv^2 + \frac{h(v,r)}{2} \, \big( dv \otimes dr + dr \otimes dv\big) + r^2 \, \mathring{\gamma} \;, 
\end{equation}
where $f , h : Q \to \R$ are smooth functions and $\mathring{\gamma}$ is the standard round metric on $\mathbb{S}^2$. We declare a time-orientation on $(M,g)$ by stipulating that $- \rd_r$ is future directed null. Moreover, let $r_{\Hp} : [0, \infty) \to (c_{\Hp},\infty)$ be a continuous function which is uniformly positive, i.e., we assume $c_{\Hp} >0$, and define a curve $\sigma_{\Hp} : [0, \infty) \to Q$ by $\sigma_{\Hp} (v) := (v, r_{\Hp}(v))$.\footnote{It helps to think of the curve $\sigma_{\Hp}$ as the event horizon of $M$, but this is by no means an assumption in the below theorem.} The continuous curve $\sigma_{\Hp}$ then separates $M $ into an exterior region $$\Mext := \{(v,r, \omega) \in M \; | \; r > r_{\Hp}(v)\}$$ and an interior region $M_{\mathrm{int}} := \{(v,r, \omega) \in M \; | \; r < r_{\Hp}(v)\}$.
\begin{figure}[h]
\centering
 \def\svgwidth{5cm}
   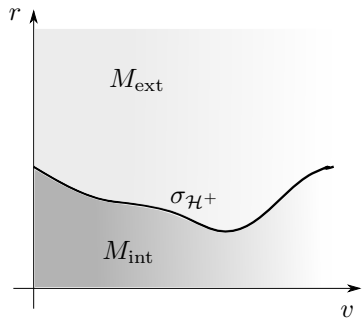
      \caption{The interior and exterior regions of $M$} \label{FigGeod}
\end{figure}

\begin{theorem} \label{ThmGeod}
Let $(M,g)$ and $r_{\Hp}$ be as above. Assume moreover that the metric functions $f$ and $h$ satisfy \underline{on $\Mext$}
\begin{equation}
\label{EqThmGeodAss}
\begin{aligned}
&| f| \leq C_f \\
&0 < c_h \leq h \leq C_h \\
&0 \leq \rd_r h \\
&f \cdot \rd_r h - \rd_r f \cdot h \leq 0 \;,
\end{aligned}
\end{equation}
where $C_f, c_h, C_h$ are constants.
Then the following holds: 

Let $\gamma : [0,b) \to M$ be an affinely parametrised future directed and future inextendible timelike geodesic with $\gamma(0) \in \Mext$, where $b \in (0, \infty]$. Then $\gamma$ is future complete or there is an $s_0 \in [0,b)$ with $\gamma(s_0) \in \mathrm{Im}(\sigma_{\Hp})$.\footnote{This is not an ``either -- or''. Both possibilities are allowed.}
\end{theorem}
\begin{remark} \label{RemAppe}
Let us emphasise here again that the assumptions \eqref{EqThmGeodAss} are only required on the exterior $\Mext$ of $M$ -- no assumptions are required on the metric coefficients in the interior $M_{\mathrm{int}}$ of $M$. Indeed, it suffices to make the assumptions on $\Mext \cap I^+(\gamma(0),M)$ as will be immediate from the proof. We also mention that the assumption $0 \leq \rd_r h$ may be replaced by $| \rd_r h| \leq C$ as follows from a slight modification of the proof, but this is not needed in this paper.
\end{remark}
\begin{proof}
Let $\gamma : [0,b) \to M$ be an affinely parametrised future directed and future inextendible timelike geodesic with $\gamma(0) \in \Mext$ as in the theorem and assume contrary to the statement of the theorem that $\gamma$ is completely contained in $\Mext$ and that $b< \infty$. Using the spherical symmetry of $(M,g)$ we can assume without loss of generality that $\gamma$ lies in the equatorial plane $\{ \theta = \frac{\pi}{2}\}$, where $\theta, \varphi$ are the standard spherical coordinates on $\mathbb{S}^2$. Since $\rd_\varphi$ is a Killing vector field we then have 
\begin{equation}
\label{EqVelPL}
L := g(\dot{\gamma}, \rd_\varphi) = r^2 \dot{\varphi} \;.
\end{equation}
 for some constant $L \in \R$.
We can moreover assume that the  velocity vector of $\gamma$ is normalised so that we have $$-1 = g(\dot{\gamma}, \dot{\gamma}) = -f \dot{v}^2 + h \dot{v} \dot{r} + r^2 \dot{\varphi}^2 = -f \dot{v}^2 + h \dot{v} \dot{r} + \frac{L^2}{r^2} \;.$$
This gives in particular
\begin{equation}
\label{EqVdotRdot}
\dot{v} \dot{r} = \frac{-1 + f\dot{v}^2 - \frac{L^2}{r^2}}{h} \;.
\end{equation}
The inverse metric of \eqref{EqMetricVR} is given by
\begin{equation*}
g^{-1} = \begin{pmatrix} 0 & \frac{2}{h} & 0 & 0 \\ \frac{2}{h} & \frac{4f}{h^2} & 0 & 0 \\
0 & 0 & r^{-2} & 0 \\ 0 & 0 & 0 & r^{-2} \sin^{-2}\theta \end{pmatrix} \;.
\end{equation*}
Since $h$ is assumed positive this gives that $dv^\sharp = \frac{2}{h} \rd_r$ is past directed and thus $\dot{v} = dv(\dot{\gamma}) = g(\frac{2}{h} \rd_r, \dot{\gamma}) >0$.
One also easily computes $$\Gamma^v_{vv} = \frac{\rd_v h + \rd_r f}{h} \;, \qquad \Gamma^v_{\theta \theta} = - \frac{2r}{h} \;, \qquad \Gamma^v_{\varphi \varphi} = -\frac{2r \sin^2 \theta}{h} \;.$$ All other Christoffel symbols of the form $\Gamma^v_{\cdot \cdot}$ vanish. Recalling that we have $\theta = \frac{\pi}{2}$ along $\gamma$ the geodesic equation gives
\begin{equation*}
\begin{split}
0 &= \ddot{v} + \Gamma^v_{vv} \dot{v}^2 + \Gamma^v_{\varphi \varphi} \dot{\varphi}^2 \\
&= \ddot{v} + \rd_v \log h \cdot \dot{v}^2 + \frac{\rd_r f}{h} \cdot \dot{v}^2 - \frac{2L^2}{h r^3} \;.
\end{split}
\end{equation*}
We now use $\frac{d}{ds} \log h = \dot{v} \rd_v \log h + \dot{r} \rd_r \log h$ and \eqref{EqVdotRdot} to further compute
\begin{equation*}
\begin{split}
\ddot{v} &= \frac{2L^2}{h r^3} - \dot{v} \frac{d}{ds} \log h + \dot{v} \dot{r} \rd_r \log h - \frac{\rd_r f}{h} \dot{v}^2 \\
&= \frac{2L^2}{h r^3} - \dot{v} \frac{d}{ds} \log h - \frac{1}{h} \big(1 + \frac{L^2}{r^2}\big) \rd_r \log h + \underbrace{\frac{f}{h} \rd_r \log h \cdot \dot{v}^2 - \frac{\rd_r f}{h} \dot{v}^2}_{= \frac{1}{h^2}(f \rd_r h - \rd_r f \cdot h) \dot{v}^2} \;.
\end{split}
\end{equation*}
Using the third and fourth assumptions in \eqref{EqThmGeodAss} this gives 
\begin{equation}
\label{EqVDD}
\ddot{v} \leq \frac{2L^2}{h r^3} - \dot{v} \frac{d}{ds} \log h \;.
\end{equation}
Using that the solution of $\ddot{x}(s) = - \dot{x}(s) \frac{d}{ds} \log h\big((v(s), r(s)\big)$, with $v(s)= \gamma^v(s)$ and $r(s) = \gamma^r(s)$ given, is given by $\dot{x}(s) = \frac{h\big(v(0), r(0)\big)}{h\big(v(s), r(s)\big)} \cdot \dot{x}(0)$, we thus obtain
\begin{equation}
\label{EqUBVD}
\dot{v}(s) \leq \frac{h(0)}{h(s)} \Big[ \int\limits_0^s \frac{h(s')}{h(0)} \cdot \frac{2L^2}{h(s') r^3(s')} \, ds' + \dot{v}(0)\Big]\qquad  \textnormal{ for all } 0 \leq s  < b < \infty \;.
\end{equation}
Using that $r$ is bounded away from $0$ by $c_{\Hp}$ in $\Mext$, the uniform bounds on $h$ from \eqref{EqThmGeodAss}, and $0<b< \infty$, this gives 
\begin{equation}
\label{EqUBVDResult}
0 < \dot{v}(s) \leq C < \infty \qquad \textnormal{ for all } 0 \leq s < b \;.
\end{equation}

\textbf{Claim:} $\lim_{s \to b} \gamma(s) \in M$ exists.

The validity of the claim is of course in direct contradiction to the future inextendibility of $\gamma$, so it remains to prove the claim. We immediately obtain from \eqref{EqUBVDResult} that $\lim_{s \to b} v(s) \in [0, \infty)$ exists. Moreover, \eqref{EqVelPL} together with $r > c_{\Hp}>0$ gives $|\dot{\varphi}| \leq C < \infty$, which again yields that $\lim_{s \to b} \varphi(s) \in \mathbb{S}^1$ exists. In order to show that $\lim_{s \to b} r(s) \in (0, \infty)$ exists we distinguish two cases:
\begin{enumerate}
\item \emph{$\dot{v}$ is lower bounded, i.e. there exists $c>0$ with $0 < c \leq \dot{v} \leq C$.} 
 We obtain from \eqref{EqVdotRdot} that 
 \begin{equation}
 \label{EqDotREq}
 \dot{r} = \frac{-1 + f\dot{v}^2 - \frac{L^2}{r^2}}{ \dot{v} h}\;,
 \end{equation}
  which, together with  $r > c_{\Hp}$ along $\gamma$ and \eqref{EqThmGeodAss} gives $|\dot{r}| \leq C$. This then implies that $\lim_{s \to b} r(s) \in [c_{\Hp}, \infty)$ exists.
 \item If the first case is not met, then \emph{there exists a sequence $[0, b) \ni s_n \to b$ for $n \to \infty$ with $\lim_{n \to \infty} \dot{v}(s_n) = 0$.} We show that we then actually have $\lim_{s \to b} \dot{v}(s) = 0$. As before \eqref{EqVDD} implies 
 \begin{equation*}
 \begin{split}
  \dot{v}(s) &\leq  \frac{h(s_n)}{h(s)} \Big[ \int \limits_{s_n}^s \frac{2L^2}{h(s_n) r^3(s')} \, ds' + \dot{v}(s_n)\Big] \\
  &\leq C \big[ C \cdot |s - s_n| + \dot{v}(s_n)\big] \qquad \textnormal{ for all } s_n \leq s < b \;.
  \end{split}
  \end{equation*}
Given $\varepsilon>0$ we can now choose $n$ large enough such that $0< \dot{v}(s) < \varepsilon$ for all $ s_n < s < b$. 

It now follows from \eqref{EqDotREq} together with \eqref{EqThmGeodAss} that $\dot{r}(s) \to - \infty$ for $s \to b$, which, together with $r > c_{\Hp}$ along $\gamma$, implies that $\lim_{s \to b} r(s) \in [c_{\Hp}, \infty)$ exists. This concludes the proof.
\end{enumerate}
\end{proof}

Let now $(M,g)$ be a spherically symmetric spacetime in double null coordinates $(u,v)$, i.e., let $$M = \underbrace{\R \times [0, \infty)}_{=:Q} \times \mathbb{S}^2$$ with canonical $(u,v)$-coordinates on the first two factors and $$g = - \frac{\Omega^2}{2} \big( du \otimes dv + dv \otimes du\big) + r^2 \, \mathring{\gamma} \;,$$ where $\Omega : Q \to (0, \infty)$ and $r : Q \to (0, \infty)$ are smooth functions. We have $dr = \rd_u r \cdot du + \rd_v r \cdot dv$. \emph{Assuming that $\rd_u r $ vanishes nowhere} we can globally transform into $(v,r)$-coordinates: we obtain  $$du = \frac{dr - \rd_v r \cdot dv}{\rd_u r} $$ and thus $$g =  \frac{\Omega^2 \rd_v r}{ \rd_u r} \, dv^2- \frac{\Omega^2}{2 \rd_u r} \, \big( dv \otimes dr + dr \otimes dv\big) + r^2 \, \mathring{\gamma} \;.$$ Recalling the definition of the Hawking mass $m : Q \to \R$ by $1 - \frac{2m}{r} := g^{-1}(dr, dr) = - \frac{4}{\Omega^2} \rd_u r \cdot \rd_v r$, introducing the definition $\kappa := \frac{\rd_v r}{1 - \frac{2m}{r}} = - \frac{\Omega^2}{4 \rd_u r}$ and comparing with \eqref{EqMetricVR} gives
\begin{equation} \label{Eq11}
f = 4 \kappa \cdot \rd_v r \qquad \textnormal{ and } \qquad h = 4 \kappa \;.
\end{equation}
We moreover note that $$\frac{\rd}{\rd u} \Big|_v = \frac{\rd r}{\rd u} \Big|_v  \frac{\rd}{\rd r}\Big|_v \;.$$
Thus we obtain \begin{equation} \label{Eq22}\frac{\rd}{\rd r}\Big|_v h = \frac{4}{\rd_u r} \rd_u \kappa 
\end{equation} and 
\begin{equation} \label{Eq33}
f \cdot \rd_r h - \rd_r f \cdot h = \frac{16 \kappa}{\rd_u r} \big( \rd_v r \cdot \rd_u \kappa - \rd_u ( \kappa \rd_v r)\big) = - \frac{16 \kappa^2}{\rd_u r} \cdot \rd_u \rd_v r \;.\end{equation}

\begin{proposition} \label{PropLukOhComplete}
Consider a Cauchy hypersurface of a two-ended subextremal Reissner-Nordstr\"om black hole and a sufficiently small spherically symmetric perturbation of the Cauchy data under the Einstein-Maxwell-scalar field system as in \cite{LukOh19II}. Consider one of the exteriors of the arising future development and let $(u,v) \in \mathcal{D} \subseteq \R \times [1,\infty)$ be null coordinates normalised as in Theorem 8.6 of \cite{LukOh19II} covering the exterior under consideration.  We then have
\begin{align}
&0 < c \leq \kappa \leq C  \label{EqAlign1}\\ 
&0< \rd_v r \leq C \label{EqAlign2}\\
&\rd_u r < 0 \label{EqAlign3}\\
&\rd_u \kappa \leq 0 \label{EqAlign4}\\
&\rd_u \rd_v r \leq 0 \label{EqAlign5}
\end{align}
in the exterior.
In particular it then follows from Theorem \ref{ThmGeod} that a future directed future inextendible timelike geodesic starting in the exterior is future complete or crosses the event horizon.
\end{proposition}
\begin{proof}
\eqref{EqAlign1} is directly given by (8.35) or (8.44) in \cite{LukOh19II}. The positivity of $\rd_v r$ follows since the exterior is non-trapped and the upper bound in \eqref{EqAlign2} follows from (8.46) by comparison with an exact subextremal Reissner-Nordstr\"om black hole. The condition \eqref{EqAlign3} follows from the Raychauduri equation (2.5) in \cite{LukOh19II} and it being satisfied along the initial data hypersurface. To show \eqref{EqAlign4} we recall that $-\frac{\rd_u \kappa}{\kappa^2} = \rd_u \kappa^{-1} = \frac{4 r (\rd_u \phi)^2}{\Omega^2} \geq 0$ by the definition $\kappa = - \frac{\Omega^2}{4 \rd_u r}$ of $\kappa$ and the Raychauduri equation (2.5) in \cite{LukOh19II}. Finally by (2.8) in \cite{LukOh19II} we have $$\rd_u \rd_v r  = \frac{2 (\varpi - \frac{e^2}{r}) \cdot \rd_u r \cdot \rd_v r}{r^2 \cdot (1 - \frac{2m}{r})} $$ where $\varpi = m + \frac{e^2}{2r}$ and $e$ is the charge of the black hole, see \cite{LukOh19II}. Since $\rd_v r >0$ in the exterior and by \eqref{EqAlign3} we also have $1 - \frac{2m}{r} >0$ in the exterior by definition of the Hawking mass $m$. In order to show \eqref{EqAlign5} it thus suffices to show $\varpi - \frac{e^2}{r} \geq0$ in the exterior, which is equivalent to $\varpi \cdot r \geq e^2$. Note that this is trivially satisfied for $r$ sufficiently large by (8.43) in \cite{LukOh19II}. For $r$ bounded away from $\infty$ we have $|r - \overline{r}| \leq C \varepsilon$ by (8.45) in \cite{LukOh19II}, where we choose  the overlined quantities to denote quantities in an exact subextremal Reissner-Nordstr\"om black hole. Using also (8.43) in \cite{LukOh19II} and the same trivial statement for the charges $e$, $\overline{e}$, which are non-dynamic, we compute for $r $ bounded away from infinity
\begin{equation*}
\begin{split}
\varpi \cdot r & \geq \overline{\varpi} \cdot \overline{r} - C \varepsilon\\
&\geq \overline{\varpi} \cdot \overline{r}_+ - C \varepsilon \\
&= \overline{\varpi}(\overline{\varpi} + \sqrt{ \overline{\varpi}^2 - \overline{e}^2}) - C \varepsilon \\
&> \overline{e}^2 + \overline{\varpi} \sqrt{ \overline{\varpi}^2 - \overline{e}^2} - C \varepsilon \\
&\geq e^2 + \overline{\varpi} \sqrt{ \overline{\varpi}^2 - \overline{e}^2} - 2C \varepsilon \\
&\geq e^2 
\end{split}
\end{equation*}
for $\varepsilon>0$ sufficiently small.

Since we have $\rd_u r <0$ we can now go over to $(v,r)$-coordinates and extend the metric arbitrarily to the full domain $[1, \infty) \times (0, \infty) \times \mathbb{S}^2$. For the curve $\sigma_{\Hp}$ in Theorem \ref{ThmGeod} we choose the event horizon and by \eqref{Eq11}, \eqref{Eq22}, and \eqref{Eq33} the conditions \eqref{EqThmGeodAss} are satisfied in the exterior to the future of any point in the exterior, see also Remark \ref{RemAppe}. It thus follows that any future directed future inextendible timelike geodesic starting in the exterior is future complete or crosses the event horizon.
\end{proof}

\bibliographystyle{acm}
\bibliography{Bibly}

\end{document}

%% file: Fig1.pdf_tex
\begingroup%
  \makeatletter%
  \providecommand\color[2][]{%
    \errmessage{(Inkscape) Color is used for the text in Inkscape, but the package 'color.sty' is not loaded}%
    \renewcommand\color[2][]{}%
  }%
  \providecommand\transparent[1]{%
    \errmessage{(Inkscape) Transparency is used (non-zero) for the text in Inkscape, but the package 'transparent.sty' is not loaded}%
    \renewcommand\transparent[1]{}%
  }%
  \providecommand\rotatebox[2]{#2}%
  \ifx\svgwidth\undefined%
    \setlength{\unitlength}{349.02958907bp}%
    \ifx\svgscale\undefined%
      \relax%
    \else%
      \setlength{\unitlength}{\unitlength * \real{\svgscale}}%
    \fi%
  \else%
    \setlength{\unitlength}{\svgwidth}%
  \fi%
  \global\let\svgwidth\undefined%
  \global\let\svgscale\undefined%
  \makeatother%
  \begin{picture}(1,0.72138271)%
    \put(0,0){\includegraphics[width=\unitlength,page=1]{Fig1.pdf}}%
    \put(0.59312414,0.47574245){\color[rgb]{0,0,0}\makebox(0,0)[lb]{\smash{$\gamma(-\mu)$}}}%
    \put(0,0){\includegraphics[width=\unitlength,page=2]{Fig1.pdf}}%
    \put(0.52436199,0.3152976){\color[rgb]{0,0,0}\makebox(0,0)[lb]{\smash{$\gamma$}}}%
    \put(-0.0060883,0.66565677){\color[rgb]{0,0,0}\makebox(0,0)[lb]{\smash{$\hat{t}$}}}%
    \put(0.91401377,0.01077988){\color[rgb]{0,0,0}\makebox(0,0)[lb]{\smash{$\chi$}}}%
    \put(0.42613051,0.13193206){\color[rgb]{0,0,0}\makebox(0,0)[lb]{\smash{$\gamma(-\mu_0)$}}}%
    \put(0,0){\includegraphics[width=\unitlength,page=3]{Fig1.pdf}}%
    \put(0.40274927,0.38637742){\color[rgb]{0,0,0}\makebox(0,0)[lb]{\smash{$\ell$}}}%
    \put(0.17026803,0.383103){\color[rgb]{0,0,0}\makebox(0,0)[lb]{\smash{$\underline{\ell}$}}}%
    \put(0,0){\includegraphics[width=\unitlength,page=4]{Fig1.pdf}}%
  \end{picture}%
\endgroup%

%% file: Fig2.pdf_tex
\begingroup%
  \makeatletter%
  \providecommand\color[2][]{%
    \errmessage{(Inkscape) Color is used for the text in Inkscape, but the package 'color.sty' is not loaded}%
    \renewcommand\color[2][]{}%
  }%
  \providecommand\transparent[1]{%
    \errmessage{(Inkscape) Transparency is used (non-zero) for the text in Inkscape, but the package 'transparent.sty' is not loaded}%
    \renewcommand\transparent[1]{}%
  }%
  \providecommand\rotatebox[2]{#2}%
  \ifx\svgwidth\undefined%
    \setlength{\unitlength}{349.02958907bp}%
    \ifx\svgscale\undefined%
      \relax%
    \else%
      \setlength{\unitlength}{\unitlength * \real{\svgscale}}%
    \fi%
  \else%
    \setlength{\unitlength}{\svgwidth}%
  \fi%
  \global\let\svgwidth\undefined%
  \global\let\svgscale\undefined%
  \makeatother%
  \begin{picture}(1,0.72138271)%
    \put(0,0){\includegraphics[width=\unitlength,page=1]{Fig2.pdf}}%
    \put(0.59312413,0.47574245){\color[rgb]{0,0,0}\makebox(0,0)[lb]{\smash{$\gamma(-\mu)$}}}%
    \put(0,0){\includegraphics[width=\unitlength,page=2]{Fig2.pdf}}%
    \put(0.52436198,0.3152976){\color[rgb]{0,0,0}\makebox(0,0)[lb]{\smash{$\gamma$}}}%
    \put(-0.00608831,0.66565678){\color[rgb]{0,0,0}\makebox(0,0)[lb]{\smash{$\hat{t}$}}}%
    \put(0.91401376,0.01077989){\color[rgb]{0,0,0}\makebox(0,0)[lb]{\smash{$\chi$}}}%
    \put(0.37374034,0.08609068){\color[rgb]{0,0,0}\makebox(0,0)[lb]{\smash{$\gamma(-\mu_0)$}}}%
    \put(0,0){\includegraphics[width=\unitlength,page=3]{Fig2.pdf}}%
    \put(0.40602369,0.33398726){\color[rgb]{0,0,0}\makebox(0,0)[lb]{\smash{$\sigma$}}}%
  \end{picture}%
\endgroup%

%% file: Fig3.pdf_tex
\begingroup%
  \makeatletter%
  \providecommand\color[2][]{%
    \errmessage{(Inkscape) Color is used for the text in Inkscape, but the package 'color.sty' is not loaded}%
    \renewcommand\color[2][]{}%
  }%
  \providecommand\transparent[1]{%
    \errmessage{(Inkscape) Transparency is used (non-zero) for the text in Inkscape, but the package 'transparent.sty' is not loaded}%
    \renewcommand\transparent[1]{}%
  }%
  \providecommand\rotatebox[2]{#2}%
  \ifx\svgwidth\undefined%
    \setlength{\unitlength}{349.02958907bp}%
    \ifx\svgscale\undefined%
      \relax%
    \else%
      \setlength{\unitlength}{\unitlength * \real{\svgscale}}%
    \fi%
  \else%
    \setlength{\unitlength}{\svgwidth}%
  \fi%
  \global\let\svgwidth\undefined%
  \global\let\svgscale\undefined%
  \makeatother%
  \begin{picture}(1,0.72138271)%
    \put(0,0){\includegraphics[width=\unitlength,page=1]{Fig3.pdf}}%
    \put(0.59312413,0.47574245){\color[rgb]{0,0,0}\makebox(0,0)[lb]{\smash{$\gamma(-\mu)$}}}%
    \put(0,0){\includegraphics[width=\unitlength,page=2]{Fig3.pdf}}%
    \put(0.52436198,0.3152976){\color[rgb]{0,0,0}\makebox(0,0)[lb]{\smash{$\gamma$}}}%
    \put(-0.00608831,0.66565678){\color[rgb]{0,0,0}\makebox(0,0)[lb]{\smash{$\hat{t}$}}}%
    \put(0.91401376,0.01077989){\color[rgb]{0,0,0}\makebox(0,0)[lb]{\smash{$\chi$}}}%
    \put(0.37374034,0.08609068){\color[rgb]{0,0,0}\makebox(0,0)[lb]{\smash{$\gamma(-\mu_0)$}}}%
    \put(0,0){\includegraphics[width=\unitlength,page=3]{Fig3.pdf}}%
    \put(0.35690795,0.31761534){\color[rgb]{0,0,0}\makebox(0,0)[lb]{\smash{$\sigma$}}}%
    \put(0,0){\includegraphics[width=\unitlength,page=4]{Fig3.pdf}}%
    \put(0.18991429,0.63850501){\color[rgb]{0,0,0}\makebox(0,0)[lb]{\smash{$\textnormal{bounded acceleration}$}}}%
  \end{picture}%
\endgroup%

%% file: ChartTChi.pdf_tex
\begingroup%
  \makeatletter%
  \providecommand\color[2][]{%
    \errmessage{(Inkscape) Color is used for the text in Inkscape, but the package 'color.sty' is not loaded}%
    \renewcommand\color[2][]{}%
  }%
  \providecommand\transparent[1]{%
    \errmessage{(Inkscape) Transparency is used (non-zero) for the text in Inkscape, but the package 'transparent.sty' is not loaded}%
    \renewcommand\transparent[1]{}%
  }%
  \providecommand\rotatebox[2]{#2}%
  \ifx\svgwidth\undefined%
    \setlength{\unitlength}{536.08047605bp}%
    \ifx\svgscale\undefined%
      \relax%
    \else%
      \setlength{\unitlength}{\unitlength * \real{\svgscale}}%
    \fi%
  \else%
    \setlength{\unitlength}{\svgwidth}%
  \fi%
  \global\let\svgwidth\undefined%
  \global\let\svgscale\undefined%
  \makeatother%
  \begin{picture}(1,0.70414322)%
    \put(0,0){\includegraphics[width=\unitlength,page=1]{ChartTChi.pdf}}%
    \put(0.00014114,0.05814479){\color[rgb]{0,0,0}\makebox(0,0)[lb]{\smash{$\hat{t}$}}}%
    \put(0.8848697,0.62735573){\color[rgb]{0,0,0}\makebox(0,0)[lb]{\smash{$\chi$}}}%
    \put(0.76974835,0.66572951){\color[rgb]{0,0,0}\makebox(0,0)[lb]{\smash{$\chi_0$}}}%
    \put(0.39027447,0.66786137){\color[rgb]{0,0,0}\makebox(0,0)[lb]{\smash{$\frac{\chi_0}{2}$}}}%
    \put(0.41585698,0.53781693){\color[rgb]{0,0,0}\makebox(0,0)[lb]{\smash{$\gamma(-\mu_0)$}}}%
    \put(0.26449377,0.49517939){\color[rgb]{0,0,0}\makebox(0,0)[lb]{\smash{$\gamma(r)$}}}%
    \put(0.45636266,0.3374206){\color[rgb]{0,0,0}\makebox(0,0)[lb]{\smash{$\Gamma_{\mathrm{right},\mu_0}\big(\frac{-\mu + r}{2};r\big)$}}}%
    \put(0,0){\includegraphics[width=\unitlength,page=2]{ChartTChi.pdf}}%
    \put(0.14297683,0.14768359){\color[rgb]{0,0,0}\makebox(0,0)[lb]{\smash{$\gamma(-\mu)$}}}%
    \put(0.73350647,0.18392556){\color[rgb]{0,0,0}\makebox(0,0)[lb]{\smash{$\ell_{\kappa}$}}}%
    \put(0.56935211,0.14128797){\color[rgb]{0,0,0}\makebox(0,0)[lb]{\smash{$\underline{\ell}_{\kappa}$}}}%
    \put(0.2517025,0.04748546){\color[rgb]{0,0,0}\makebox(0,0)[lb]{\smash{$\partial_u$}}}%
    \put(0.40946145,0.04108985){\color[rgb]{0,0,0}\makebox(0,0)[lb]{\smash{$\partial_v$}}}%
    \put(0,0){\includegraphics[width=\unitlength,page=3]{ChartTChi.pdf}}%
  \end{picture}%
\endgroup%

%% file: BoundaryChart.pdf_tex
\begingroup%
  \makeatletter%
  \providecommand\color[2][]{%
    \errmessage{(Inkscape) Color is used for the text in Inkscape, but the package 'color.sty' is not loaded}%
    \renewcommand\color[2][]{}%
  }%
  \providecommand\transparent[1]{%
    \errmessage{(Inkscape) Transparency is used (non-zero) for the text in Inkscape, but the package 'transparent.sty' is not loaded}%
    \renewcommand\transparent[1]{}%
  }%
  \providecommand\rotatebox[2]{#2}%
  \ifx\svgwidth\undefined%
    \setlength{\unitlength}{516.31425512bp}%
    \ifx\svgscale\undefined%
      \relax%
    \else%
      \setlength{\unitlength}{\unitlength * \real{\svgscale}}%
    \fi%
  \else%
    \setlength{\unitlength}{\svgwidth}%
  \fi%
  \global\let\svgwidth\undefined%
  \global\let\svgscale\undefined%
  \makeatother%
  \begin{picture}(1,0.90893224)%
    \put(0,0){\includegraphics[width=\unitlength,page=1]{BoundaryChart.pdf}}%
    \put(0.55811108,0.52079649){\color[rgb]{0,0,0}\makebox(0,0)[lb]{\smash{$\tilde{p}$}}}%
    \put(0.15304216,0.83732579){\color[rgb]{0,0,0}\makebox(0,0)[lb]{\smash{$\Reps$}}}%
    \put(0.12869376,0.62483057){\color[rgb]{0,0,0}\makebox(0,0)[lb]{\smash{$\lim_{s \to 0} \tilde{\varphi}\big(\tilde{\gamma}(s)\big)$}}}%
    \put(0.54040317,0.18213234){\color[rgb]{0,0,0}\makebox(0,0)[lb]{\smash{$\tilde{\varphi}\big(\tilde{\gamma}(-\mu)\big)$}}}%
    \put(0.64665077,0.81961788){\color[rgb]{0,0,0}\makebox(0,0)[lb]{\smash{$\mathrm{graph} f$}}}%
    \put(-0.00411571,0.0116935){\color[rgb]{0,0,0}\makebox(0,0)[lb]{\smash{$\Big(\tilde{\varphi}\big(\tilde{\gamma}(-\mu)\big) + C^+_{\nicefrac{5}{8}}\Big) \cap \Big(\lim_{s \to 0}\tilde{\varphi}\big(\tilde{\gamma}(s)\big) + C^-_{\nicefrac{5}{8}}\Big)$}}}%
    \put(0,0){\includegraphics[width=\unitlength,page=2]{BoundaryChart.pdf}}%
  \end{picture}%
\endgroup%

%% file: HolonomyTChi.pdf_tex
\begingroup%
  \makeatletter%
  \providecommand\color[2][]{%
    \errmessage{(Inkscape) Color is used for the text in Inkscape, but the package 'color.sty' is not loaded}%
    \renewcommand\color[2][]{}%
  }%
  \providecommand\transparent[1]{%
    \errmessage{(Inkscape) Transparency is used (non-zero) for the text in Inkscape, but the package 'transparent.sty' is not loaded}%
    \renewcommand\transparent[1]{}%
  }%
  \providecommand\rotatebox[2]{#2}%
  \ifx\svgwidth\undefined%
    \setlength{\unitlength}{617.370221bp}%
    \ifx\svgscale\undefined%
      \relax%
    \else%
      \setlength{\unitlength}{\unitlength * \real{\svgscale}}%
    \fi%
  \else%
    \setlength{\unitlength}{\svgwidth}%
  \fi%
  \global\let\svgwidth\undefined%
  \global\let\svgscale\undefined%
  \makeatother%
  \begin{picture}(1,0.61142799)%
    \put(0,0){\includegraphics[width=\unitlength,page=1]{HolonomyTChi.pdf}}%
    \put(0.10034605,0.05048879){\color[rgb]{0,0,0}\makebox(0,0)[lb]{\smash{$\hat{t}$}}}%
    \put(0.86858153,0.54475117){\color[rgb]{0,0,0}\makebox(0,0)[lb]{\smash{$\chi$}}}%
    \put(0.76861832,0.57807224){\color[rgb]{0,0,0}\makebox(0,0)[lb]{\smash{$\chi_0$}}}%
    \put(0.43911014,0.57992339){\color[rgb]{0,0,0}\makebox(0,0)[lb]{\smash{$\frac{\chi_0}{2}$}}}%
    \put(0.46132417,0.46700203){\color[rgb]{0,0,0}\makebox(0,0)[lb]{\smash{$\gamma(-\mu_0)$}}}%
    \put(0,0){\includegraphics[width=\unitlength,page=2]{HolonomyTChi.pdf}}%
    \put(0.22252325,0.11527974){\color[rgb]{0,0,0}\makebox(0,0)[lb]{\smash{$\gamma(-\mu)$}}}%
    \put(0.73714844,0.15970788){\color[rgb]{0,0,0}\makebox(0,0)[lb]{\smash{$\ell_{\kappa}$}}}%
    \put(0.59460844,0.12268443){\color[rgb]{0,0,0}\makebox(0,0)[lb]{\smash{$\underline{\ell}_{\kappa}$}}}%
    \put(0,0){\includegraphics[width=\unitlength,page=3]{HolonomyTChi.pdf}}%
    \put(0.04283723,0.46094129){\color[rgb]{0,0,0}\makebox(0,0)[lb]{\smash{$\hat{t}_{\mu_0}$}}}%
    \put(0.00026032,0.40725734){\color[rgb]{0,0,0}\makebox(0,0)[lb]{\smash{$\hat{t}_{\mathrm{right}}$}}}%
    \put(-0.03861425,0.35172224){\color[rgb]{0,0,0}\makebox(0,0)[lb]{\smash{$\hat{t}_{\mathrm{right}, \mu_0}$}}}%
    \put(0.4590901,0.2647173){\color[rgb]{0,0,0}\makebox(0,0)[lb]{\smash{$\sigma_{\mathrm{right},\mu_0}$}}}%
    \put(0,0){\includegraphics[width=\unitlength,page=4]{HolonomyTChi.pdf}}%
    \put(0.04442808,0.1536471){\color[rgb]{0,0,0}\makebox(0,0)[lb]{\smash{$\hat{t}_{\mu}$}}}%
  \end{picture}%
\endgroup%

%% file: PhiHomotop.pdf_tex
\begingroup%
  \makeatletter%
  \providecommand\color[2][]{%
    \errmessage{(Inkscape) Color is used for the text in Inkscape, but the package 'color.sty' is not loaded}%
    \renewcommand\color[2][]{}%
  }%
  \providecommand\transparent[1]{%
    \errmessage{(Inkscape) Transparency is used (non-zero) for the text in Inkscape, but the package 'transparent.sty' is not loaded}%
    \renewcommand\transparent[1]{}%
  }%
  \providecommand\rotatebox[2]{#2}%
  \ifx\svgwidth\undefined%
    \setlength{\unitlength}{627.58196973bp}%
    \ifx\svgscale\undefined%
      \relax%
    \else%
      \setlength{\unitlength}{\unitlength * \real{\svgscale}}%
    \fi%
  \else%
    \setlength{\unitlength}{\svgwidth}%
  \fi%
  \global\let\svgwidth\undefined%
  \global\let\svgscale\undefined%
  \makeatother%
  \begin{picture}(1,0.65863348)%
    \put(0,0){\includegraphics[width=\unitlength,page=1]{PhiHomotop.pdf}}%
    \put(0.28433963,0.61853625){\color[rgb]{0,0,0}\makebox(0,0)[lb]{\smash{$s$}}}%
    \put(0.21696084,0.51473647){\color[rgb]{0,0,0}\makebox(0,0)[lb]{\smash{$-\mu_1$}}}%
    \put(0.21878187,0.44917875){\color[rgb]{0,0,0}\makebox(0,0)[lb]{\smash{$-\mu_0$}}}%
    \put(0.24063444,0.39818935){\color[rgb]{0,0,0}\makebox(0,0)[lb]{\smash{$\lambda$}}}%
    \put(-0.00338601,0.25614763){\color[rgb]{0,0,0}\makebox(0,0)[lb]{\smash{$(\sigma^\varphi)^{-1}\big(\gamma^\varphi(\lambda)\big)$}}}%
    \put(0.22970819,0.13231627){\color[rgb]{0,0,0}\makebox(0,0)[lb]{\smash{$-\mu$}}}%
    \put(0.33168688,0.01030601){\color[rgb]{0,0,0}\makebox(0,0)[lb]{\smash{$\gamma^\varphi(-\mu)$}}}%
    \put(0.63944408,-0.04978854){\color[rgb]{0,0,0}\makebox(0,0)[lb]{\smash{}}}%
    \put(0.69953871,0.01212709){\color[rgb]{0,0,0}\makebox(0,0)[lb]{\smash{$\pi$}}}%
    \put(0.93445391,0.0412639){\color[rgb]{0,0,0}\makebox(0,0)[lb]{\smash{$\varphi$}}}%
    \put(0,0){\includegraphics[width=\unitlength,page=2]{PhiHomotop.pdf}}%
    \put(0.49376018,0.00848494){\color[rgb]{0,0,0}\makebox(0,0)[lb]{\smash{$\gamma^\varphi(\lambda)$}}}%
    \put(0.66493876,0.29074755){\color[rgb]{0,0,0}\makebox(0,0)[lb]{\smash{$\sigma^\varphi(s)$}}}%
    \put(0.49922332,0.62764151){\color[rgb]{0,0,0}\makebox(0,0)[lb]{\smash{$\gamma^\varphi(s)$}}}%
    \put(0,0){\includegraphics[width=\unitlength,page=3]{PhiHomotop.pdf}}%
  \end{picture}%
\endgroup%

%% file: HomotopyNull.pdf_tex
\begingroup%
  \makeatletter%
  \providecommand\color[2][]{%
    \errmessage{(Inkscape) Color is used for the text in Inkscape, but the package 'color.sty' is not loaded}%
    \renewcommand\color[2][]{}%
  }%
  \providecommand\transparent[1]{%
    \errmessage{(Inkscape) Transparency is used (non-zero) for the text in Inkscape, but the package 'transparent.sty' is not loaded}%
    \renewcommand\transparent[1]{}%
  }%
  \providecommand\rotatebox[2]{#2}%
  \ifx\svgwidth\undefined%
    \setlength{\unitlength}{520.41699988bp}%
    \ifx\svgscale\undefined%
      \relax%
    \else%
      \setlength{\unitlength}{\unitlength * \real{\svgscale}}%
    \fi%
  \else%
    \setlength{\unitlength}{\svgwidth}%
  \fi%
  \global\let\svgwidth\undefined%
  \global\let\svgscale\undefined%
  \makeatother%
  \begin{picture}(1,0.81841646)%
    \put(0,0){\includegraphics[width=\unitlength,page=1]{HomotopyNull.pdf}}%
    \put(0.06399401,0.64708408){\color[rgb]{0,0,0}\makebox(0,0)[lb]{\smash{$\sigma_0(\lambda)$}}}%
    \put(0,0){\includegraphics[width=\unitlength,page=2]{HomotopyNull.pdf}}%
    \put(0.08375837,0.78104259){\color[rgb]{0,0,0}\makebox(0,0)[lb]{\smash{$\sigma_0(-\mu_1)$}}}%
    \put(0.47245764,0.44724434){\color[rgb]{0,0,0}\makebox(0,0)[lb]{\smash{$\sigma_0(-\mu_0)$}}}%
    \put(0.4285368,0.55924244){\color[rgb]{0,0,0}\makebox(0,0)[lb]{\smash{$\sigma_0$}}}%
    \put(0,0){\includegraphics[width=\unitlength,page=3]{HomotopyNull.pdf}}%
    \put(0.75794305,0.5987712){\color[rgb]{0,0,0}\makebox(0,0)[lb]{\smash{$\tau$}}}%
    \put(0,0){\includegraphics[width=\unitlength,page=4]{HomotopyNull.pdf}}%
    \put(0.68766972,0.76786637){\color[rgb]{0,0,0}\makebox(0,0)[lb]{\smash{$\Gamma_{\mu_1}(\cdot \,;-\mu_1)$}}}%
    \put(0.72939445,0.68222075){\color[rgb]{0,0,0}\makebox(0,0)[lb]{\smash{$\Gamma_{\mu_1}(\cdot \, ; \lambda)$}}}%
    \put(0,0){\includegraphics[width=\unitlength,page=5]{HomotopyNull.pdf}}%
    \put(0.19136439,0.16615102){\color[rgb]{0,0,0}\makebox(0,0)[lb]{\smash{$u$}}}%
    \put(0.67010136,0.2034837){\color[rgb]{0,0,0}\makebox(0,0)[lb]{\smash{$v$}}}%
    \put(0.44390913,0.00584003){\color[rgb]{0,0,0}\makebox(0,0)[lb]{\smash{$-\infty$}}}%
    \put(0.31873476,0.01023216){\color[rgb]{0,0,0}\makebox(0,0)[lb]{\smash{$-\infty$}}}%
    \put(0.00689694,0.34403036){\color[rgb]{0,0,0}\makebox(0,0)[lb]{\smash{$0$}}}%
    \put(0.85456882,0.39453936){\color[rgb]{0,0,0}\makebox(0,0)[lb]{\smash{$0$}}}%
  \end{picture}%
\endgroup%

%% file: LoopsNull.pdf_tex
\begingroup%
  \makeatletter%
  \providecommand\color[2][]{%
    \errmessage{(Inkscape) Color is used for the text in Inkscape, but the package 'color.sty' is not loaded}%
    \renewcommand\color[2][]{}%
  }%
  \providecommand\transparent[1]{%
    \errmessage{(Inkscape) Transparency is used (non-zero) for the text in Inkscape, but the package 'transparent.sty' is not loaded}%
    \renewcommand\transparent[1]{}%
  }%
  \providecommand\rotatebox[2]{#2}%
  \ifx\svgwidth\undefined%
    \setlength{\unitlength}{472.35886461bp}%
    \ifx\svgscale\undefined%
      \relax%
    \else%
      \setlength{\unitlength}{\unitlength * \real{\svgscale}}%
    \fi%
  \else%
    \setlength{\unitlength}{\svgwidth}%
  \fi%
  \global\let\svgwidth\undefined%
  \global\let\svgscale\undefined%
  \makeatother%
  \begin{picture}(1,0.73189095)%
    \put(0,0){\includegraphics[width=\unitlength,page=1]{LoopsNull.pdf}}%
    \put(0.19631717,0.22779481){\color[rgb]{0,0,0}\makebox(0,0)[lb]{\smash{$\tau$}}}%
    \put(0.2785791,0.08988522){\color[rgb]{0,0,0}\makebox(0,0)[lb]{\smash{$A$}}}%
    \put(-0.00449865,0.12133827){\color[rgb]{0,0,0}\makebox(0,0)[lb]{\smash{$B$}}}%
    \put(0.13341102,0.3705435){\color[rgb]{0,0,0}\makebox(0,0)[lb]{\smash{$C(v_1)$}}}%
    \put(0.49391171,0.30763729){\color[rgb]{0,0,0}\makebox(0,0)[lb]{\smash{$D(v_1)$}}}%
    \put(0.80360353,0.19392226){\color[rgb]{0,0,0}\makebox(0,0)[lb]{\smash{$\{v = 0\}$}}}%
    \put(0.29793484,0.62942652){\color[rgb]{0,0,0}\makebox(0,0)[lb]{\smash{$\partial_u$}}}%
    \put(0.41406924,0.17940548){\color[rgb]{0,0,0}\makebox(0,0)[lb]{\smash{$\partial_v$}}}%
    \put(0.49633116,0.58587608){\color[rgb]{0,0,0}\makebox(0,0)[lb]{\smash{$\partial_\varphi$}}}%
    \put(0,0){\includegraphics[width=\unitlength,page=2]{LoopsNull.pdf}}%
    \put(0.12857019,0.01278165){\color[rgb]{0,0,0}\makebox(0,0)[lb]{\smash{$\partial_\varphi$}}}%
  \end{picture}%
\endgroup%

%% file: PenroseRNV.pdf_tex
\begingroup%
  \makeatletter%
  \providecommand\color[2][]{%
    \errmessage{(Inkscape) Color is used for the text in Inkscape, but the package 'color.sty' is not loaded}%
    \renewcommand\color[2][]{}%
  }%
  \providecommand\transparent[1]{%
    \errmessage{(Inkscape) Transparency is used (non-zero) for the text in Inkscape, but the package 'transparent.sty' is not loaded}%
    \renewcommand\transparent[1]{}%
  }%
  \providecommand\rotatebox[2]{#2}%
  \ifx\svgwidth\undefined%
    \setlength{\unitlength}{537.02664112bp}%
    \ifx\svgscale\undefined%
      \relax%
    \else%
      \setlength{\unitlength}{\unitlength * \real{\svgscale}}%
    \fi%
  \else%
    \setlength{\unitlength}{\svgwidth}%
  \fi%
  \global\let\svgwidth\undefined%
  \global\let\svgscale\undefined%
  \makeatother%
  \begin{picture}(1,0.92173127)%
    \put(0,0){\includegraphics[width=\unitlength,page=1]{PenroseRNV.pdf}}%
    \put(0.71960384,0.32253782){\color[rgb]{0,0,0}\makebox(0,0)[lb]{\smash{$\mathcal{I}^+$}}}%
    \put(0.62596656,0.42894384){\color[rgb]{0,0,0}\makebox(0,0)[lb]{\smash{$i^+$}}}%
    \put(0.45784507,0.59280913){\color[rgb]{0,0,0}\makebox(0,0)[lb]{\smash{$CH^+$}}}%
    \put(0.4940231,0.2629505){\color[rgb]{0,0,0}\makebox(0,0)[lb]{\smash{$\mathcal{H}^+$}}}%
    \put(0.85367541,0.1863382){\color[rgb]{0,0,0}\makebox(0,0)[lb]{\smash{$i^0$}}}%
    \put(0.73237254,0.05439469){\color[rgb]{0,0,0}\makebox(0,0)[lb]{\smash{$\mathcal{I}^-$}}}%
    \put(0.2875955,0.18846626){\color[rgb]{0,0,0}\makebox(0,0)[lb]{\smash{$u = u_{\mathcal{H}^+}$}}}%
    \put(0.3745491,0.04043427){\color[rgb]{0,0,0}\makebox(0,0)[lb]{\smash{$v = v_0$}}}%
    \put(0,0){\includegraphics[width=\unitlength,page=2]{PenroseRNV.pdf}}%
    \put(0.0045555,0.4821469){\color[rgb]{0,0,0}\makebox(0,0)[lb]{\smash{$u = u_T$}}}%
    \put(0.31738912,0.50342804){\color[rgb]{0,0,0}\makebox(0,0)[lb]{\smash{$\partial_vr <0$}}}%
    \put(0.00029927,0.69921509){\color[rgb]{0,0,0}\makebox(0,0)[lb]{\smash{$r =0$}}}%
    \put(0.00455551,0.26933485){\color[rgb]{0,0,0}\makebox(0,0)[lb]{\smash{$\textnormal{apparent horizons}$}}}%
    \put(0,0){\includegraphics[width=\unitlength,page=3]{PenroseRNV.pdf}}%
  \end{picture}%
\endgroup%

%% file: RV.pdf_tex
\begingroup%
  \makeatletter%
  \providecommand\color[2][]{%
    \errmessage{(Inkscape) Color is used for the text in Inkscape, but the package 'color.sty' is not loaded}%
    \renewcommand\color[2][]{}%
  }%
  \providecommand\transparent[1]{%
    \errmessage{(Inkscape) Transparency is used (non-zero) for the text in Inkscape, but the package 'transparent.sty' is not loaded}%
    \renewcommand\transparent[1]{}%
  }%
  \providecommand\rotatebox[2]{#2}%
  \ifx\svgwidth\undefined%
    \setlength{\unitlength}{434.34465543bp}%
    \ifx\svgscale\undefined%
      \relax%
    \else%
      \setlength{\unitlength}{\unitlength * \real{\svgscale}}%
    \fi%
  \else%
    \setlength{\unitlength}{\svgwidth}%
  \fi%
  \global\let\svgwidth\undefined%
  \global\let\svgscale\undefined%
  \makeatother%
  \begin{picture}(1,0.68754231)%
    \put(0,0){\includegraphics[width=\unitlength,page=1]{RV.pdf}}%
    \put(-0.00489243,0.49541384){\color[rgb]{0,0,0}\makebox(0,0)[lb]{\smash{$r_+(\infty)$}}}%
    \put(-0.00226119,0.17703599){\color[rgb]{0,0,0}\makebox(0,0)[lb]{\smash{$r_-(\infty)$}}}%
    \put(0.11088132,0.64276221){\color[rgb]{0,0,0}\makebox(0,0)[lb]{\smash{$r$}}}%
    \put(0.89498542,0.02442518){\color[rgb]{0,0,0}\makebox(0,0)[lb]{\smash{$v$}}}%
    \put(0.33979763,0.01390027){\color[rgb]{0,0,0}\makebox(0,0)[lb]{\smash{$v_0$}}}%
    \put(0,0){\includegraphics[width=\unitlength,page=2]{RV.pdf}}%
    \put(0.60555102,0.42700203){\color[rgb]{0,0,0}\makebox(0,0)[lb]{\smash{$r_+(v)$}}}%
    \put(0.65028183,0.24807895){\color[rgb]{0,0,0}\makebox(0,0)[lb]{\smash{$r_-(v)$}}}%
  \end{picture}%
\endgroup%

%% file: DLO.pdf_tex
\begingroup%
  \makeatletter%
  \providecommand\color[2][]{%
    \errmessage{(Inkscape) Color is used for the text in Inkscape, but the package 'color.sty' is not loaded}%
    \renewcommand\color[2][]{}%
  }%
  \providecommand\transparent[1]{%
    \errmessage{(Inkscape) Transparency is used (non-zero) for the text in Inkscape, but the package 'transparent.sty' is not loaded}%
    \renewcommand\transparent[1]{}%
  }%
  \providecommand\rotatebox[2]{#2}%
  \ifx\svgwidth\undefined%
    \setlength{\unitlength}{474.05517098bp}%
    \ifx\svgscale\undefined%
      \relax%
    \else%
      \setlength{\unitlength}{\unitlength * \real{\svgscale}}%
    \fi%
  \else%
    \setlength{\unitlength}{\svgwidth}%
  \fi%
  \global\let\svgwidth\undefined%
  \global\let\svgscale\undefined%
  \makeatother%
  \begin{picture}(1,0.45510579)%
    \put(0,0){\includegraphics[width=\unitlength,page=1]{DLO.pdf}}%
    \put(0.91885784,0.01273587){\color[rgb]{0,0,0}\makebox(0,0)[lb]{\smash{$i^0$}}}%
    \put(0.7091173,0.25622771){\color[rgb]{0,0,0}\makebox(0,0)[lb]{\smash{$i^+$}}}%
    \put(0.21731197,0.25381692){\color[rgb]{0,0,0}\makebox(0,0)[lb]{\smash{$i^+$}}}%
    \put(-0.0044826,0.01514671){\color[rgb]{0,0,0}\makebox(0,0)[lb]{\smash{$i^0$}}}%
    \put(0.09677145,0.12845474){\color[rgb]{0,0,0}\makebox(0,0)[lb]{\smash{$\mathcal{I}^+$}}}%
    \put(0.82242528,0.13086558){\color[rgb]{0,0,0}\makebox(0,0)[lb]{\smash{$\mathcal{I}^+$}}}%
    \put(0.28722548,0.35989253){\color[rgb]{0,0,0}\makebox(0,0)[lb]{\smash{$CH^+$}}}%
    \put(0.56205782,0.38158989){\color[rgb]{0,0,0}\makebox(0,0)[lb]{\smash{$CH^+$}}}%
    \put(0.34267408,0.17184939){\color[rgb]{0,0,0}\makebox(0,0)[lb]{\smash{$\mathcal{H}^+$}}}%
    \put(0.5789335,0.1935467){\color[rgb]{0,0,0}\makebox(0,0)[lb]{\smash{$\mathcal{H}^+$}}}%
  \end{picture}%
\endgroup%

%% file: Geod.pdf_tex
\begingroup%
  \makeatletter%
  \providecommand\color[2][]{%
    \errmessage{(Inkscape) Color is used for the text in Inkscape, but the package 'color.sty' is not loaded}%
    \renewcommand\color[2][]{}%
  }%
  \providecommand\transparent[1]{%
    \errmessage{(Inkscape) Transparency is used (non-zero) for the text in Inkscape, but the package 'transparent.sty' is not loaded}%
    \renewcommand\transparent[1]{}%
  }%
  \providecommand\rotatebox[2]{#2}%
  \ifx\svgwidth\undefined%
    \setlength{\unitlength}{381.05179434bp}%
    \ifx\svgscale\undefined%
      \relax%
    \else%
      \setlength{\unitlength}{\unitlength * \real{\svgscale}}%
    \fi%
  \else%
    \setlength{\unitlength}{\svgwidth}%
  \fi%
  \global\let\svgwidth\undefined%
  \global\let\svgscale\undefined%
  \makeatother%
  \begin{picture}(1,0.84971109)%
    \put(0,0){\includegraphics[width=\unitlength,page=1]{Geod.pdf}}%
    \put(0.87619322,0.00987399){\color[rgb]{0,0,0}\makebox(0,0)[lb]{\smash{$v$}}}%
    \put(-0.00557667,0.79866816){\color[rgb]{0,0,0}\makebox(0,0)[lb]{\smash{$r$}}}%
    \put(0,0){\includegraphics[width=\unitlength,page=2]{Geod.pdf}}%
    \put(0.26135366,0.62171436){\color[rgb]{0,0,0}\makebox(0,0)[lb]{\smash{$M_{\mathrm{ext}}$}}}%
    \put(0.2433584,0.16283415){\color[rgb]{0,0,0}\makebox(0,0)[lb]{\smash{$M_{\mathrm{int}}$}}}%
    \put(0.42331139,0.31279492){\color[rgb]{0,0,0}\makebox(0,0)[lb]{\smash{$\sigma_{\Hp}$}}}%
  \end{picture}%
\endgroup%

%% file: 21_07_09_RevisedVersionForDMJ.bbl
\begin{thebibliography}{10}

\bibitem{BonVai70}
{\sc Bonnor, W., and Vaidya, P.}
\newblock {Spherically Symmetric Radiation of Charge in Einstein-Maxwell
  Theory}.
\newblock {\em Gen. Relativ. Gravit. 1}, 2 (1970), 127--130.

\bibitem{Chris09}
{\sc Christodoulou, D.}
\newblock {\em {The formation of black holes in general relativity}}.
\newblock European Mathematical Society, 2009.

\bibitem{Chrus91}
{\sc Chru\'sciel, P.}
\newblock {On Uniqueness in the Large of Solutions of Einstein's Equations}.
\newblock {\em Proceedings of the CMA 27\/} (1991).

\bibitem{ChrusGra12}
{\sc Chru\'sciel, P., and Grant, J.}
\newblock {On Lorentzian causality with continuous metrics}.
\newblock {\em Class. Quantum Grav. 29\/} (2012).

\bibitem{ChrusKli12}
{\sc Chru\'sciel, P., and Klinger, P.}
\newblock {The annoying null boundaries}.
\newblock {\em J. Phys. Conf. Ser. 968\/} (2018).

\bibitem{CoGiNaDru17}
{\sc Costa, J.~L., Gir{\~a}o, P.~M., Nat{\'a}rio, J., and Silva, J.~D.}
\newblock {On the Global Uniqueness for the Einstein--Maxwell-Scalar Field
  System with a Cosmological Constant: Part 3. Mass Inflation and Extendibility
  of the Solutions}.
\newblock {\em Annals of PDE 3}, 1 (2017), 8.

\bibitem{Daf03}
{\sc Dafermos, M.}
\newblock {Stability and instability of the Cauchy horizon for the spherically
  symmetric Einstein-Maxwell-scalar field equations}.
\newblock {\em Ann. of Math. 158\/} (2003), 875--928.

\bibitem{Daf05a}
{\sc Dafermos, M.}
\newblock {The interior of charged black holes and the problem of uniqueness in
  general relativity}.
\newblock {\em Comm. Pure Appl. Math. 58\/} (2005), 445--504.

\bibitem{Daf14}
{\sc Dafermos, M.}
\newblock Black holes without spacelike singularities.
\newblock {\em Comm. Math. Phys. 332\/} (2014), 729--757.

\bibitem{DafLuk17}
{\sc Dafermos, M., and Luk, J.}
\newblock {The interior of dynamical vacuum black holes I: The $C^0$-stability
  of the Kerr Cauchy horizon}.
\newblock {\em arXiv:1710.01772\/} (2017).

\bibitem{DafRod05}
{\sc Dafermos, M., and Rodnianski, I.}
\newblock {A proof of Price's law for the collapse of a self-gravitating scalar
  field}.
\newblock {\em Invent. math. 162\/} (2005), 381--457.

\bibitem{FouSbi20}
{\sc Fournodavlos, G., and Sbierski, J.}
\newblock {Generic Blow-Up Results for the Wave Equation in the Interior of a
  Schwarzschild Black Hole}.
\newblock {\em Arch. Rational Mech. Anal. 235\/} (2020), 927--971.

\bibitem{GalLin16}
{\sc Galloway, G., and Ling, E.}
\newblock {Some Remarks on the $C^0$-(in)extendibility of Spacetimes}.
\newblock {\em Ann. Henri Poincar\'e 18}, 10 (2017), 3427--3477.

\bibitem{GalLinSbi17}
{\sc Galloway, G., Ling, E., and Sbierski, J.}
\newblock {Timelike completeness as an obstruction to $C^0$-extensions}.
\newblock {\em Comm. Math. Phys. 359}, 3 (2018), 937--949.

\bibitem{Gar99}
{\sc Garfinkle, D.}
\newblock {Metrics with distributional curvature}.
\newblock {\em Class. Quantum Grav. 16\/} (1999), 4101--4109.

\bibitem{GerTra87}
{\sc Geroch, R., and Traschen, J.}
\newblock {Strings and other distributional sources in general relativity}.
\newblock {\em Phys. Rev. D 36}, 4 (1987), 1017--1031.

\bibitem{GrafLing18}
{\sc Graf, M., and Ling, E.}
\newblock {Maximizers in Lipschitz spacetimes are either timelike or null}.
\newblock {\em Class. Quantum Grav. 35}, 8 (2018).

\bibitem{GraKuSa19}
{\sc Grant, J., Kunzinger, M., and S\"amann, C.}
\newblock {Inextendibility of spacetimes and Lorentzian length spaces }.
\newblock {\em Ann. Glob. Anal. Geom. 55\/} (2019), 133--147.

\bibitem{GraKuSaSt19}
{\sc Grant, J., Kunzinger, M., S\"amann, C., and Steinbauer, R.}
\newblock {The future is not always open}.
\newblock {\em Lett. Math. Phys. 110\/} (2020), 83--103.

\bibitem{HawkEllis}
{\sc Hawking, S., and Ellis, G.}
\newblock {\em The large scale structure of space-time}.
\newblock Cambridge University Press, 1973.

\bibitem{HerHis92}
{\sc Herman, R., and Hiscock, W.}
\newblock {Strength of the mass inflation singularity}.
\newblock {\em Phys. Rev. D 46}, 4 (1992), 1863--1865.

\bibitem{Hirsch12}
{\sc Hirsch, M.}
\newblock {\em {Differential Topology}}.
\newblock Springer, 2012.

\bibitem{His81}
{\sc Hiscock, W.}
\newblock {Evolution of the interior of a charged black hole}.
\newblock {\em Phys. Rev. Lett. 83A\/} (1981), 110--112.

\bibitem{LeeRiem}
{\sc Lee, J.}
\newblock {\em {Riemannian Manifolds: An Introduction to Curvature}}.
\newblock Springer, 1997.

\bibitem{Ling20}
{\sc Ling, E.}
\newblock {The Big Bang is a Coordinate Singularity for $k = -1$ inflationary
  FLRW Spacetimes}.
\newblock {\em Foundations of Physics 50\/} (2020), 385--428.

\bibitem{LukOh19I}
{\sc Luk, J., and Oh, S.-J.}
\newblock {Strong cosmic censorship in spherical symmetry for two-ended
  asymptotically flat initial data I. The interior of the black hole region}.
\newblock {\em Annals of Math. 190}, 1 (2019), 1--111.

\bibitem{LukOh19II}
{\sc Luk, J., and Oh, S.-J.}
\newblock {Strong cosmic censorship in spherical symmetry for two-ended
  asymptotically flat initial data II. The exterior of the black hole region}.
\newblock {\em Annals of PDE 5}, 6 (2019).

\bibitem{MinSuhr19}
{\sc Minguzzi, E., and Suhr, S.}
\newblock {Some regularity results for Lorentz-Finsler spaces}.
\newblock {\em Ann. Glob. Anal. Geom. 56\/} (2019), 597--611.

\bibitem{MTW}
{\sc Misner, C., Thorne, K., and Wheeler, J.}
\newblock {\em Gravitation}.
\newblock W. H. Freeman and Co., 1973.

\bibitem{Ori91}
{\sc Ori, A.}
\newblock {Inner Structure of a Charged Black Hole: An Exact Mass-Inflation
  Solution}.
\newblock {\em Phys. Rev. Lett. 67\/} (1991).

\bibitem{Ori00}
{\sc Ori, A.}
\newblock {Strength of curvature singularities}.
\newblock {\em Phys. Rev. D 61\/} (2000).

\bibitem{Pen72}
{\sc Penrose, R.}
\newblock {\em {`The geometry of impulsive gravitational waves' \emph{in
  General relativity (papers in honour of J. L. Synge)}}}.
\newblock Clarendon Press, Oxford, 1972.

\bibitem{PoiIs89}
{\sc Poisson, E., and Israel, W.}
\newblock {Inner-Horizon Instability and Mass Inflation in Black Holes}.
\newblock {\em Phys. Rev. Lett. 63}, 16 (1989), 1663--1666.

\bibitem{PoiIs90}
{\sc Poisson, E., and Israel, W.}
\newblock {Internal structure of black holes}.
\newblock {\em Phys. Rev. D 41\/} (1990), 1796--1809.

\bibitem{Sbie13b}
{\sc Sbierski, J.}
\newblock {Characterisation of the Energy of Gaussian Beams on Lorentzian
  Manifolds: with Applications to Black Hole Spacetimes}.
\newblock {\em {Analysis \& PDE} 8\/} (2015), 1379--1420.

\bibitem{Sbie18}
{\sc Sbierski, J.}
\newblock {On the proof of the $C^0$-inextendibility of the Schwarzschild
  spacetime}.
\newblock {\em J. Phys. Conf. Ser. 968\/} (2018).

\bibitem{Sbie15}
{\sc Sbierski, J.}
\newblock {The $C^0$-inextendibility of the Schwarzschild spacetime and the
  spacelike diameter in Lorentzian geometry}.
\newblock {\em J. Diff. Geom. 108}, 2 (2018), 319--378.

\bibitem{Tip77}
{\sc Tipler, F.}
\newblock {Singularities in conformally flat spacetimes}.
\newblock {\em Phys. Lett. 64A}, 1 (1977), 8--10.

\bibitem{VdM18}
{\sc Van~de Moortel, M.}
\newblock {Stability and Instability of the Sub-extremal
  Reissner--Nordstr{\"o}m Black Hole Interior for the
  Einstein--Maxwell--Klein--Gordon Equations in Spherical Symmetry}.
\newblock {\em Commun. Math. Phys. 360}, 1 (2018), 103--168.

\end{thebibliography}
